\documentclass[10pt]{article}
\usepackage[margin=.8in,left=.8in]{geometry}

\usepackage{amsthm}
\usepackage{amssymb}
\usepackage{amsmath}
\usepackage{mathtools}
\usepackage{adjustbox}

\usepackage[table]{xcolor}

\usepackage{stmaryrd}    
\usepackage{xfrac}

\usepackage{tensor}

\usepackage{enumitem}\setlist{nosep}

\usepackage{slashed}

\usepackage[safe]{tipa} 

\usepackage{hhline}

\usepackage{tikz}
\usetikzlibrary{
  backgrounds, 
  decorations.pathmorphing, 
  decorations.markings,
  decorations.text,
  snakes
}
\usetikzlibrary{cd}

\usepackage{mathrsfs} 
\DeclareMathAlphabet{\mathpzc}{OT1}{pzc}{m}{it} 

\usepackage{hyperref}
\hypersetup{
  colorlinks = true,
  linkcolor= darkblue, citecolor=darkblue            
}

\hypersetup{
        colorlinks=true,
        linkcolor=darkblue,
        filecolor=darkblue,      
        urlcolor=black,
        citecolor=darkblue,
        linktoc=page
    }

\usepackage{cleveref}

\newcommand{\Porth}{P_{\hspace{-2pt}\raisebox{1.2pt}{\scalebox{.43}{$\perp$}}}}

\definecolor{darkblue}{rgb}{0.05,0.25,0.65}
\definecolor{darkgreen}{RGB}{20,140,10}
\definecolor{lightgray}{rgb}{0.9,0.9,0.9}
\definecolor{darkorange}{RGB}{200,100,5}
\definecolor{darkyellow}{rgb}{.91,.91,0}
\definecolor{orangeii}{RGB}{200,100,5}
\definecolor{lightblue}{RGB}{243, 250, 255}

\newtheorem{theorem}{Theorem}[section]

\newtheorem{lemma}[theorem]{Lemma}

\theoremstyle{definition}
\newtheorem{definition}[theorem]{Definition}

\newtheorem{example}[theorem]{Example}

\newtheorem{remark}[theorem]{Remark}

\usepackage{accents}
\newlength{\dhatheight}

\newcommand{\SecondFundamentalForm}{\mbox{\rm I\hspace{-1.25pt}I}}

\newcommand{\CoverOf}[1]{\widetilde{#1}}

\let\PLAINthebibliography\thebibliography
\renewcommand\thebibliography[1]{
  \PLAINthebibliography{#1}
  \setlength{\parskip}{0.5pt}
  \setlength{\itemsep}{0.5pt plus .3ex}
}

\newcommand{\proofstep}[1]{\scalebox{.85}{#1}}

\newcommand{\evencoordinateindex}{r}
\newcommand{\oddcoordinateindex}{\rho}

\newcommand{\ZTwo}{\mathbb{Z}_2}

\newcommand{\defneq}{\equiv}

\newcommand{\differential}{\mathrm{d}}

\newcommand\bos[1]{\mathstrut\mkern2.5mu#1\mkern-14mu\raise1.7ex%
  \hbox{$\scriptstyle\rightsquigarrow$}}

\newcommand\bosonic[1]{\mathstrut\mkern2.5mu#1\mkern-14mu\raise1.7ex%
  \hbox{$\scriptstyle\rightsquigarrow$}}

\usepackage[cmtip,all]{xy}
\newcommand{\longsquiggly}{\xymatrix{{}\ar@{~>}[r]&{}}}

\usepackage{bbm} 

\usepackage[new]{old-arrows}   

\newcommand{\fixed}[1]{\mathcolor{black}{#1}}
\newcommand{\fixedtext}[1]{{\color{black}{#1}}}

\begin{document}

\setlength{\abovedisplayskip}{3pt}
\setlength{\belowdisplayskip}{3pt}
\setlength{\abovedisplayshortskip}{-4pt}
\setlength{\belowdisplayshortskip}{2pt}

\title{Holographic M-Brane Super-Embeddings}


\author{
  Grigorios Giotopoulos${}^{\ast}$,
  \;\;
  Hisham Sati${}^{\ast \dagger}$,
  \;\;
  Urs Schreiber${}^{\ast}$
}

\maketitle

\begin{abstract}
 Over a decade before the modern formulation of AdS/CFT duality,
 Duff et al. had observed a candidate microscopic explanation by identifying the CFT fields with fluctuations of probe $p$-branes stretched out in parallel near the horizon of their own black brane incarnation. A profound way to characterize these and more general probe $p$-brane configurations, especially for M5-branes, is expected to be as ``super-embeddings'' of their super-worldvolumes into target super-spacetime --- but no concrete example of these had appeared in the literature. 
 Here we fill this gap by constructing the explicit holographic $\sfrac{1}{2}$-BPS super-embeddings of probe M5-branes and M2-branes into their corresponding super-AdS backgrounds.
\end{abstract}

\vspace{.8cm}

\begin{center}
\begin{minipage}{10cm}
  \tableofcontents
\end{minipage}
\end{center}

\medskip

\vfill

\hrule
\vspace{5pt}

{
\footnotesize
\noindent
\def\arraystretch{1}
\tabcolsep=0pt
\begin{tabular}{ll}
${}^*$\,
&
Mathematics, Division of Science; and
\\
&
Center for Quantum and Topological Systems,
\\
&
NYUAD Research Institute,
\\
&
New York University Abu Dhabi, UAE.  
\end{tabular}
\hfill
\adjustbox{raise=-15pt}{
\includegraphics[width=3cm]{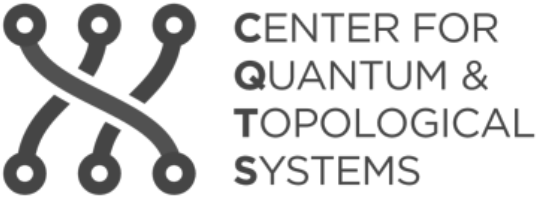}
}

\vspace{1mm} 
\noindent ${}^\dagger$The Courant Institute for Mathematical Sciences, NYU, NY.

\vspace{.2cm}

\noindent
The authors acknowledge the support by {\it Tamkeen} under the 
{\it NYU Abu Dhabi Research Institute grant} {\tt CG008}.
}

\newpage

\section{Introduction}
\label{IntroductionAndOverview}

\noindent
{\bf Microscopic holography via probe $p$-branes.}\label{MicroscopicHolography}
While holographic duality has become common-place (review includes \cite{AGMOO00}\cite{Natsuume15})
it may be less widely appreciated that well before its modern formulation a candidate microscopic explanation had been found by Duff et al., first discussed for the M2-brane \cite{BDPS87}\cite{BlencoweDuff88}\cite{DFFFTT99} then generalized to include also M5-branes and D-branes \cite{CKvP98}\cite{CKKTvP98}\cite{PST99}\cite{GM00}\cite{NurmagambetovPark02}, reviewed in \cite{Duff99-AdSa}\cite{Duff99-AdSb} (more recent variations include \cite{DGTZ20}\cite{Gupta21}\cite{Gupta24}):

\smallskip 
In this {\it microscopic $p$-brane holography} -- as we shall call it here for lack of an established name --  one considers (as indicated in \hyperlink{FigureBraneConfiguration}{Figure B}) probe $p$-branes (i.e., {\it light} branes described by sigma-models not back-reacting onto the ambient spacetime, cf. \cite{Simon12}) embedded in parallel near the (asymptotically $\mathrm{AdS}$) horizon of their own black-brane incarnation (their {\it heavy} back-reacted version described by singular solutions of supergravity, cf. \cite{DuffLu94}\cite[\S 5]{Duff99-MTheory}) and finds that their fluctuations about this configuration are described by the conformal field theory (CFT) known from AdS/CFT duality.

\smallskip 

In this picture, the otherwise somewhat mysterious holographic duality between 
(i) quantum systems 
and 
(ii) gravity 
reflects but two perspectives on the expected nature of branes: 

(i) as dynamical (fluctuating) physical objects in themselves, and 

(ii) as sources of gravitational (and higher gauge-) fields propagating away from the black brane. \footnotemark

\vspace{-.2cm}
\hspace{-.8cm}
\begin{tabular}{p{9.1cm}}
  \hypertarget{FigureBraneConfiguration}{}
  \footnotesize
  {\bf Figure B.}
  Schematics of a probe brane worldvolume immersed (embedded) near the horizon of its own black brane incarnation, parallel to it at some coordinate distance $r_{\mathrm{prb}}$. (Precise details on the black M5-brane background are in \S\ref{SuperAdS} and on the probe M5 in \S\ref{TheHolographicM5Immersion}.)

  \vspace{.3pt}
  
  The curvy line indicates (quantum-)fluctuations about this parallel configuration, thought to incarnate the strongly coupled quantum system holographically encoded in the ambient gravitational field.
\end{tabular}
\hspace{-10pt}
\adjustbox{
  raise=-1.7cm
}{
\begin{tikzpicture}[
  yscale=1.1
]

\shade[
  top color=lightgray,
  bottom color=black,
]
  (-2.54,1) rectangle (2.54,0-.05);
\shade[
  top color=black,
  bottom color=lightgray,
]
  (-2.54,0+.05) rectangle (2.54,-1);

\draw[
  orange,
  line width=1.5,
]
  (-2.54,-.6) -- (2.54,-.6);
\draw[
  orange,
  line width=1.5,
  opacity=.6,
  snake=coil, segment aspect=0,
  segment length=25
]
  (-2.54,-.6) -- (2.54,-.6);

\node
  at(0,-.54)
  {
    \scalebox{.7}{
      \color{darkyellow}
      \bf
      probe brane
    }
  };

\draw[
  draw=white,
  line width=2.5pt
]
  (-2.4,1) -- (-2.4,-1);
\draw[
  draw=white,
  line width=2.5pt
]
  (-2.2,1) -- (-2.2,-1);
\draw[
  draw=white,
  line width=2.5pt
]
  (-2,1) -- (-2,-1);

\draw[
  draw=white,
  line width=30pt
]
  (+2.6,1) -- (+2.6,-1);

\node[
]
  at (1.54,0)
  {
    \rlap{
    $
      \left.
      \def\arraystretch{2.9}
      \begin{array}{c}
        {}
        \\
        {}
      \end{array}
      \right\}
    $
    \hspace{-11pt}
    \scalebox{.7}{
      {
      \color{darkblue}
      \bf
      spacetime
      }
      \hspace{-8pt}
      $X$
    }
    }
  };

\draw[
  line width=9pt,
  draw=white,
  draw opacity=.8
]
  (2.1,-.54) -- (2.3,-.54);

\node[
]
  at (1.85,-.54)
  { 
    \rlap{
    \scalebox{.7}{
      {
      \color{darkblue}
      \bf
      immersion
      }
      \hspace{0pt}
      $\Sigma$
      }
    }
  };

\draw[
  ->,
]
  (3.77, -.4) -- (3.77, -.15);

\node
  at (0,0)
  {
    \scalebox{.7}{
      \color{magenta}
      \bf
      black brane
    }
  };

\node
  at
  (-3.5,.05)
  {
    \scalebox{.7}{
      \color{darkblue}
      \bf
    horizon
    }
  };

\draw[
  gray
]
  (-3,0) -- (-1.5,0);

\draw[
  gray,
  -Latex
]
  (-2.8,-0) -- (-2.8,-.6);

\node
  at
  (-2.8,-.25)
  {
    \scalebox{.7}{
      \colorbox{white}{
        \hspace{-5pt}$r_{\mathrm{prb}}$\hspace{-5pt}
      }
    }
  };

\draw[
  gray
]
  (-3,-.6) -- (-1.5,-.6);

\end{tikzpicture}
}

\footnotetext{
  For the case of 0-branes, namely for {\it particles}, the investigation of these dual perspectives --- (1.) as quanta and (2.) as black hole solutions --- goes back all the way to \cite{EinsteinInfeldHoffmann38}, and has fascinated authors since, see for instance \cite{ArcosPereira04}\cite{Burinskii08}.
}

\smallskip

\noindent
{\bf The open problem of holographic flux quantization.}
Even though a key aspect and motivation of holographic duality is the access it provides to non-perturbative strongly-coupled quantum physics on the brane worldvolume, existing discussions tend to ignore the main non-perturbative effect already in the classical worldvolume theory, namely the global completion of its (higher) gauge field content by flux quantization laws \cite{SS24-Flux}\cite{SS24-Phase}. Such flux-quantization provides the solitonic field content in analogy to how familiar Dirac charge quantization gives rise to Dirac monopole and Abrikosov vortex field configurations in the electromagnetic field, and is thus crucial for a complete picture of the non-perturbative physics on the brane.

\smallskip
\noindent
{\bf The need for probe brane super-embeddings.} \label{NeedForSuperEmbeddings}
But in previous articles \cite{GSS24-FluxOnM5}\cite{GSS24-SuGra}, we have explained that the issue of flux quantization especially on probe M5-branes may be solved once their worldvolume fields are promoted to ``super-embeddings'', or rather to {\it $\sfrac{1}{2}$BPS super-immersions}. Therefore the goal of the present article is to construct explicit examples of holographic M-brane super-immersion. (Based on this, we discuss the resulting worldvolume flux quantization in the companion article \cite{SS24-Companion}.)

\smallskip

In fact, what we construct here seems to be the first non-trivial example of brane ``super-embeddings'':

\vspace{1mm} 
\noindent {\bf (i)} The existing literature 
\cite{BPSTV95}\cite{HoweSezgin97a}\cite{HRS98} \cite{HoweSezgin97b}\cite{Sorokin00} (recent review in \cite{BandosSorokin23})
contains arguments that ``super-embeddings'' (i.e., $\sfrac{1}{2}$BPS super-im\-mer\-sions, \cite[Def. 2.19]{GSS24-FluxOnM5}) of super $p$-brane worldvolumes {\it imply} the equations of motion of the corresponding super $p$-brane $\sigma$-model. However, the converse conclusion --- that {\it no further} constraints than these equations of motion are implied --- is far from obvious and has only partially been addressed (e.g. for some aspects of the M2-brane in \cite[(2.50-52)]{BPSTV95}). 
Related to this may be the absence of previously published examples of non-trivial super-embeddings.

\vspace{1mm} 
\noindent {\bf (ii)}  The analogous issue in the derivation of 11d supergravity (from the superspace torsion constraint) had similarly remained unaddressed in published literature. In this case, we had settled the reverse implication with the substantial help of mechanized computer algebra \cite[Thm. 3.1]{GSS24-SuGra}. The humongous cancellations that happen to make this work seem nothing less than a miracle, quite reinforcing the idea that 11d supergravity occupies a special point in the space of all field theories.

\vspace{1mm} 
\noindent {\bf (iii)} A similar miracle may be needed to guarantee that for constructing an M5 super-immersion it is sufficient to solve its equations of motion. In lack of a complete argument to this extent, but to still have the desired implication of the super-flux Bianchi identity (\cite[Prop. 3.17]{GSS24-FluxOnM5}, needed for the flux quantization argument in \cite{SS24-Companion}), we have to construct M5 super-immersions explicitly.

\medskip 
\noindent
{\bf The construction of holographic M-brane super-embeddings.}
This is what we do here for the case of holographic M-branes (for M5-branes in \S\ref{TheHolographicM5Immersion} and for M2-branes in \S\ref{HolographicM2Immersions}).

\smallskip 
Apart from serving as a prerequisite for worldvolume {\it flux quantization} in \cite{SS25-Seifert}\cite{SS24-Companion}\cite{SS25-Srni}, these solutions should be of interest in their own right as rare explicit examples of non-trivial $p$-brane super-embeddings. 

\smallskip 
The key tool we use for the construction is the explicit rheonomy equations due to \cite{Tsimpis04} for 11d supergravity fields on superspace, of which we give a detailed re-derivation in \S\ref{ExplicitRheonomy}. Moreover, we find it most useful to not use a matrix representation of the relevant worldvolume spin representations (in contrast to most existing literature) but instead to carve these out of the 11d spinor representation by suitable projection operators (in \S\ref{SpinorsOnM5Branes} and \S\ref{SpinorsOnM2Branes}), an approach that is naturally adapted to the discussion of BPS brane super-immersions.

\medskip

Our supergravity superspace notation follows \cite{GSS24-SuGra}\cite{GSS24-FluxOnM5} (close to that of \cite{CDF91}), briefly recalled in \S\ref{TensorConventionsAnd11dSpinors} in the Appendix. 

 \medskip 
 
\section{Explicit rheonomy in 11d}
\label{ExplicitRheonomy}

Here we present explicit formulas for extending solutions of 11d supergravity from ordinary spacetime to super-spacetime, in those cases where the $(\Psi^0)$-component of the gravitino field strength vanishes \eqref{VanishingGravitinoFieldStrength} -- which is the case relevant for SuGra backgrounds (cf. \cite[\S 12.6]{FreedmanVanProeyen12}).

\smallskip

This extension process (or the property that it exists) has been called {\it rheonomy} \cite[\S III.3.3]{CDF91}, alluding to the idea that the ordinary fields ``flow'' in the odd coordinate directions from the bosonic submanifold over the full supermanifold, to become super-fields.
Explicit such formulas have been claimed 
for the special case of coset-spacetimes (like $\mathrm{AdS}_{p+2} \times S^{D-p+2}$) by \cite[p. 156]{dWPPS98}\cite{Claus99} (following \cite{KRR98}\cite{ClausKallosh99}), and a derivation in full generality has been given by \cite{Tsimpis04}. 

\smallskip 
We closely follow the latter but find that the specialization \eqref{VanishingGravitinoFieldStrength}
to vanishing gravitino field strength (which still subsumes all the former examples) gives a substantial improvement in transparency and usability that may be of interest in its own right. Additionally, we provide full details in order to secure the relative prefactors in the formulas.

\smallskip

The strategy of the construction is to expand the super-fields and their structural equations in a suitable gauge on a suitable super-coordinate chart in order to obtain explicit differential equations for the flow along the odd coordinate directions. Therefore we start by considering:

\medskip 

\hspace{-.85cm}
\begin{tabular}{p{10cm}}
{\bf Coordinate-components of superfields.}
On a super-chart with coordinates $(X,\Theta)$
we have the expansion of the super-gravitational fields \eqref{LocalCartanConnection} first into their coefficients of the coordinate-differentials and then further their super-field expansion as polynomials in the odd coordinates  (with index convention as shown on the right),
\end{tabular}
\quad 
{\small 
\def\arraystretch{1.3}
\def\tabcolsep{5pt}
\begin{tabular}{c|cc}
  \hline
  & 
  {\bf Even}
  &
  {\bf Odd}
  \\
  \hline
\rowcolor{lightgray}  
{\bf Frame}
  & 
  $a \in \{0, \cdots, 10\}$
  &
  $\alpha \in \{1, \cdots, 32\}$ 
  \\
  {\bf Coord.}
  & 
  $\evencoordinateindex \in \{0, \cdots, 10\}$
  &
  $\oddcoordinateindex \in \{1, \cdots, 32\}$ 
  \\
  \hline
\end{tabular}
}

\begin{equation}
\label{CoordinateComponentsOfCoframeField}
\hspace{-6mm} 
  \def\arraystretch{1.8}
  \begin{array}{ccccc}
  E^a 
  &=:&
  E^a_{\evencoordinateindex}
  \,
  \mathrm{d} X^\evencoordinateindex
  &+&
  E^a_{\oddcoordinateindex}
  \,
  \mathrm{d} \Theta^\oddcoordinateindex
  \,
  \\
  \Psi^\alpha 
  &=:&
  \Psi^\alpha_{\evencoordinateindex}
  \,
  \mathrm{d} X^\evencoordinateindex
  &+&
  \Psi^\alpha_{\oddcoordinateindex}
  \,
  \mathrm{d} \Theta^\oddcoordinateindex
  \\
  \Omega^{a b}
  &=&
  \Omega^{a b}_{\evencoordinateindex}
  \,
  \mathrm{d} X^\evencoordinateindex
  &+&
  \Omega^{a b}_{\oddcoordinateindex}
  \,
  \mathrm{d} \Theta^\oddcoordinateindex
  \end{array}
  \hspace{.9cm}
  \def\arraystretch{1.8}
  \begin{array}{l}
  E^a_{r/\rho}
  \;=:\;
  \sum_{n=0}^{32}
  \big(
    E^{(n)}
  \big)^a_{r/\rho}
  \;=:\;
  \sum_{n=0}^{32}
  \tfrac{1}{n!}
  \,
  \Theta^{\rho_1}
  \cdots 
  \Theta^{\rho_n}
  \,
  \big(
    E^{(n)}_{\rho_1 \cdots \rho_n}
  \big)^a_{r/\rho}
  \\
  \Psi^\alpha_{r/\rho}
  \;=:\;
  \sum_{n=0}^{32}
  \big(
    \Psi^{(n)}
  \big)^\alpha_{r/\rho}
  \;=:\;
  \sum_{n=0}^{32}
  \tfrac{1}{n!}
  \,
  \Theta^{\rho_1}
  \cdots 
  \Theta^{\rho_n}
  \,
  \big(
    \Psi^{(n)}_{\rho_1 \cdots \rho_n}
  \big)^\alpha_{r/\rho}
  \\
  \Omega^{a b}_{r/\rho}
  \;=:\;
  \sum_{n=0}^{32}
  \big(
    \Omega^{(n)}
  \big)^{a b}_{r/\rho}
  \;=:\;
  \sum_{n=0}^{32}
  \tfrac{1}{n!}
  \,
  \Theta^{\rho_1}
  \cdots 
  \Theta^{\rho_n}
  \,
  \big(
    \Omega^{(n)}_{\rho_1 \cdots \rho_n}
  \big)^{a b}_{r/\rho}
  \mathrlap{\,,}
  \end{array}
\end{equation}

\vspace{1mm} 
\noindent whose coefficients are functions on the underlying bosonic manifold which are skew-symmetric in their indices:
\begin{equation}
  \label{SkewSymmetryOfCoeffficientFunctions}
  \left(\!\!
    \def\arraystretch{1.6}
    \begin{array}{l}
    E^{(n)}_{\rho_1 \cdots \rho_n}
    \\
    \Psi^{(n)}_{\rho_1 \cdots \rho_n}
    \\
    \Omega^{(n)}_{\rho_1 \cdots \rho_n}
    \end{array}
  \!\!\! \right)
  \;:\;
  \bosonic{X} 
  \xrightarrow{\quad}
  \mathfrak{iso}
  \big(
    \mathbb{R}^{1,10 \vert \mathbf{32}}
  \big),
  \hspace{1.4cm}
  \def\arraystretch{1.7}
  \begin{array}{l}
  E^{(n)}_{\rho_1 \cdots \rho_n} 
  \;=\; 
  E^{(n)}_{[\rho_1 \cdots \rho_n]}
  \\
  \Psi^{(n)}_{\rho_1 \cdots \rho_n} 
  \;=\; 
  \Psi^{(n)}_{[\rho_1 \cdots \rho_n]}
  \\
  \Omega^{(n)}_{\rho_1 \cdots \rho_n} 
  \;=\; 
  \Omega^{(n)}_{[\rho_1 \cdots \rho_n]}
  \,.
  \end{array}
\end{equation}
Notice that this implies:
\newpage 
\begin{equation}
  \label{OneMoreIndexSkewSymmetrizing}
  \def\arraystretch{1.8}
  \begin{array}{l}
    \big(
    E
      ^{(n)}
      _{[\rho' \, \rho_2 \cdots \rho_n }
    \big)^a_{\rho]}
    \;\;=\;\;
    \tfrac{1}{n+1}
    \Big(
    n
    \big(
    E
      ^{(n)}
      _{\rho' \, [\rho_2 \cdots \rho_n}
    \big)^a_{\rho]}
    -
    \big(
    E
      ^{(n)}
      _{\rho \, \rho_2 \cdots \rho_n}
    \big)^a_{\rho'}    
    \Big)
    \,.
  \end{array}
\end{equation}

\vspace{2mm} 
\noindent Also notice the $\mathbb{N} \!\times\! \ZTwo$ bi-degrees (cf. \cite[\S 2.1.1]{GSS24-SuGra}) of the $\Psi$-components,
\begin{equation}
  \label{BiDegreeOfComponentFunctions}
  \def\tabcolsep{-2pt}
  \begin{array}{cccccccc}
     &
     \Psi^\alpha
     &=&
     \Psi^\alpha_r 
     &
     \mathrm{d}X^r
     &+&
     \Psi^\alpha_\rho
     &
     \mathrm{d}\Theta^\rho
     \\
     \mbox{deg:}
     &
     \scalebox{.9}{$(1,1)$}
     &&
     \scalebox{.9}{$(0,1)$}
     &
     \scalebox{.9}{$(1,0)$}
     &&
     \scalebox{.9}{$(0,0)$}
     &
     \scalebox{.9}{$(1,1)$}
     \mathrlap{\,,}
  \end{array}
\end{equation}

\vspace{1mm} 
\noindent which implies in particular that the component functions $\Psi^\alpha_\rho$ commute with all other terms.

\medskip 
\noindent
{\bf Wess-Zumino-Tsimpis gauge.}
On these components, we may impose the following gauge conditions
(\cite[(39-42)]{Tsimpis04}, following \cite[(A.3-4)]{McArthur84}\cite[(17-18)]{AtickDhar87}):

\begin{definition}[\bf Wess-Zumino-Tsimpis gauge]\footnote{
  Recall (e.g. \cite[\S 3.4.3]{BuchbinderKuzenko95}) that the {\it Wess-Zumino gauge} on chiral superfields constrains their dependence on the super-coordinates, hence their auxiliary super-components, but not the physical fields. The suggestion to think of this, in the context of curved superspace/supergravity, as a special case of fermionic Riemann normal coordinates may be due to \cite{AtickDhar87}, and the higher component generalization \eqref{GaugeConditions} that we use is due to \cite{Tsimpis04}.
  }
\label{WZTGauge}
The {\it WZT gauge} is given by the following conditions: 
\begin{equation}
  \label{GaugeConditions}
  \hspace{-.7cm}
  \def\arraystretch{1.6}
  \begin{array}{ccl}
    \big(
      E^{(0)}
    \big)_\rho^a 
      &\defneq&
    0
    \\
    \big(
      \Psi^{(0)}
    \big)_\rho^\alpha
     &\defneq&
    \delta^\alpha_\rho
    \\
    \big(
      \Omega^{(0)}
    \big)_\rho^{a b}
    &\defneq&
    0
  \end{array}
  \hspace{1cm}
  \mbox{and}
  \hspace{.8cm}
  \underset{ n \in \{1, \cdots, 32\}}{\forall}
  \;\;
  \left\{\!\!
  \def\arraystretch{1.7}
  \begin{array}{l}
  \big(
    E^{(n)}
    _{[\rho_1 \cdots \rho_n}
  \big)_{\rho]}^a
  \;\defneq\;
  0
  \\
  \big(
    \Psi^{(n)}
    _{[\rho_1 \cdots \rho_n}
  \big)_{\rho]}^\alpha
  \;\defneq\;
  0
  \\
  \big(
    \Omega^{(n)}
    _{[\rho_1 \cdots \rho_n}
  \big)_{\rho]}^{ab}
  \;\defneq\;
  0
  \mathrlap{\,.}
  \end{array}
  \right.
\end{equation}
\end{definition}
\begin{lemma}[\bf Direct implications of WZT gauge]
The WZT gauge conditions \eqref{GaugeConditions} imply:
\begin{equation}
  \label{GaugeConditionImplication}
  \def\arraystretch{1.6}
  \begin{array}{l}
    \Theta^\rho
    E^a_\rho
    \;=\;
    0
    \\
    \Theta^\rho
    \,
    \Psi^\alpha_\rho
    \;=\;
    \Theta^\rho
    \,
    \delta^\alpha_\rho 
    \,=:\,
    \Theta^\alpha
    \\
    \Theta^\rho
    \,
    \Omega_\rho{}^{a b}
    \;=\;
    0
  \end{array}
  \hspace{1cm}
  \mbox{\rm and}
  \hspace{.6cm}
  \underset{
    n \in \{1,\cdots, 32\}
  }{\forall}
  \;\;
  \left\{\!\!
  \def\arraystretch{1.7}
  \begin{array}{l}
    \Theta^\rho
    \, \partial_{\rho'}
    \big(E^{(n)}\big)^a_{\rho}
    \;=\;
    \big(E^{(n)}\big)^a_{\rho'}
    \\
    \Theta^\rho
    \, \partial_{\rho'}
    \big(\Psi^{(n)}\big)^\alpha_{\rho}
    \;=\;
    \big(\Psi^{(n)}\big)^\alpha_{\rho'}
    \\
    \Theta^\rho
    \, \partial_{\rho'}
    \big(\Omega^{(n)}\big)^{a b}_{\rho}
    \;=\;
    \big(\Omega^{(n)}\big)^{a b}_{\rho'}
    \mathrlap{\,.}
  \end{array}
  \right.
\end{equation}
\end{lemma}
\begin{proof}
The implications on the left of 
\eqref{GaugeConditionImplication}
are immediate (cf. \cite[(43-44)]{Tsimpis04}).
To see the equations on the right of \eqref{GaugeConditionImplication} we may proceed as follows:

\vspace{-1mm}
\begin{equation}
  \def\arraystretch{2}
  \begin{array}{lll}
    \Theta^\rho
    \,
    \partial_{\rho'}
    \big(
      E^{(n+1)}
    \big)^a_\rho
    &
    \;=\;
    \tfrac{1}{n!}
    \,
    \Theta^{\rho}
    \,
    \Theta^{\rho_2}
    \cdots
    \Theta^{\rho_{n+1}}
    \big(
      E
        ^{(n+1)}
        _{\rho' \, [\rho_2 \cdots \rho_{n+1} }
    \big)
      ^a_{\rho]}
    &
    \proofstep{
      by 
      \eqref{CoordinateComponentsOfCoframeField}
    }
    \\
  &  \;=\;
    \tfrac{1}{(n+1)!}
    \,
    \Theta^{\rho}
    \,
    \Theta^{\rho_2}
    \cdots
    \Theta^{\rho_{n+1}}
    \big(
      E
        ^{(n+1)}
        _{\rho \, \rho_2 \cdots \rho_{n+1} }
    \big)
      ^a_{\rho'}    
    &
    \proofstep{
      by \eqref{OneMoreIndexSkewSymmetrizing}
      \&
      \eqref{GaugeConditions}
    }
    \\
   & \;=\;
    \big(
      E^{(n+1)}
      _{\rho\, \rho_2 \cdots \rho_{n+1}}
    \big)^a_{\rho'}
    &
    \proofstep{
      by
      \eqref{CoordinateComponentsOfCoframeField}
      ,
    }
  \end{array}
\end{equation}
and verbatim so also for $E$ replaced by $\Psi$ or $\Omega$.
\end{proof}
\begin{remark}[\bf Fermionic normal coordinates and Rheonomy]
The WZT gauge of Def. \ref{WZTGauge} may be understood as a fermionic form of Riemann normal coordinates \cite[(A.3-4)]{McArthur84}\cite[(17-18)]{AtickDhar87}. In particular, the implication $\Theta^\rho \,\Omega_\rho{}^{ab} \,=\, 0$  \eqref{GaugeConditionImplication} has the further consequence that for translations along the odd coordinate direction (``rheonomy'' \cite[\S III.3.3]{CDF91}) the covariant derivative reduces to the plain coordinate derivative:
\begin{equation}
  \label{OddCovariantTranslationInWZTGauge}
  \Theta^\rho \, \nabla_\rho
  \;=\;
  \Theta^\rho \, \partial_\rho
  \,.
\end{equation}
\end{remark}

\smallskip

\noindent
{\bf Gravitino-flat supergravity solutions on super-space.}
For our purpose here, we focus on solutions to 11d supergravity, for which
the ordinary component of the gravitino field strength 
\eqref{CoFrameComponentsOfSuperFieldStrengths}
vanishes,
\begin{equation}
  \label{VanishingGravitinoFieldStrength}
  \rho_{a b} 
  \,\defneq\,
  0
\end{equation}
(which is the case for essentially all supergravity solutions of interest, cf. \cite[\S 12.6]{FreedmanVanProeyen12}).

With $\rho_{a b}$ also the super-curvature component $J^{a_1 a_2}{}_b$ vanishes  (cf. \cite[(161)]{GSS24-SuGra}), so that on gravitino-flat solutions the super-field strengths \eqref{CoFrameComponentsOfSuperFieldStrengths} have the form
\begin{equation}
  \label{CoFrameComponentsOfSuperFieldStrengthsOnGravitinoFlat}
  \def\arraystretch{1.3}
  \begin{array}{lcl}
    T^a &=& 0
    \\
    \rho &=& 
    H_a \Psi\, E^a
    \\
    R^{a_1 a_2}
    &=&
    \tfrac{1}{2}
    R^{a_1 a_2}{}_{b_1 b_2}
    \,
    E^{a_1}\, E^{a_2}
    \,+\,
    \big(\hspace{.8pt}
      \overline{\Psi}
      \,K^{a_1 a_2}\,
      \Psi
    \big)
    \,.
  \end{array}
\end{equation}

 \newpage 

\begin{lemma}[\bf $\Theta$-independence of field components]
  For gravitino-flat \eqref{VanishingGravitinoFieldStrength}
  super-space solutions of 11d SuGra in WZT gauge {\rm (Def. \ref{WZTGauge})} the following super-field strength components \eqref{CoFrameComponentsOfSuperFieldStrengthsOnGravitinoFlat} are all independent of the odd coordinates $\Theta^\rho$:
  \vspace{-2mm} 
\begin{equation}
  \label{ThetaIndependenceOfFieldComponents}
  \def\arraystretch{1.8}
  \begin{array}{ll}
  \scalebox{.7}{
    \bf \color{darkblue}
    \def\arraystretch{.9}
    \begin{tabular}{l}
    The flux densities
    \end{tabular}
  }
  &
  \partial_{\rho}
  \big(
    (G_4)_{a_1 \cdots a_4}
  \big)
  \;=\;
  0
  \,,\;\;\;\;\;\;\;\;
  \partial_{\rho}
  \big(
    (G_7)_{a_1 \cdots a_7}
  \big)
  \;=\;
  0
  \,,
  \\
  \scalebox{.7}{
    \color{darkblue}
    \bf
    \def\arraystretch{.9}
    \begin{tabular}{l}
      Odd co-frame component of
      \\
      the gravitino field strength
    \end{tabular}
  }
  &
  \partial_\rho
  \big(
    H_a
  \big)
  \;=\;
  0
  \,,
  \\
  \scalebox{.7}{
    \color{darkblue}
    \bf
    \def\arraystretch{.9}
    \begin{tabular}{l}
      Odd co-frame components
      \\
      of the super-curvature
    \end{tabular}
  }
  &
  \partial_\rho
  \big(
    K^{a_1 a_2}
  \big)
  \;=\;
  0
  \,.
\end{array}
\end{equation}
\end{lemma}
\begin{proof}
This follows by use of the well-known super-space constraints, which we quote from \cite{GSS24-SuGra} (where full derivation and referencing are given). First, the $\Theta$-independence of $G_4$ follows by 
$$
  \def\arraystretch{1.6}
  \begin{array}{lll}
  \Theta^\rho
  \,
  \partial_\rho
  \big(
    (G_4)_{a_1 \cdots a_4}
  \big)
  &
  \;=\;
  \Theta^\rho
  \,
  \nabla_\rho
  \big(
    (G_4)_{a_1 \cdots a_4}
  \big)
  &
  \proofstep{
    by
    \eqref{OddCovariantTranslationInWZTGauge}
  }
  \\
  &\;=\;
  12 
  \big(
    \overline{\Theta}
    \,\Gamma_{[a_1 a_2}\,
    \rho_{a_2 a_3]}
  \big)
  &
  \proofstep{
    by
    \cite[(136)]{GSS24-SuGra}
  }
  \\
  &\;=\;
  0
  &
  \proofstep{
    by 
    \eqref{VanishingGravitinoFieldStrength}.
  }
\end{array}
$$
But the remaining components in \eqref{ThetaIndependenceOfFieldComponents} are linear functions of $(G_4)_{a_1 \cdots a_4}$: 
\begin{equation}
  \label{HK}
  \def\arraystretch{1.6}
  \begin{array}{ccll}
    H_a
    &=&
    \tfrac{1}{6}
    \tfrac{1}{3!}
    (G_4)_{a \, b_1 b_2 b_3}
    \,
    \Gamma^{b_1 b_2 b_3}
    \,-\,
    \tfrac{1}{12}
    \tfrac{1}{4!}
    (G_4)^{b_1 \cdots b_4}
    \,
    \Gamma_{a\, b_1 \cdots b_4}
    &
    \proofstep{
      \cite[(135)]{GSS24-SuGra}
    }
    \\
    &=&
    \tfrac{1}{6}
    \tfrac{1}{3!}
    (G_4)_{a\, b_1 b_2 b_3}
    \,
    \Gamma^{b_1 b_2 b_3}
    \,+\,
    \tfrac{1}{12}
    \tfrac{1}{6!}
    (G_7)_{a\, c_1 \cdots c_6}
    \,
    \Gamma^{c_1 \cdots c_6}
    &
    \proofstep{
      \cite[(148)]{GSS24-SuGra}
    }
    \\[10pt]
    K^{a_1 a_2}
    & = &
    -
    \tfrac{1}{6}
    \Big(
      (G_4)^{a_1 a_2 \, b_1 b_2}
      \Gamma_{b_1 b_2}
      \;+\;
      \tfrac{1}{4!}
      (G_4)_{b_1 \cdots b_4}
      \Gamma
        ^{a_1 a_2 \, b_1 \cdots b_4}
    \Big)
    &
    \proofstep{
      \cite[(162)]{GSS24-SuGra}
    }
    \\[4pt]
    & = &
    -
    \tfrac{1}{6}
    \Big(
      (G_4)^{a_1 a_2 \, b_1 b_2}
      \Gamma_{b_1 b_2}
      \;+\;
      \tfrac{1}{5!}
      (G_7)
        ^{a_1 a_2 \, b_1 \cdots b_5}
      \Gamma
        _{b_1 \cdots b_5}
    \Big)
  \end{array}
\end{equation}
and hence their $\Theta$-dependence vanishes with that of $G_4$ and $G_7$.
\end{proof}


\medskip

\noindent
{\bf Supergravity field extension to super-space.} We now consider solutions to the rheonomy equations for extending on-shell 11d supergravity fields to superspace, cast into recursion relations in the polynomial order of their odd coordinate field dependence as in \cite{Tsimpis04} (similar to \cite[(3.9)]{dWPPS98}), but specialized to the case of gravitino-flat spacetimes \eqref{VanishingGravitinoFieldStrength}.

\begin{lemma}[\bf Rheonomy for the graviton]
  \label{RheonomyForBosonicCOFrameComponents}
  In WZT gauge 
  \eqref{GaugeConditions}
  the following recursion relations hold for the bosonic coframe field components \eqref{CoordinateComponentsOfCoframeField},
  recursing in their odd coordinate degree $n + 1 \,\in\, \{1, \cdots, 32\}$:
  \begin{equation}
    \label{RecursionRelationForGraviton}
    \def\arraystretch{1.8}
    \begin{array}{ccl}
    (
      E^{(n+1)}
    )^a_\rho
    &=&
    \tfrac{2}{n+2}
    \big(\,
      \overline{\Theta}
      \,\Gamma^a\,
      \Psi^{(n)}_{\rho}
    \big)
    \,,
    \\
    (
      E^{(n+1)}
    )^a_r
    &=&
    \tfrac{2}{n+1}
    \big(\,
      \overline{\Theta}
      \,\Gamma^a\,
      \Psi^{(n)}_r
    \big)
    \end{array}
  \end{equation}
\end{lemma}
\noindent
(cf. \cite[(58, 59)]{Tsimpis04}.\footnote{
\label{RelativeFactorsComparedToTsimpis}
  The factor of 
  ``$\mathrm{i}/2$'' by which our 
  \eqref{RecursionRelationForGraviton} differs from \cite[(58, 59)]{Tsimpis04} 
  is absorbed by our convention for the spacetime signature, the Clifford algebra and the Majorana spinor: Our $\Gamma$-matrices are $\mathrm{i}$ times the Gamma matrices there (which makes all expressions in Majorana spinors manifestly real, cf. \cite[Rem. 1.7]{GSS24-SuGra})), and we do not include a factor of $\sfrac{1}{2}$ multiplying the $(\Psi^2)$-term in the definition of the super-torsion \eqref{GravitationalFieldStrengths}.
  }
\begin{proof}
The $\mathrm{d}\Theta^\rho$-component 
of \eqref{RecursionRelationForGraviton}
follows as:
$$
  \def\arraystretch{1.7}
  \begin{array}{cccll}
    &
    \mathrm{d}
    \, 
    E^a
    &=&
    \Omega^a{}_b \, E^b
    \,+\,
    \big(
     \overline{\Psi}
     \,\Gamma^a\,
     \Psi
    \big)
    &
    \proofstep{
      from
      \eqref{GravitationalFieldStrengths}
    }
    \\
    \Rightarrow
    &
    \Theta^\rho
    \,
    \partial{\mathclap{\phantom{)}}}_{(\rho} 
    \,
    E^a_{\rho')}
    &=&
    \Theta^\rho
    \,
    \big(
      \Omega^a{}_b
    \big)_{\!(\rho} E^b_{\rho')}
    \,+\,
    \Theta^\rho
    \,
    \Psi^{\alpha}_{(\rho}
    \,
    \Psi^{\alpha'}_{\rho')}
    \Gamma^a_{\alpha \alpha'}
    &
    \proofstep{
      by
      \eqref{CoordinateComponentsOfCoframeField}
    }
    \\
    \Leftrightarrow
    &
    \Theta^\rho
    \partial{\mathclap{\phantom{)}}}_{(\rho} 
    \,
    E^a_{\rho')}
    &=&
    \Theta^\rho
    \delta^{\alpha}_{\rho}
    \,
    \Psi^{\alpha'}_{\rho'}
    \Gamma^a_{\alpha \alpha'}
    &
    \proofstep{
      by
      \eqref{GaugeConditionImplication}
      \& 
      \eqref{BiDegreeOfComponentFunctions}
    }
    \\
    \Rightarrow
    &
    \underbrace{
      \Theta^\rho
      \partial{\mathclap{\phantom{)}}}_{(\rho} 
      \,
      \big(
        E^{(n+1)}
      \big)^a_{\rho')}
    }_{\color{gray}
      \scalebox{.7}{$
        \tfrac{(n+2)}{2}
        \big(
          E^{(n+1)}
        \big)^a_{\rho'}
      $}
    }
    &=&
    \underbrace{
      \Theta^\alpha
      (
        \Psi^{(n)}
      )^{\alpha'}
        _{\rho'}
      \Gamma^a_{\alpha \alpha'}    
    }_{\color{gray}
      \scalebox{.7}{$
      \big(
        \overline{\Theta}
        \,\Gamma^a\,
        \Psi^{(n)}
      \big)
      $}
    }
    &
    \proofstep{
      by
      \eqref{CoordinateComponentsOfCoframeField}
      \&
      \eqref{GaugeConditionImplication},
    }
  \end{array}
$$
and the $\mathrm{d}X^r$-component as:
$$
  \def\arraystretch{1.7}
  \begin{array}{cccll}
    &
    \mathrm{d}
    \, 
    E^a
    &=&
    \Omega^a{}_b \, E^b
    \,+\,
    \big(\,
     \overline{\Psi}
     \,\Gamma^a\,
     \Psi
    \big)
    &
    \proofstep{
      from
      \eqref{GravitationalFieldStrengths}
    }
    \\
    \Rightarrow
    &
    \Theta^\rho
    \,
    \partial_\rho E^a_r
    &=&
    \Theta^\rho
    \big(\Omega^a{}_b\big)\mathclap{\phantom{)}}_{\!\rho}
    \,
    E^b_r
    \,-\,
    \Theta^\rho
    \big(\Omega^a{}_b\big)_{\!r}
    \,
    E^b_{\rho}
    \,+\,
    2
    \,
    \Theta^\rho
    \,
    \Psi^\alpha_{\rho}
    \,
    \Psi^{\alpha'}_{r}
    \Gamma^a_{\alpha \alpha'}
    &
    \proofstep{
      by
      \eqref{CoordinateComponentsOfCoframeField}
    }
    \\
    \Leftrightarrow
    &
    \Theta^\rho
    \partial_\rho E^a_r
    &=&
    2
    \,
    \Theta^\rho
    \delta^\alpha_\rho
    \,
    \Psi^{\alpha'}_{r}
    \Gamma^a_{\alpha \alpha'}
    &
    \proofstep{
      by
      \eqref{GaugeConditionImplication}
    }
    \\
    \Rightarrow
    &
    \underbrace{
    \Theta^\rho
    \partial_{\rho}
    \big(
      E^{(n+1)}
    \big)^a_r
    }_{\color{gray}
      \scalebox{.7}{$
      (n+1)\big(
        E^{(n+1)}
      \big)^a_r
      $}
    }
    &=&
    2
    \,
    \big(\,
      \overline{\Theta}
      \,\Gamma^a\,
      \Psi^{(n)}_r
    \big)
    &
    \proofstep{
      by 
      \eqref{CoordinateComponentsOfCoframeField}
      \&
      \eqref{GaugeConditionImplication}.
    }
  \end{array}
$$

\vspace{-4mm} 
\end{proof}

\vspace{1mm} 
\begin{lemma}[\bf Rheonomy for the spin-connection]
  \label{RheonomyForSpinConnection}
  On gravitino-flat \eqref{VanishingGravitinoFieldStrength} super-spacetimes in WZT gauge 
  \eqref{GaugeConditions} we have
  the following recursion relations for the 
  spin connection \eqref{CoordinateComponentsOfCoframeField},
  recursing in the odd coordinate degree $n + 1 \,\in\, \{1, \cdots, 32\}$:
  \begin{equation}
    \label{RecursionForSpinConnection}
    \def\arraystretch{1.8}
    \begin{array}{ccl}
      \big(
        \Omega^{(n+1)}
      \big)^{a_1 a_2}_\rho
      &=&
      \tfrac{2}{n+2}
      \big(\,
        \overline{\Theta}
        \,K^{a_1 a_2}\,
        \Psi^{(n)}_\rho
      \big)
      \\
      \big(
        \Omega^{(n+1)}
      \big)^{a_1 a_2}_r
    &=&
    \tfrac{2}{n+1}
    \big(\,
      \overline{\Theta}
      \,K^{a_1 a_2}\,
      \Psi^{(n)}_r
    \big)
    \end{array}
  \end{equation}
\end{lemma}
\noindent
(cf. \cite[(61, 64)]{Tsimpis04}\footnote{
  As in footnote \ref{RelativeFactorsComparedToTsimpis},
  the difference of
  \cite[(61, 64)]{Tsimpis04}
  from \eqref{RecursionForSpinConnection} 
  by a factor of $\mathrm{i}/2$ is due to our spinor convention.
} noticing our \eqref{ThetaIndependenceOfFieldComponents}).
\begin{proof}
In \eqref{RecursionForSpinConnection} the $\mathrm{d}\Theta^\rho$-component follows by:
$$
  \def\arraystretch{1.8}
  \begin{array}{lccll}
    &
    \mathrm{d}
    \, 
    \Omega^{a_1 a_2}
    &=&
    \Omega^{a_1}{}_b
    \,
    \Omega^{b a_2}
    +
    R^{a_1 a_2}
    &
    \proofstep{
      from
      \eqref{GravitationalFieldStrengths}
    }
    \\
    \Rightarrow
    &
    \Theta^{\rho'}
    \,
    \partial_{(\rho'}
    \,
    (\Omega^{a_1 a_2})_{\rho)}
    &=&
    \Theta^{\rho'}
    \,
    \delta^{\alpha'}_{\rho'}
    \,
    \Psi^{\alpha}_{\rho}
    \,
    K^{a_1 a_2}_{\alpha' \alpha}
    &
    \proofstep{
      by
      \eqref{CoFrameComponentsOfSuperFieldStrengthsOnGravitinoFlat},
      \eqref{CoordinateComponentsOfCoframeField}
      \&
      \eqref{GaugeConditionImplication}
    }
    \\
    \Rightarrow
    &
    \underbrace{
      \Theta^{\rho'}
      \,
      \partial_{(\rho'}
      \,
      \big(
        \Omega^{(n+1)}
      \big)^{a_1 a_2}_{\rho)}
    }_{\color{gray}
      \scalebox{.7}{$
      \tfrac{(n+2)}{2}
      \big(
        \Omega^{(n+1)}
      \big)^{a_1 a_2}_{\rho}
      $}
    }
    &=&
    \big(\,
      \overline{\Theta}
      \,K^{a_1 a_2}\,
      \Psi^{(n)}_{\rho}
    \big)
    &
    \proofstep{
      by 
      \eqref{CoordinateComponentsOfCoframeField},
      \eqref{GaugeConditionImplication}
      \&
      \eqref{ThetaIndependenceOfFieldComponents},
    }
  \end{array}
$$
and the $\mathrm{d}X^r$-component by:
$$
  \def\arraystretch{1.8}
  \begin{array}{lccll}
    &
    \mathrm{d}
    \, \Omega^{a_1 a_2}
    &=&
    \Omega^{a_1}{}_b
    \,
    \Omega^{b a_2}
    +
    R^{a_1 a_2}
    &
    \proofstep{
      from
      \eqref{GravitationalFieldStrengths}
    }
    \\
    \Rightarrow
    &
    \Theta^\rho
    \,
    \partial_\rho
    (\Omega^{a_1 a_2})_r
    &=&
    2\,
    \Theta^\rho
    \,
    \Psi^{\alpha}_\rho
    \,
    \Psi^{\alpha'}_r
    \,
    K^{a_1 a_2}_{\alpha \alpha'}
    &
    \proofstep{
      by
      \eqref{CoFrameComponentsOfSuperFieldStrengthsOnGravitinoFlat}, \eqref{CoordinateComponentsOfCoframeField},
            \&
      \eqref{GaugeConditionImplication}
    }
    \\
    \Rightarrow
    &
    \underbrace{
    \Theta^\rho
    \,
    \partial_\rho
    \big(\Omega^{(n+1)}\big)^{a_1 a_2}_r
    }_{\color{gray}
      \scalebox{.7}{$
      (n+1)
      \big(
        \Omega^{(n+1)}
      \big)^{a_1 a_2}_r
      $}
    }
    &=&
    2\,
    \big(\,
      \overline{\Theta}
      \,K^{a_1 a_2}\,
      \Psi^{(n)}
    \big)
    &
    \proofstep{
      by
      \eqref{CoordinateComponentsOfCoframeField},
      \eqref{GaugeConditionImplication}
      \&
      \eqref{ThetaIndependenceOfFieldComponents}.
    }
  \end{array}
$$

\vspace{-3mm} 
\end{proof}

\begin{lemma}[\bf Rheonomy for the gravitino]
  On gravitino-flat \eqref{VanishingGravitinoFieldStrength}
  super-spacetimes
  in WZT gauge 
  \eqref{GaugeConditions}
  the following recursion relations 
  hold for the odd coordinate dependence of the 
  gravitino field \eqref{CoordinateComponentsOfCoframeField}:
  \def\arraystretch{1.9}
  \begin{equation}
    \label{RecursionRelationForGravitino}
    \def\arraystretch{1.9}
    \begin{array}{lcl}
      \big(
        \Psi^{(n+1)}
      \big)^\alpha_{\rho}
      &=&
    +
    \tfrac{1}{n+2}
    \tfrac{1}{4}
    \big(
      \Gamma_{a b} 
      \Theta
    \big)^\alpha
    \big(
      \Omega^{(n)}
    \big)^{a b}_{\rho}
    \,+\,
    \tfrac{1}{n+2}
    \,
    (H_a \Theta)^\alpha
    \big(
      E^{(n)}
    \big)^a_{\rho}
    \\
    \big(
      \Psi^{(n+1)}
    \big)^\alpha_r
    &=&
    -
    \tfrac{1}{n+1}
    \tfrac{1}{4}
    (\Gamma_{a b} \Theta)^\alpha
    \big(
      \Omega^{(n)}
    \big)^{a b}_r
    \;+\,
    \tfrac{1}{n+1}
    (H_a \Theta)^\alpha
    \big(
      E^{(n)}
    \big)^a_r
    \mathrlap{\,.}
    \end{array}
  \end{equation}
  $\phantom{AA}$
\end{lemma}

\begin{proof}
In \eqref{RecursionRelationForGravitino} the $\mathrm{d}\Theta^\rho$-component follows by:
$$
  \def\arraystretch{1.9}
  \begin{array}{cccll}
    &
    \mathrm{d}
    \, 
    \Psi^\alpha
    &=&
    \tfrac{1}{4}\Omega^{ab} 
    \,
    (\Gamma_{a b} \Psi)^\alpha
    \,+\,
    \rho^\alpha
    &
    \proofstep{
      from
      \eqref{GravitationalFieldStrengths}
    }
    \\
    \Rightarrow
    &
    \Theta^{\rho'} 
    \,
    \partial\mathclap{\phantom{)}}_{(\rho'}
    \, 
    \Psi^\alpha_{\rho)}
    &=&
    \tfrac{1}{4}
    \Theta^{\rho'}
    (\Omega^{ab})_{(\rho'} 
    (\Gamma_{a b} \Psi_{\rho)})^\alpha
    \,+\,
    \Theta^{\rho'}
    (H_a \Psi_{(\rho'})^\alpha \, E^a_{\rho)}
    &
    \proofstep{
      by
      \eqref{CoFrameComponentsOfSuperFieldStrengthsOnGravitinoFlat},
     \eqref{CoordinateComponentsOfCoframeField},
      \&
      \eqref{GaugeConditionImplication}
    }
    \\
    \Rightarrow
    &
    \underbrace{
      \Theta^{\rho'}
      \,
      \partial\mathclap{\phantom{)}}_{(\rho'}
      \big(
        \Psi^{(n+1)}
      \big){}^\alpha_{\rho)}
    }_{\color{gray}
      \scalebox{.9}{$
      \tfrac{n+2}{2}
      \big(
        \Psi^{(n+1)}
      \big)^\alpha_{\rho}
      $}
    }
    &=&
    \tfrac{1}{2}
    \tfrac{1}{4}
    \big(
      \Gamma_{a b} 
      \Theta
    \big)^\alpha
    \big(
      \Omega^{(n)}
    \big)^{a b}_{\rho}
    \,+\,
    \tfrac{1}{2}
    \,
    (H_a \Theta)^\alpha
    \big(
      E^{(n)}
    \big)^a_{\rho}
   &
   \proofstep{
     by
     \eqref{CoordinateComponentsOfCoframeField},
      \eqref{GaugeConditionImplication}
      \&
      \eqref{ThetaIndependenceOfFieldComponents},
   }
\end{array}
$$

\vspace{2mm} 
\noindent and the $\mathrm{d}X^a$-component by:
$$
  \def\arraystretch{1.9}
  \begin{array}{cccll}
    &
    \mathrm{d}
    \, 
    \Psi^\alpha
    &=&
    \tfrac{1}{4}\Omega^{ab} 
    (\Gamma_{a b} \Psi)^\alpha
    \,+\,
    \rho^\alpha
    &
    \proofstep{
      from
      \eqref{GravitationalFieldStrengths}
    }
    \\
    \Rightarrow
    &
    \Theta^\rho
    \,
    \partial_\rho
    \Psi^\alpha_r
    &=&
    -
    \Theta^\rho
    \,
    \tfrac{1}{4}
    \Omega^{a b}_r
    (\Gamma_{a b} \Psi_\rho)
    \,+\,
    \Theta^\rho
    (H_a \Psi_{\rho})^\alpha
    \,
    E^a_r
    &
    \proofstep{
      by
      \eqref{CoFrameComponentsOfSuperFieldStrengthsOnGravitinoFlat},
     \eqref{CoordinateComponentsOfCoframeField},
      \&
      \eqref{GaugeConditionImplication}
    }
    \\
    \Rightarrow
    &
    \underbrace{
      \Theta^\rho
      \partial_\rho
      \big(
        \Psi^{(n+1)}
      \big)^\alpha_r
    }_{\color{gray}
      (n+1)
      \big(
        \Psi^{(n+1)}
      \big)^\alpha_r
    }
    &=&
    -\tfrac{1}{4}
    (\Gamma_{a b}\Theta)
    \big(\Omega^{(n)}\big)^{a b}_r
    \,+\,
    (H_a \Theta)^\alpha
    \,
    \big(
      E^{(n)}
    \big)^a_r
    &
    \proofstep{
      by
     \eqref{CoordinateComponentsOfCoframeField},
      \eqref{GaugeConditionImplication}
      \&
      \eqref{ThetaIndependenceOfFieldComponents}.
    }
\end{array}
$$
Notice here how the sign in the second line appears since only the coefficient of $\mathrm{d}X^r\, \mathrm{d}\Theta^\rho$ contributes in the first term, which picks up a sign $\mathrm{d}X^r\, \mathrm{d}\Theta^\rho \,=\, -\, \mathrm{d}\Theta^\rho\, \mathrm{d}X^r$ in comparison to the
left hand side.
\end{proof}

By inserting these recursion relations into each other, we may decouple them (resulting in a formulation similar to \cite[(3.9)]{dWPPS98}):
\begin{lemma}[\bf Decoupled rheonomy recursion relations]
  \label{DecoupledRheonomyRecursionRelations}
  On gravitino-flat \eqref{VanishingGravitinoFieldStrength} super-spacetimes in WZT gauge 
  \eqref{GaugeConditions}
  the following decoupled recursion relations hold for the odd coordinate dependence of the super-fields:
  \begin{equation}
    \label{DecoupledRecursionForGravitino}
    \hspace{-5mm} 
    \def\arraystretch{2}
    \begin{array}{ll}
    \begin{array}{ccl}
      \big(
        \Psi^{(n+2)}
      \big)^\alpha_\rho
      &=&
      +
      \tfrac{1}{n+4}
      \tfrac{2}{n+3}
      \tfrac{1}{4}
      \big(
        \Gamma_{a_1 a_2}
        \Theta
      \big)^\alpha
      \Big(
        \overline{\Theta}
        \,K^{a_1 a_2}\,
        \Psi^{(n)}_\rho
      \Big)
      \,+\,
      \tfrac{1}{n+4}
      \tfrac{2}{n+3}
      (H_a \Theta)^\alpha
      \Big(
        \overline{\Theta}
        \,\Gamma^a\,
        \Psi^{(n)}_\rho
      \Big)
      \\
      \big(
        \Psi^{(n+2)}
      \big)^\alpha_r
      &=&
      -\tfrac{1}{n+2}
      \tfrac{1}{n+1}
      \tfrac{1}{4}
      (\Gamma_{a_1 a_2}\Theta)^\alpha
      \Big(
        \overline{\Theta}
        \,K^{a_1 a_2}\,
        \Psi^{(n)}_r
      \Big)
      \,+\,
      \tfrac{1}{n+2}
      \tfrac{1}{n+1}
      (H_a \Theta)^\alpha
      \Big(
        \overline{\Theta}
        \,\Gamma^a\,
        \Psi^{(n)}_r
      \Big)
    \end{array}
    &
      \proofstep{
        \rm
        \def\arraystretch{.95}
        \begin{tabular}{l}
        by
        inserting
        \\
        \eqref{RecursionForSpinConnection}
        \&
        \eqref{RecursionRelationForGraviton}
        \\
        into 
        \eqref{RecursionRelationForGravitino}\,.
        \end{tabular}
      }
    \end{array}
  \end{equation}
\end{lemma}

\section{Holographic M5-Branes}
\label{HolographicM5Branes}

\subsection{Spinors on M5-branes}
\label{SpinorsOnM5Branes}

We briefly recall and record some properties of spinors in 6d among spinors in 11d, following \cite[\S 3.2]{GSS24-FluxOnM5}, which we will need below. In particular, we establish a Fierz identity (in Lem. \ref{FierzIdentityIn6d} below), which is crucial in the proof of the M5-immersion in \S\ref{TheHolographicM5Immersion} below.
In contrast to existing literature on super-embeddings, we do not use a matrix representation of the 6d Clifford algebra but instead use projection operators \eqref{ProjectorOn6dSpinRep}
to algebraically carve it out of the 11d Clifford algebra (as indicated in \cite[\S A]{LambertPapageorgakis10}). We find that this helps considerably with providing
the proofs in the following sections.

\medskip 
\noindent
{\bf Spinors in 6d form 11d.}
Following \cite[\S 3.2]{GSS24-FluxOnM5} we conveniently identify the chiral $\mathrm{Spin}(1,5)$-representations $2 \cdot \mathbf{8}_{\pm} \in \mathrm{Rep}_{\mathbb{R}}\big(\mathrm{Spin}(1,5)\big)$ with the linear subspaces of the $\mathrm{Spin}(1,10)$-representation $\mathbf{32}$ \eqref{The11dMajoranaRepresentation} which are the images of the projection operators
(\cite[(92)]{GSS24-FluxOnM5})
\begin{equation}
  \label{ProjectorOn6dSpinRep}
  \def\arraystretch{1.4}
  \begin{array}{l}
    P
    \,:=\,
    \tfrac{1}{2}
    \big(
      1 
      +
      \Gamma_{5'6789}
    \big)
    \\
    \overline{P}
    \,:=\,
    \tfrac{1}{2}
    \big(
      1 
      -
      \Gamma_{5'6789}
    \big)
  \end{array}
  \;:\;
  \mathbf{32}
  \xrightarrow{\;\;}
  \mathbf{32}
  \,,
\end{equation}
respectively, satisfying the following evident but consequential relations (cf. \cite[(89)]{GSS24-FluxOnM5}):
\begin{equation}
  \label{PropertiesOfProjectionOperator}
  \begin{array}{l}
    P \,P \,=\, P
    \\
    \overline{P} \, \overline{P}
    \,=\, \overline{P}
    \\
    \overline{P} \, P \,=\, 0
    \\
    P \, \overline{P} \,=\, 0
  \end{array}
  \hspace{.9cm}
  \begin{array}{ll}
    \begin{array}{l}
    \Gamma^a \, P 
    \;=\;
    \overline{P}\, \Gamma^a
    \\
    \Gamma^a \, \overline{P} 
    \;=\;
    P\, \Gamma^a
    \end{array}
    &
    a \in \{0,1,2,3,4,5\}
    \\
    \begin{array}{l}
      \Gamma^{5'} P 
      \,=\,
      P \, \Gamma^{5'}
    \end{array}
    \\
    \begin{array}{l}
    \Gamma^{i} \, P 
    \,=\,
    P \, \Gamma^{i}
    \end{array}
    &
    i \in \{6,7,8,9\}
  \end{array}
  \hspace{.9cm}
  \def\arraystretch{1.4}
  \begin{array}{l}
    \Gamma_{5'6789}
    \,
    P
    \;=\;
    + P
    \\
    \Gamma_{5'6789}
    \,
    \overline{P}
    \;=\;
    - \overline{P}
    \\[+5pt]
    \Gamma_{6789} P
    \;=\;
    \Gamma_{5'} P
    \,,
  \end{array}
\end{equation}
where we suggestively denote the 11d Clifford generators as follows:
\begin{equation}
  \label{TheCliffordGenerators}
  \hspace{-6.7cm}
  \begin{tikzcd}[
    column sep=3pt,
    row sep=0pt
  ]
    &&
    \mathclap{
      \hspace{20pt}
      \overbrace{
      \phantom{-------------}
      }^{
        \scalebox{.7}{
          \color{darkblue}
          \bf tangential
        }
      }
    }
    &&&&
    \overbrace{
      \phantom{--}
    }^{
        \scalebox{.7}{
          \color{darkblue}
          \bf radial
        }
    }
    &&
    \mathclap{
      \hspace{15pt}
      \overbrace{
      \phantom{\phantom{--------}}
      }^{
        \scalebox{.7}{
          \color{darkblue}
          \bf transversal
        }
      }
    }
    \\[-10pt]
    \Gamma_0
    \ar[
      d,
      |->,
      start anchor={[yshift=+1pt]},
      end anchor={[yshift=-1pt]},
      bend right=90,
      shift right=0pt,
      "{\color{darkgreen}
        \def\arraystretch{.9}
        \begin{array}{l}
          \phantom{+}\overline{P}(-)P
          \\
          +
          P(-)\overline{P}
        \end{array}
      }"{swap, pos=.5}
    ]
    &
    \Gamma_1
    &
    \Gamma_2
    &
    \Gamma_3
    &
    \Gamma_4
    &
    \Gamma_5
    &
    \Gamma_{5'}
    &
    \Gamma_{6}
    &
    \Gamma_{7}
    &
    \Gamma_{8}
    &
    \Gamma_{9}
    &
    \mathrlap{
       \in 
      \mathrm{Pin}^+(1,10)
      \,\subset\,
      \mathrm{End}_{\mathbb{R}}(\mathbf{32})
    }
    \\
    \gamma_0
    &
    \gamma_1
    &
    \gamma_2
    &
    \gamma_3
    &
    \gamma_4
    &
    \gamma_5
    &
    &&&&&
    \mathrlap{
    \in 
    \mathrm{Pin}^+(1,5)
    \,\subset\,
    \mathrm{End}_{\mathbb{R}}
    \big(
      2\cdot \mathbf{8}_+
      \oplus
      2 \cdot \mathbf{8}_-
    \big)
    \,,
    }
  \end{tikzcd}
\end{equation}
in that under the corresponding inclusion 
$$
  \mathrm{Spin}(1,5) 
  \xhookrightarrow{\;\;}
  \mathrm{Spin}(1,10)
$$
there are isomorphisms \cite[(86-90)]{GSS24-FluxOnM5}
\vspace{2mm} 
\begin{equation}
  \label{6dSpinRepAsImageOfProjector}
  \begin{tikzcd}[sep=0pt]
    \mathllap{
      2 \cdot \mathbf{8}
      \,:=\;\;
    }
    2 \cdot \mathbf{8}_+
    \;\simeq\;
    P(\mathbf{32})
    \\[-1pt]
    2 \cdot \mathbf{8}_-
    \;\simeq\;
    \overline{P}(\mathbf{32})
    \mathrlap{\,.}
  \end{tikzcd}
\end{equation}

Combined with the vector representation of $\mathrm{Spin}(1,10)$ and $\mathrm{Spin}(1,5)$ on $\mathbb{R}^{1,10}$ and $\mathbb{R}^{1,5}$, respectively, we may regard $P$
\eqref{ProjectorOn6dSpinRep}
as a projector of super-vector spaces
\begin{equation}
  \label{SuperProjectionOperator}
  \adjustbox{
    raise=-6pt
  }{
   \begin{tikzcd}
    \mathbb{R}
      ^{
      1,10 \,\vert\, \mathbf{32}
      }
    \ar[
      r,
      ->>
    ]
    \ar[
      rr,
      rounded corners,
      to path={
           ([yshift=+00pt]\tikztostart.north)  
        -- ([yshift=+08pt]\tikztostart.north)
        -- node[]{
          \scalebox{.8}{
            \colorbox{white}{
              $P$
            }
          }
        }
           ([yshift=+08pt]\tikztotarget.north)
        -- ([yshift=-00pt]\tikztotarget.north)
      }
    ]
    &
    \mathbb{R}
      ^{1,5 \,\vert\, 2 \cdot \mathbf{8}}
    \ar[
      r,
      hook
    ]
    &
    \mathbb{R}^{
      1,10 \,\vert\, \mathbf{32}
    }
  \end{tikzcd}
  }
  \hspace{.4cm}
  \def\arraystretch{1.3}
  \begin{array}{l}
     P 
       \,:=\, 
    \tfrac{1}{2}
    \big(
      1 + \Gamma_{5'6789}
    \big) 
    \\
    \overline{P}
    \,:=\,
    \tfrac{1}{2}
    \big(
      1 - \Gamma_{5'6789}
    \big) 
    \,,
  \end{array}
\end{equation}
which is convenient for unifying the conditions on tangential and transversal super-coframe components in a $\sfrac{1}{2}$BPS super-immersion (Def. \ref{HalfBPSSuperImmersion} below).

\smallskip

\begin{lemma}[\bf A Fierz identity in 6d]
  \label{FierzIdentityIn6d}
  Elements $\theta \in (2 \cdot \mathbf{8}_+)_{\mathrm{odd}}$ satisfy 
  \begin{equation}
    \label{6dFierzIdentity}
    \gamma_a\theta
    \,
    \overline{\theta}\gamma^a
    \;=\;
    0
    \,,
  \end{equation}
  where $\gamma_a$ are the Clifford generators in 6D according to \eqref{TheCliffordGenerators}.
\end{lemma}
\begin{proof}
  Recall from \eqref{6dSpinRepAsImageOfProjector} that we may and do regard $\theta = P \theta \in 2 \cdot \mathbf{8}_+ \subset \mathbf{32}$ as an 11d spinor but constrained to be in the image of the projector $P := \tfrac{1}{2}\big(1 + \Gamma_{5'6789}\big)$, see \eqref{ProjectorOn6dSpinRep}. With this, we may use the formula for Clifford expansion 
  \eqref{CliffordExpansionOfEndomorphismOf32}
  of general endomorphisms $\phi \in \mathrm{End}_{\mathbb{R}}(\mathbf{32})$
  in the case where 
  $$
    \begin{tikzcd}[
      sep=0pt
    ]
    \phi 
      \;\defneq\;
    \theta \, \overline{\theta}
    \ar[
      r,
      phantom,
      "{ : }"
    ]
    &[+1pt]
    \mathbf{32}
    \ar[
      rr
    ]
    &&
    \mathbf{32}
    \\[-2pt]
    &
    \Phi 
      &\longmapsto&
    \theta\big(\hspace{.8pt}
      \overline{\theta}
      \,
      \Phi
    \big)
    \,,
    \end{tikzcd}
  $$
  with the spinor pairing \eqref{TheSpinorPairing} on the right.
  
  But since $\theta$ (as opposed to $\mathrm{d}\theta$, cf. \cite[Rem. 2.62]{GSS24-SuGra}) is a skew-commuting variable, it is only the skew-symmetric Clifford basis elements among $\Gamma_{a_1 \cdots a_p}$ ($p \leq 5$) which are non-vanishing when evaluated in $\big(\hspace{1pt}\overline{\theta}(-)\theta\big)$, and these are precisely those with 0, 3 or 4 indices \eqref{SkewSpinorPairings}. Hence \eqref{CliffordExpansionOfEndomorphismOf32} specializes to:
  $$
    \theta \, \overline{\theta}
    \;=\;
    -
    \tfrac{1}{32}
    \Big(
      \big(\,
        \overline{\theta}
        \,
        \theta
      \big)
      \,-\,
      \tfrac{1}{3!}
      \big(\hspace{.8pt}
        \overline{\theta}
        \,
        \Gamma_{a_1 a_2 a_3}
        \,
        \theta
      \big)
      \,
      \Gamma^{a_1 a_2 a_3}
    \Big)
    \,+\,
      \tfrac{1}{4!}
      \big(\hspace{.8pt}
        \overline{\theta}
        \,
        \Gamma_{a_1 \cdots a_4}
        \,
        \theta
      \big)
      \,
      \Gamma^{a_1 \cdots a_4}
    \Big)
    \,,
    \hspace{.6cm}
    a_i \in \{0,\!\cdots\!,\!5',6,7,8,9\}
    \,.
  $$
  Moreover, since the only Clifford elements which remain non-vanishing when sandwiched in $\overline{P}(-)P$ are those carrying an odd number of tangential (6d) indices, by \eqref{PropertiesOfProjectionOperator}, this reduces further to  
  \begin{equation}
    \label{IntermediateFormOfThetaBarTheta}
    \def\arraystretch{1.8}
    \begin{array}{ccl}
    \theta \, \overline{\theta}
    &=&
    \tfrac{1}{32}
    \Big(
      \tfrac{1}{3!}
      \big(\hspace{.8pt}
        \overline{\theta}
        \,
        \gamma_{a_1 a_2 a_3}
        \,
        \theta
      \big)
      \,
      \gamma^{a_1 a_2 a_3}
      \,-\,
      \tfrac{1}{3!}
      \big(\hspace{.8pt}
        \overline{\theta}
        \,
        \gamma_{a_1 a_2 a_3} 
        \Gamma_{i}
        \,
        \theta
      \big)
      \,
      \gamma^{a_1 a_2 a_3}
      \Gamma^{i}
    \\
    &&
    \;\;\;\;\;
    \,+\,
      \tfrac{1}{2}
      \big(\hspace{.8pt}
        \overline{\theta}
        \,
        \gamma_{a}
        \Gamma_{i_1 i_2}
        \,
        \theta
      \big)
      \,
      \gamma^{a}
      \Gamma^{i_1 i_2}
      \,-\,
      \tfrac{1}{3!}
      \big(\hspace{.8pt}
        \overline{\theta}
        \,
        \gamma_{a}
        \Gamma_{i_1 i_2 i_3}
        \,
        \theta
      \big)
      \,
      \gamma^{a}
      \Gamma^{i_1 i_2 i_3}
    \Big)
    \end{array}
    \;\;\;
    \def\arraystretch{1.2}
    \begin{array}{l}
      a_i \in \{0,\!\cdots\!,\!5\}
      \\
      i_i \in \{5',6,7,8,9\}
      \mathrlap{\,.}
    \end{array}
  \end{equation}

  But finally, by Hodge duality in the transverse directions
  \vspace{1mm} 
  \begin{equation}
    \label{TransverseHodgeDuality}
    \Gamma_{i_1 i_2 i_3} P
    \underset{
      \scalebox{.7}{
        \eqref{PropertiesOfProjectionOperator}
      }
    }{\;=\;}
    \Gamma_{i_1 i_2 i_3} 
    \Gamma_{5'6789}
    P
    \;=\;
    \pm 
    \tfrac{1}{2}
    \epsilon_{i_1 i_2 i_3 \, i_4 i_5}
    \Gamma^{i_4 i_5} P
    \hspace{.6cm}
    i_i
    \in
    \{5,6,7,8,9\}
    \,,
  \end{equation}
  we have for the last summand in \eqref{IntermediateFormOfThetaBarTheta}:
  $$
    \def\arraystretch{1.6}
    \begin{array}{lll}
      \tfrac{1}{3!}
      \big(\hspace{.8pt}
        \overline{\theta}
        \gamma_a
        \Gamma_{i_1 i_2 i_3}
        \theta
      \big)
      \gamma^a 
      \Gamma^{i_1 i_2 i_3}
      &
      \;=\;
      \tfrac{1}{3!}
      \tfrac{1}{2 \cdot 2}
      \epsilon_{i_1 i_2 i_3 \, i_4 i_5}
      \epsilon^{i_1 i_2 i_3\, j_4 j_5}
      \big(\hspace{.8pt}
        \overline{\theta}
        \gamma_a
        \Gamma^{i_4 i_5}
        \theta
      \big)
      \gamma^a 
      \Gamma_{j_4 j_5}
      &
      \proofstep{
        by
        \eqref{TransverseHodgeDuality}
      }
      \\
     & \;=\;
      \tfrac{1}{3!}
      \tfrac
        {3! \cdot 2!}
        {2 \cdot 2}
      \delta
        _{i_4 i_4}
        ^{j_4 j_5}
      \big(\hspace{.8pt}
        \overline{\theta}
        \gamma_a
        \Gamma^{i_4 i_5}
        \theta
      \big)
      \gamma^a 
      \Gamma_{j_4 j_5}
      &
      \proofstep{
        by 
        \eqref{ContractingKroneckerWithSkewSymmetricTensor}
      }
      \\
      &\;=\;
      \tfrac{1}{2}
      \big(\hspace{.8pt}
        \overline{\theta}
        \gamma_a
        \Gamma^{i_4 i_5}
        \theta
      \big)
      \gamma^a 
      \Gamma_{i_4 i_5}
      &
      \proofstep{
        by
        \eqref{ContractingKroneckerWithSkewSymmetricTensor},
      }
    \end{array}
  $$
  whereby the last two summands in \eqref{IntermediateFormOfThetaBarTheta} cancel each other, and 
  we are left with:
  \begin{equation}
    \label{FinalFormOfThetaBarTheta}
    \def\arraystretch{1.8}
    \begin{array}{ccl}
    \theta \, \overline{\theta}
    &=&
    \tfrac{1}{32}
    \Big(
      \tfrac{1}{3!}
      \big(\hspace{.8pt}
        \overline{\theta}
        \,
        \gamma_{a_1 a_2 a_3}
        \,
        \theta
      \big)
      \,
      \gamma^{a_1 a_2 a_3}
      \,-\,
      \tfrac{1}{3!}
      \big(\hspace{.8pt}
        \overline{\theta}
        \,
        \gamma_{a_1 a_2 a_3} 
        \Gamma_{i}
        \,
        \theta
      \big)
      \,
      \gamma^{a_1 a_2 a_3}
      \Gamma^{i}
      \Big)
      \,,
    \end{array}
    \;\;\;
    \def\arraystretch{1.2}
    \begin{array}{l}
      a_i \in \{0,\!\cdots\!,\!5\}
      \\
      i_i \in \{5',6,7,8,9\}
      \,.
    \end{array}
 \end{equation}

 Now observing (by decomposing the sum and making a simple case analysis) that
 \begin{equation}
   \label{SummedConjugationOf3IndexClifford}
     \gamma_b
     \,
     \gamma_{a_1 a_2 a_3}
     \gamma^b
     \;=\;
     0
     \,,
     \hspace{.6cm}
     a_i, b \in \{0,\!\cdots\!,\!5\}
     \,,
 \end{equation}
 the claim \eqref{6dFierzIdentity} follows:
 $$
   \def\arraystretch{1.6}
   \begin{array}{lll}
     \gamma_a \theta\, 
     \overline{\theta}\gamma^a
     &
     =\;
    \tfrac{1}{32}
    \Big(
      \tfrac{1}{3!}
      \big(
        \overline{\theta}
        \,
        \gamma_{b_1 b_2 b_3}
        \,
        \theta
      \big)
      \,
      \underbrace{
      \gamma_a
      \gamma^{b_1 b_2 b_3}
      \gamma^a
      }_{ \color{gray} = 0 }
      \,+\,
      \tfrac{1}{3!}
      \big(
        \overline{\theta}
        \,
        \gamma_{b_1 b_2 b_3} 
        \Gamma_{i}
        \,
        \theta
      \big)
      \underbrace{
      \gamma_a
      \gamma^{b_1 b_2 b_3}
      \gamma^a
      }_{\color{gray} = 0 }
      \Gamma^i
      \Big)
      &
      \proofstep{
        by
        \eqref{FinalFormOfThetaBarTheta}
      }
      \\[-3pt]
      & =\;
      0
      &
      \proofstep{
        by \eqref{SummedConjugationOf3IndexClifford}.
      }
   \end{array}
 $$

\vspace{-.4cm}
\end{proof}

 \newpage 
 
\subsection{Super $\mathrm{AdS}_7$-spacetime}
\label{SuperAdS}

With the result of \S\ref{ExplicitRheonomy} in hand, we may give explicit formulas for super $\mathrm{AdS}_7 \!\times\! S^4$-spacetime by first recalling the ordinary bosonic AdS-geometry and then rheonomically extending to super-spacetime.

\medskip

\noindent
{\bf Near-horizon geometry of black M5-branes.}
The bosonic near-horizon geometry of $N$ black M5-brane is (cf. \cite[(6.6)]{CKvP98}\cite[\S 2.1.2]{AFHS00}, following \cite{GibbonsTownsend93}\cite{DuffGibbonsTownsend94}) represented on a chart of the form 
\vspace{.1cm} 
\begin{equation}
  \label{ChartNearM5Singularity}
  \mathbb{R}^{1,10} 
    \setminus 
  \mathbb{R}^{1,5}
  \;\;
    \underset{\color{orangeii} \mathrm{diff}}{\simeq}
  \;\;
  \mathbb{R}^{1,5}
  \times
  \big(
    \mathbb{R}^5 \setminus \{0\}
  \big)
  \;\;
    \underset{\color{orangeii} \mathrm{diff}}{\simeq}
  \;\;
  \mathbb{R}^{1,5}
  \times
  \mathbb{R}_{> 0}
  \times 
  S^4
\end{equation}
with its canonical coordinate functions
\begin{equation}
  \label{CanonicalCoordinatesOnPoincareChart}
  \begin{tikzcd}[
    sep=0pt
  ]
    X^a
    &:&
    \mathbb{R}^{1,5}
    \ar[r]
    &[20pt]
    \mathbb{R}
    &[10pt]
    \mbox{for $a \in \{0,1,\cdots, 5\}$}
    \\[-2pt]
    r
    &:&
    \mathbb{R}_{> 0}
    \ar[
      r,
      hook
    ]
    &
    \mathbb{R}
  \end{tikzcd}
\end{equation}
by the $\mathrm{AdS}_7$-metric in ``Poincar{\'e} coordinates'' (cf. \cite[\S 39.3.7]{Blau22})
plus the metric on the round $S^4$:
\vspace{1mm} 
\begin{equation}
  \label{MetricTensorForBlackM5}
  \begin{array}{ccl}
  \mathrm{d}s^2_{N \mathrm{M5}}
  &=&
  \tfrac{r^2}{N^{2/3}}
  \mathrm{d}s^2_{\mathbb{R}^{1,5}}
  \,+\,
  \tfrac{N^{2/3}}{r^2}
  \mathrm{d}r^2
  \;\;+\;\;
  \tfrac{N^{2/3}}{4}
  \,
  \mathrm{d}s^2_{S^4}
  \end{array}
\end{equation}

\vspace{1mm} 
\noindent (where $R_{N\mathrm{M5}}/2 :=  N^{1/3}/2$ is the radius of the 4-sphere in Planck units $2\pi^{1/3}\, \ell_{\! P}$, cf. \eqref{TheEinsteinEquation} below).
So the singular brane locus\footnote{
\label{OnTheSingularity}
  The locus $r = 0$ is not actually a curvature singularity of the near horizon geometry -- as essentially first highlighted by \cite[(3.12)]{GHT95} and manifest below in \eqref{CurvatureTwoForm} -- just a coordinate singularity of the Poincar{\'e} chart \eqref{MetricTensorForBlackM5} --- but it {\it is} a singularity of the  C-field flux $c\, \mathrm{dvol}_{S^4_{N \mathrm{M5}}}$\eqref{CFieldFluxDensityNearM5Singularity} {\it per} unit metric 4-volume $r^4 \, \mathrm{dvol}_{S^4_{N \mathrm{M5}}}$, witnessing $r = 0$ as the necessarily singular source of this flux.
}
 $\simeq \mathbb{R}^{1,5}$ is (or would be) at $r = 0$.
The C-field flux density $G_4$ supporting this is a multiple of the volume form on the $S^4$-factor pulled back to the chart along the projection map:
\begin{equation}
  \label{CFieldFluxDensityNearM5Singularity}
  G_4
  \;:=\;
  c 
  \,
  \mathrm{dvol}_{S^4_{N \mathrm{M5}}}
  \,\in\,
  \Omega^4_{\mathrm{dR}}
  \big(
    S^4
  \big)
  \xhookrightarrow{\phantom{--}}
  \Omega^4_{\mathrm{dR}}
  \big(
    \mathbb{R}^{1,5}
    \times 
    \mathbb{R}_{> 0}
    \times
    S^4
  \big)  
  \,,
\end{equation}

\vspace{2mm} 
\noindent for some prefactor $c$ which is determined, up to its sign, by the Einstein equations, see \eqref{ScaleFactorOfCFieldFlux} below, and determined including its sign by the existence of $\sfrac{1}{2}$BPS M5-immersions, see \eqref{ConditionForPositiveCharge} below.

\medskip

\noindent
For the near-horizon geometry \eqref{CanonicalCoordinatesOnPoincareChart} one says that:
\begin{itemize}[
  leftmargin=.7cm,
  topsep=2pt,
  itemsep=2pt
]
\item $r \to 0$ is the {\it horizon}, cf. footnote \ref{OnTheSingularity};

\item  $r \to \infty$ is the {\it conformal boundary} (e.g. \cite[p. 904]{Blau22}), 

at which 
$
  \underset{r \to \infty}{\mathrm{lim}}
  \big(
    \tfrac{1}{r^2}
    \mathrm{d}s^2_{N \mathrm{M5}}
  \big)
  \;=\;
  \mathrm{d}s^2_{\mathbb{R}^{1,5}}
$
is the Minkowski metric on $\mathbb{R}^{1,5}$ (and zero on $\mathbb{R}_{>0} \times S^4$).
\end{itemize} 
This makes it natural to identify the $\mathbb{R}^{1,5}$-factor at finite $r$ with the worldvolume of a probe M5-brane, to be called a {\it holographic M5-brane} (cf. the terminology of \cite{Gupta21}\cite{Gupta24}):

\smallskip

\noindent
{\bf Chart around a holographic M5-brane embedding.}
We pick a point
$
  s_{\mathrm{prb}} \,\in\, S^4
  \subset \mathbb{R}^5 \setminus \{0\}
$
to designate the direction in which we wish to consider a probe M5-brane worldvolume immersed into this background, at some coordinate distance $r_{\mathrm{prb}}$ from the M5 horizon (cf. \cite[(5.22)]{CKvP98}\cite[\S 8]{GM00} and 
\hyperlink{FigureBraneConfiguration}{Figure B}):
\begin{equation}
  \label{OrdinaryHolographicM5Embedding}
  \begin{tikzcd}[
    column sep=25pt,
    row sep = -2pt
  ]
    \mathllap{
    \scalebox{.7}{
      \color{darkblue}
      \bf
      \def\arraystretch{.9}
      \begin{tabular}{c}
        probe M5
        \\
        worldvolume
      \end{tabular}
    }
    }
    \mathbb{R}^{1,5}
    \ar[
      rr,
      "{ \phi }",
      "{
        \scalebox{.7}{
          \color{darkgreen}
          \bf
          embedding
        }
      }"{swap, yshift=-2pt}
    ]
    &&
    \mathbb{R}^{1,5}
    \times 
    \mathbb{R}_{> 0}
    \times
    S^4
    \\
    x
    &\longmapsto&
    \big(
      x,
      r_{\mathrm{prb}}
      ,
      s_{\mathrm{prb}}
    \big).
  \end{tikzcd}
\end{equation}
Around this point, we may pick a coordinate chart for $S^4$
$$
  \begin{tikzcd}[row sep=small, column sep=large]
    \{0\}
    \ar[r, "\sim"]
    \ar[d, hook]
    &
    \{s_{\mathrm{prb}}\}
    \ar[d,hook]
    \\
    \mathbb{D}^4
    \ar[
      r,
      hook,
      "{ \iota }"
    ]
    &
    S^4
  \end{tikzcd}
$$
on which we find globally defined co-frame forms $(E^i)_{i = 1}^4$ which are orthonormal for the round metric $\mathrm{d}s^2_{S^4}$ on $S^4$
and torsion-free with respect to the corresponding Levi-Civita connection:
\begin{equation}
  \label{CoframeOnChartOf4Sphere}
  \big(
  E^i \in 
  \Omega^1_{\mathrm{dR}}(\mathbb{D}^4)
  \big)_{i = 1}^4
  \,,
  \;\;
  \;\;\;
  \mathrm{such\; that}
  \;\;\;
  \mathrm{d}E^i 
  =
  \big(
    \iota^\ast \Omega_{S^4}^{i j} 
  \big)
  E_j
  \;\;\;
  \mbox{and}
  \;\;\;
  \iota^\ast 
  \mathrm{d}s^2_{S^4}
  \;=\;
  \sum_{i=6}^9
  E^i \otimes E^i
  \,,
\end{equation}
and such that
\begin{equation}
  \label{FourSphereVolumeFormOnChart}
  \iota^\ast 
  \mathrm{dvol}_{S^4_{N \mathrm{M5}}}
  \;=\;
  \tfrac{1}{4!}
  \,
  \epsilon_{i_1 \cdots i_4}
  \,
  E^{i_1} \cdots E^{i_4}
  \,.
\end{equation}
Using this, we obtain a contractible coordinate chart of the near horizon geometry \eqref{ChartNearM5Singularity}:
\begin{equation}
  \label{NeighbourhoodChartOfWorldvolume}
  \begin{tikzcd}[column sep=large]
    \mathbb{R}^{1,5}
    \times
    \mathbb{R}_{> 0 }
    \times
    \mathbb{D}^4
    \ar[
      r,
      hook,
      "{
        \mathrm{id}
        \times
        \iota
      }"
    ]
    &
    \mathbb{R}^{1,5}
    \times
    \mathbb{R}_{> 0 }
    \times
    S^4
    \,.
  \end{tikzcd}
\end{equation}
Since this is a neighborhood of the worldvolume submanifold \eqref{OrdinaryHolographicM5Embedding}, for the purpose of establishing its super-embedding it is sufficient to consider this chart.

\medskip

\noindent
{\bf Cartan geometry around the holographic M5.} On the chart \eqref{NeighbourhoodChartOfWorldvolume},
we evidently have the following coframe forms 
\begin{equation}
  \label{CoFrameAroundHolographicM5}
  \def\arraystretch{1.7}
  \begin{array}{clclcl}
    &
    E^a 
      &:=&
    \tfrac{r}{N^{1/3}}
    \;
    \mathrm{d}X^a
    &
    \scalebox{.8}{tangential}
    &
    a \in \{0,1,2,3,4,5\}
    \\[-11pt]
    \scalebox{.7}{
      \color{darkblue}
      \bf
      AdS
    }
    \\[-11pt]
    &
    E^{5'} 
    &:=&
    \tfrac{N^{1/3}}{r}
    \;
    \mathrm{d} r
    &
    \scalebox{.8}{radial}
    &
    a \in \{5'\}
    \\
    \scalebox{.7}{
      \color{darkblue}
      \bf
      S
    }
    &
    E^a 
    &=&
    \tfrac{N^{1/3}}{2}
    \delta^a_i E^i
    &
    \scalebox{.8}{transversal}
    &
    a \in \{6,7,8,9\}
    \;
    \mbox{via \eqref{CoframeOnChartOf4Sphere}}
    \,,
  \end{array}
\end{equation}
which are orthonormal for the metric \eqref{MetricTensorForBlackM5} in that
$
  \mathrm{d}s^2_{N \mathrm{M5}}
  \;=\;
  \eta_{a b}
  \,
  E^a \otimes E^b
  \,,
$
on this chart
and make the C-field flux density \eqref{CFieldFluxDensityNearM5Singularity}
appear as
\begin{equation}
  \label{FluxDensityOnChart}
  G_4 
  \;=\;
  \tfrac{c}{4!}
  \,
  \epsilon_{i_1 \cdots i_4}
  E^{i_1} \cdots E^{i_4}
  \,,
\end{equation}
for some constant $c$, determined in \eqref{ScaleFactorOfCFieldFlux} below.

\smallskip 
For the following formulas, we may focus on the AdS-factor in \eqref{CoFrameAroundHolographicM5}. Hence we let the indices $a_i, b_i$ run only through $\{0,1,2,3,4,5\}$, to be called the {\it tangential} index values -- namely tangential to the worldvolume \eqref{OrdinaryHolographicM5Embedding} -- with the further \textit{radial} index $5'$ carried along separately.

\smallskip

The torsion-free {\bf spin connection} on the AdS-factor of \eqref{CoFrameAroundHolographicM5},
characterized by
$$
    \mathrm{d}E^a 
    \,=\,
    \Omega^a{}_b\, E^b
    + 
    \Omega^a{}_{5'}\, E^{5'},
    \qquad
    \mathrm{d}E^{5'}
    \,=\,
    \Omega^{5'}{}_a \, E^a
    \,,
$$
is readily seen to have the following as only non-vanishing components:
\vspace{1mm} 
\begin{equation}
  \label{SpinConnection}
  \Omega^{a 5'}
  \,=\,
  -\Omega^{5' a}
  \;=\;
    -
    \tfrac{r}{N^{2/3}}
    \,
    \mathrm{d}X^a
    \;\;
    \scalebox{.8}{
      tangential $a$.
    }
\end{equation}
Therefore its {\bf curvature} 2-form has non-vanishing components
\begin{equation}
  \label{CurvatureTwoForm}
  \def\arraystretch{1.5}
  \begin{array}{ccl}
  R^{a 5'}
  &=&
  \mathrm{d}
  \Omega^{a 5'}
  \;=\;
  -
  \tfrac{1}{N^{2/3}}
  \,
  \mathrm{d}r
  \,
  \mathrm{d}X^a
  \;=\;
  -
  \tfrac{1}{N^{2/3}}
  \,
  E^{5'}\, E^a
  \\
  &=&
  +
  \tfrac{1}{N^{2/3}}
  E^{a}\, E^{5'}
  \\
  R^{a_1 a_2}
  &=&
  -
  \Omega^{a_1}{}_{5'}
  \,
  \Omega^{5' a_2}
  \;=\;
  +
  \tfrac{r^2}{N^{4/3}}
  \mathrm{d}X^{a_1}
  \,
  \mathrm{d}X^{a_2}
  \\
  &=&
  +
  \tfrac{1}{N^{2/3}}
  \,
  E^{a_1} \, E^{a_1}
  \,.
  \end{array}
\end{equation}
Hence the {\bf Riemann tensor}
has non-vanishing components (cf. our normalization of $\delta$ in \ref{KroneckerSymbol})
\begin{equation}
  \label{RiemannTensor}
  \def\arraystretch{1.6}
  \begin{array}{ccl}
  R^{a 5'}{}_{b 5'}
  &=&
  +\tfrac{1}{N^{2/3}}
  \,
  \delta^{a}_b
  \\
  R^{a_1 a_2}{}_{b_1 b_2}
  &=&
  +
  \tfrac{2}{N^{2/3}}
  \,
  \delta^{a_1 a_2}{}_{b_1 b_2}
  \,,
  \end{array}
\end{equation}
and the {\bf Ricci tensor} is proportional to the metric tensor, as befits an Einstein manifold:
\begin{equation}
  \label{RicciTensor}
   \def\arraystretch{1.4}
  \begin{array}{ccl}
  \mathrm{Ric}_{a_1 a_2}
  &=&
  R_{a_1}{}^b{}_{b a_2}
  +
  R_{a_1}{}^{5'}{}_{5' a_2}
  \\
  &=&
  -
  \tfrac{(6-1)}{N^{2/3}}
  \,
  \eta_{a_1 a_2}
  -
  \tfrac{1}{N^{2/3}} 
  \,
  \eta_{a_1 a_2}
  \\
  &=&
  -
  \tfrac{6}{N^{2/3}}
  \,
  \eta_{a_1 a_2}
  \\[+6pt]
  \mathrm{Ric}_{5' 5'}
  &=&
  R^{5'}{}^a{}_{a 5'}
  \\
  &=&
  -
  \tfrac{6}{N^{2/3}}
  \,,
  \end{array}
\end{equation}
similar to the Ricci tensor of the 4-sphere factor (e.g. \cite[Cor. 11.20]{Lee18}):
\begin{equation}
  \label{RicciTensorOfFourSphere}
  \mathrm{Ric}_{i_1 i_2}
  \;=\;
  +
  \tfrac{3}{N^{2/3}/4}
  \, 
  \delta_{i_1 i_2}
  \,.
\end{equation}

Therefore the  {\bf Einstein equation} with source the C-field flux density \eqref{FluxDensityOnChart} has non-vanishing components (cf. \cite[(174-5)]{GSS24-SuGra})
\newpage 
\begin{equation}
  \label{TheEinsteinEquation}
  \def\arraystretch{1.4}
  \begin{array}{cccl}
    &
    \mathrm{Ric}_{a_1 a_2}
    &=&
    -
    \tfrac{1}{12}
    \tfrac{1}{12}
    (G_4)_{i_1 \cdots i_4}
    (G_4)^{i_1 \cdots i_4}
    \,
    \eta_{a_1 a_2}
    \\
    \Leftrightarrow
    &
    -
    \tfrac{6}{N^{2/3}}
    \,
    \eta_{a_1 a_2}
    &=&
    -
    \tfrac{1}{6}
    c^2
    \,
    \eta_{a_1 a_2}
    \\[+9pt]
    &
    \mathrm{Ric}_{5' 5'}
    &=&
    -
    \tfrac{1}{12}
    \tfrac{1}{12}
    (G_4)_{i_1 \cdots i_4}
    (G_4)^{i_1 \cdots i_4}
    \,
    \eta_{5' 5'}
    \\
    \Leftrightarrow
    &
    -
    \tfrac{6}{N^{2/3}}
    &=&
    -
    \tfrac{1}{6}
    c^2
    \\[+9pt]
    &
    \mathrm{Ric}_{i_1 i_2}
    &=&
    \tfrac{1}{12}
    (G_4)_{i_1 \, j_1 j_2 j_3}
    (G_4)_{i_2}{}^{j_1 j_2 j_3}
    \,-\,
    \tfrac{1}{12}
    \tfrac{1}{12}
    (G_4)_{j_1 \cdots j_4}
    (G_4)^{j_1 \cdots j_4}
    \,
    \delta_{i_1 i_2}    
    \\
    \Leftrightarrow
    &
    +
    \tfrac{3}{
      N^{2/3}
      \color{purple}{/4}
    } \, \delta_{i_1 i_2} 
    &=&
    \tfrac{1}{2}
    c^2\, \delta_{i_1 i_2}
    \,-\,
    \tfrac{1}{6}
    c^2 \, \delta_{i_1 i_2}
    \\
    &
    &=&
    +
    \tfrac{1}{3}
    c^2
    \, 
    \delta_{i_1 i_2}
  \end{array}
\end{equation}
thus is solved \footnote{NB: The last line in \eqref{TheEinsteinEquation} is the reason that the radius of $S^4$ has to be {\it half} that of $\mathrm{AdS}_7$ in \eqref{MetricTensorForBlackM5}.}  by
\vspace{2mm} 
\begin{equation}
  \label{ScaleFactorOfCFieldFlux}
  c 
  \;=\; 
  \pm
  \frac{6}
  {N^{1/3}
  }
  \,,
  \hspace{.6cm}
  \underset{
    \scalebox{.7}{
      \eqref{CFieldFluxDensityNearM5Singularity}
    }
  }{\mbox{hence}} 
  \hspace{.5cm}
  G_4
  \;=\;
  \pm 
  \tfrac{6}{N^{1/3}}
  \mathrm{dvol}_{S^4_{N \mathrm{M5}}}
  \,.
\end{equation}
At this point, both of the signs in \eqref{ScaleFactorOfCFieldFlux} are equally admissible, but we see below in Rem. \ref{CriticalDistanceForHolographicM5} that the + sign is singled out by the existence of holographic M5-brane embedding.

\medskip

\noindent
{\bf Super-Cartan geometry near M5 horizons.} In now passing to the super-spacetime enhancement of $\mathrm{AdS}_7 \times S^4$, we use the notation and conventions for 6d spinors among 11d spinors form 
\cite[\S 3.2]{GSS24-FluxOnM5}, recalled in \S\ref{SpinorsOnM5Branes}.
In particular, we denote the Minkowski frame of Clifford generators adapted to the 1+5+1+4 dimensional split of the tangent space to $\mathrm{AdS}_7 \times S^4$ in Poincar{\'e} coordinates \eqref{ChartNearM5Singularity} by
\cite[(85)]{GSS24-FluxOnM5}
\begin{equation}
  \label{AdaptedCliffordGenerators}
  \hspace{-6.8cm}
  \begin{tikzcd}[
    column sep=3pt,
    row sep=0pt
  ]
    &&
    \mathclap{
      \hspace{20pt}
      \overbrace{
      \phantom{-------------}
      }^{
        \mathbb{R}^{1,5}
      }
    }
    &&&&
    \overbrace{
      \phantom{--}
    }^{\mathbb{R}_{>0}}
    &&
    \mathclap{
      \hspace{15pt}
      \overbrace{
      \phantom{\phantom{--------}}
      }^{
        \mathbb{D}^{4}
      }
    }
    \\[-10pt]
    \Gamma_0
    \ar[
      d,
      |->,
      start anchor={[yshift=+1pt]},
      end anchor={[yshift=-1pt]},
      bend right=90,
      shift right=0pt,
      "{\color{darkgreen} \bf 
        \def\arraystretch{.9}
        \begin{array}{l}
          \phantom{+}\overline{P}(-)P
          \\
          +
          P(-)\overline{P}
        \end{array}
      }"{swap, pos=.5}
    ]
    &
    \Gamma_1
    &
    \Gamma_2
    &
    \Gamma_3
    &
    \Gamma_4
    &
    \Gamma_5
    &
    \Gamma_{5'}
    &
    \Gamma_{6}
    &
    \Gamma_{7}
    &
    \Gamma_{8}
    &
    \Gamma_{9}
    &
    \mathrlap{
       \in 
      \mathrm{Pin}^+(1,0)
      \,\subset\,
      \mathrm{End}_{\mathbb{R}}(\mathbf{32})
    }
    \\
    \gamma_0
    &
    \gamma_1
    &
    \gamma_2
    &
    \gamma_3
    &
    \gamma_4
    &
    \gamma_5
    &
    &&&&&
    \mathrlap{
    \in 
    \mathrm{Pin}^+(1,5)
    \,\subset\,
    \mathrm{End}_{\mathbb{R}}
    \big(
      2\cdot \mathbf{8}_+
      \oplus
      2 \cdot \mathbf{8}_-
    \big)
    \,,
    }
  \end{tikzcd}
\end{equation}
where \cite[(86-90)]{GSS24-FluxOnM5}
\begin{equation}
  \label{TheHalfProjectionOperator}
  \def\arraystretch{1.3}
  \begin{array}{l}
  P \Psi
  \;:=\;
  \tfrac{1}{2}
  \big(
    1 + 
    \Gamma_{5'6789}
  \big)
  \Psi
  \\
  \overline{P} \Psi
  \;:=\;
  \tfrac{1}{2}
  \big(
    1 -
    \Gamma_{5'6789}
  \big)
  \Psi
  \,,
  \end{array}
  \hspace{.7cm}
  \def\arraystretch{1.3}
  \begin{array}{l}
    P(\mathbf{32})
    \;\simeq\;
    2\cdot \mathbf{8}_+
    \,\in\,
    \mathrm{Rep}_{\mathbb{R}}\big(
      \mathrm{Spin}(1,5)
    \big)
    \\
    \overline{P}(\mathbf{32})
    \;\simeq\;
    2\cdot \mathbf{8}_-
    \,\in\,
    \mathrm{Rep}_{\mathbb{R}}\big(
      \mathrm{Spin}(1,5)
    \big)
    \,.
  \end{array}
\end{equation}

\medskip

\noindent
{\bf Super-Cartan geometry around holographic M5s.}
We now obtain the super-extension of the above Cartan geometry \eqref{CoFrameAroundHolographicM5}.
Inserting the bosonic AdS Cartan geometry \eqref{CoFrameAroundHolographicM5} \eqref{SpinConnection} into the initial conditions for WZT gauge \eqref{GaugeConditions}
means that
\begin{equation}
  \label{InitialValuesInAdS}
  \def\arraystretch{1.9}
  \begin{array}{lcrcll}
    \big(
      E^{(0)}
    \big)^a
    &=&
    \tfrac{r}{N^{1/3}}
    \,
    \mathrm{d}X^a
    &
      \;\;\;\;\;
      \Leftrightarrow
      \;\;\;\;\;
    &
    \Big(
      \big(
        E^{(0)}
      \big)^a_r \,=\,
      \tfrac{r}{N^{1/3}}
      ,
      &
      \big(
        E^{(0)}
      \big)^a_\rho \,=\,
      0
      \;    
    \Big)
    \\
    \big(
      E^{(0)}
    \big)^{5'}
    &=&
    \tfrac{N^{2/3}}{r}
    \,
    \mathrm{d}X^{5'}
    &
      \;\;\;\;\;
      \Leftrightarrow
      \;\;\;\;\;
    &
    \Big(
      \big(
        E^{(0)}
      \big)^{5'}_r \,=\,
      \tfrac{N^{2/3}}{r}
      ,
      &
      \big(
        E^{(0)}
      \big)^{5'}_\rho \,=\,
      0
      \;
    \Big)
    \\
    \big(
      \Psi^{(0)}
    \big)^\alpha
    &=&
    \mathrm{d}\Theta^\alpha
    &
      \;\;\;\;\;
      \Leftrightarrow
      \;\;\;\;\;
    &
    \Big(
      \big(
        \Psi^{(0)}
      \big)^\alpha_r
      \,=\,
      0
      ,
      &
      \big(
        \Psi^{(0)}
      \big)^\alpha_\rho
      \,=\,
      \delta^\alpha_\rho
    \Big)
    \\
    \big(
      \Omega^{(0)}
    \big)^{5' a}
    &=&
    \tfrac{r}{N^{2/3}} 
    \, 
    \mathrm{d}X^a
    &
      \;\;\;\;\;
      \Leftrightarrow
      \;\;\;\;\;
    &
    \Big(
      \big(
        \Omega^{(0)}
      \big)^{5' a}_r
      \,=\,
      \tfrac{r}{N^{2/3}}
      ,
      &
      \big(
        \Omega^{(0)}
      \big)^{5' a}_\rho
      \,=\,
      0\;
      \Big)
      \,.
  \end{array}
\end{equation}
Moreover, inserting the flux density \eqref{FluxDensityOnChart} into the super-field strength components \eqref{HK} yields
\begin{equation}
  \label{OddComponentsOfFieldStrengthsOnAdS}
  \def\arraystretch{1.4}
  \begin{array}{ccl}
    H_a
    &=&
    -
    \tfrac{c}{12}
    \Gamma_a
    \,
    \Gamma_{6789}
    \\
    H_{5'}
    &=&
    -
    \tfrac{c}{12}
    \,
    \Gamma_{5'6789}
    \\
    H_i 
    &=&
    \tfrac{c}{6}
    \tfrac{1}{3!}
    \,
    \epsilon_{i \, i_1 i_2 i_3}
    \Gamma^{i_1 i_2 i_3}
    \\[6pt]
    K^{a_1 a_2}
    &=&
    -
    \tfrac{c}{6}
    \Gamma^{a_1 a_2}
    \Gamma_{6789}
    \\
    K^{5' a}
    &=&
    +
    \tfrac{c}{6}
    \,
    \Gamma^{a}
    \Gamma_{5'6789}
    \\
    K^{i_1 i_2}
    &=&
    -\tfrac{c}{6}
    \epsilon^{ i_1 i_2\, i_3 i_4 }
    \Gamma_{i_3 i_4}
    \\
    K^{i a}
    &=&
    0
    \\
    K^{5' i}
    &=&
    0\,.
  \end{array}
  \hspace{1.5cm}
  \mbox{for}
  \hspace{1cm}
  \def\arraystretch{1.2}
  \begin{array}{l}
    a_i \in \{0,1,2,3,4,5\}
    \\
    \,i_i \in \{6,7,8,9\}
  \end{array}
\end{equation}
From this, we now obtain the super-field extension of the supergravity fields on $\mathrm{AdS}_7 \times S^4$. 


\begin{example}[\bf $\mathrm{AdS}_7 \times S^4$ super-fields to first $\Theta$-order]
\label{SuperFieldsToFirstThetaOrder}
Based on the 0th-order expressions \eqref{InitialValuesInAdS},
we obtain to first order in $\Theta$ (similar to \cite[(3.11)]{dWPPS98}):
\begin{equation}
  \label{SuperfieldsToFirstOrder}
  \hspace{-5mm} 
  \def\arraystretch{1.8}
  \begin{array}{lcrcrcrlll}
  E^a 
  &=&
  \tfrac{r}{N^{1/3}}
  \,
  \mathrm{d}X^a
  &+&
  \big(\,
    \overline{\Theta}
    \,\Gamma^a\,
    \mathrm{d}\Theta
  \big)
  &+&
  \mathcal{O}\big(\Theta^2\big)
  &
  \proofstep{
    by \eqref{RecursionRelationForGraviton}
  }
  \\
  E^{5'}
  &=&
  \tfrac{N^{\fixed{2}/3}}{r}
  \,
  \mathrm{d}X^{5'}
  &+&
  \big(\,
    \overline{\Theta}
    \,\Gamma^{5'}\,
    \mathrm{d}\Theta
  \big)
  &+&
  \mathcal{O}\big(\Theta^2\big)
  &
  \proofstep{
    by \eqref{RecursionRelationForGraviton}
  }
  \\
  \Omega^{5' a}
  &=&
  \tfrac{r}{N^{2/3}}
  \,
  \mathrm{d}X^a
  \,
  &+&
  \tfrac{c}{6}
  \,
  \big(\,
    \overline{\Theta}
    \,
    \Gamma^{a}
    \,
    \Gamma_{5'6789}
    \,
    \mathrm{d}\Theta
  \big)
  &+&
  \mathcal{O}\big(\Theta^2\big)
  &
  \proofstep{
    by
    \eqref{RecursionForSpinConnection},
    \eqref{OddComponentsOfFieldStrengthsOnAdS}
    \&
    \eqref{GaugeConditions}
  }
  \\
  \Psi^\alpha
  &=&
  \mathrm{d}\Theta^\alpha
  &+&
  \big(
    -
    \tfrac{1}{2}
    \tfrac{ r }{ N^{2/3} }
    (
      \Gamma_{5' a}
      \Theta
    )^\alpha
    -
    \tfrac{c}{12}
    \fixed{
      \tfrac{r}{N^{1/3}}
    }
    (\Gamma_a \Gamma_{6789}\Theta)^\alpha
  \big)
  \,
  \mathrm{d}X^a
  \\
  &&
  &+&
  -\tfrac{c}{12}
  \fixed{
    \tfrac
      {N^{2/3}}
      {r}
  }
  (
    \Gamma_{5'6789}
    \Theta
  )^\alpha
  \;
  \mathrm{d}X^{5'}
  \\
  &&
  &+&
  \tfrac{c}{6}
  \tfrac{1}{3!}
  \epsilon_{i \, i_1 i_2 i_3}
  (
    \Gamma^{i_1 i_2 i_3}
    \Theta
  )^\alpha
  \,
  \mathrm{d}X^i
  &+&
  \mathcal{O}\big(
    \Theta^2
  \big)
  &
  \proofstep{
    by
    \eqref{RecursionRelationForGravitino}
    \&
    \eqref{OddComponentsOfFieldStrengthsOnAdS},
  }
  \end{array}
\end{equation}
where $a \in \{0,\!\cdots\!,\!5\}$.
\end{example}


\subsection{Holographic M5 immersion}
\label{TheHolographicM5Immersion}

With the background super-spacetime in hand (\S\ref{SuperAdS}), we are ready to inspect the $\sfrac{1}{2}$BPS super-immersions of holographic M5-branes.

\smallskip

Our main result here is Thm. \ref{ExistenceOfHolographicM5BraneProbes}, which says that the evident super-immersion of an M5-brane worldvolume into the Minkowski-part of the Poincar{\'e} chart of the near-horizon super-geometry of $N$ black M5-branes is $\sfrac{1}{2}$BPS (hence is a ``super-embedding'').

\medskip 
\noindent
{\bf $\sfrac{1}{2}$BPS super-immersions.}
Recall (e.g. \cite[p. 27]{Varadarajan04}, cf. \cite[Rem. 2.10, Def. 2.18]{GSS24-FluxOnM5}) that:
\begin{definition}[\bf Super-immersions]
\label{SuperImmersion}
A map of supermanifolds 
(e.g. \cite[Ex. 2.13]{GSS24-SuGra})
\begin{equation}
  \label{GeneralSuperImmersion}
  \begin{tikzcd}
    \scalebox{.7}{
      \color{darkblue}
      \bf
      \def\arraystretch{.7}
      \begin{tabular}{c}
        super-
        \\
        worldvolume
      \end{tabular}
    }
    \Sigma^{1,p \,\vert\, \mathbf{n}}
    \ar[
      rrr,
      "{
        \phi
      }",
      "{
        \scalebox{.7}{
          \color{darkgreen}
          \bf
          immersion
        }
      }"{swap}
    ]
    &&&
    X^{1,d \,\vert\, \mathbf{N} }
    \scalebox{.7}{
      \color{darkblue}
      \bf
      \def\arraystretch{.7}
      \begin{tabular}{c}
        super-
        \\
        spacetime
      \end{tabular}
    }
  \end{tikzcd}
\end{equation}
is a {\it super-immersion}
if it induces injections on all super-tangent spaces 
$$
  \underset{
    \scalebox{.7}{$
    \sigma \in \bosonic{\Sigma}
    $}
  }{\forall}
  \hspace{.8cm}
  \begin{tikzcd}
    \mathbb{R}^{1,p \vert \mathbf{n}}
    \ar[
      rr,
      hook,
    ]
    \ar[
      d,
       "{ \sim }"{sloped}
    ]
    &&
    \mathbb{R}^{1,d \vert \mathbf{N}}
    \ar[d,"{ \sim }"{sloped}]
    \\
    T_{\sigma}\Sigma^{1,p \,\vert\, \mathbf{n}}
    \ar[
      rr,
      "{
        \mathrm{d}\phi_\sigma
      }"
    ]
    &&
    T_{\phi(\sigma)}
    X^{1,d \,\vert\, \mathbf{N} }
    \,.
  \end{tikzcd}
$$
\end{definition}
We say, following  \cite[\S 2.2]{GSS24-FluxOnM5}, that:
\begin{definition}[\bf $\sfrac{1}{2}$BPS super-immersions]
\label{HalfBPSSuperImmersion}
A super-immersion $\phi$ \eqref{GeneralSuperImmersion} is {\it $\sfrac{1}{2}$BPS} if for a linear projection operator $P$
from the target super-space onto the ``{\it tangential}\,'' worldvolume super-dimensions (with $\overline{P} := 1-P$ the ``{\it transversal}\,'' projection), projecting onto the fixed locus of a $\mathrm{Pin}^{+}(1,d)$-element (a {\it $p$-brane involution} \cite[Def. 4.4]{HSS19})

\vspace{-.3cm}
\begin{equation}
  \label{GeneralSuperProjectionOperator}
  \adjustbox{
    raise=-6pt
  }{
   \begin{tikzcd}
    \mathbb{R}
      ^{
      1,d \,\vert\, \mathbf{N}
      }
    \ar[
      r,
      ->>
    ]
    \ar[
      rr,
      rounded corners,
      to path={
           ([yshift=+00pt]\tikztostart.north)  
        -- ([yshift=+08pt]\tikztostart.north)
        -- node[]{
          \scalebox{.8}{
            \colorbox{white}{
              $P$
            }
          }
        }
           ([yshift=+08pt]\tikztotarget.north)
        -- ([yshift=-00pt]\tikztotarget.north)
      }
    ]
    &
    \mathbb{R}
      ^{1,p \,\vert\, \mathbf{n}}
    \ar[
      r,
      hook
    ]
    &
    \mathbb{R}^{
      1,d \,\vert\, \mathbf{N}
    }
    \,,
  \end{tikzcd}
  }
\end{equation}
there exists an orthonormal local co-frame field $(E,\Psi)$ \eqref{LocalCartanConnection} on $X$ which is 
{\it super-Darboux} with respect to $\phi$ in that:

\begin{itemize}[
  leftmargin=.8cm,
  topsep=1pt,
  itemsep=2pt
]
\item[\bf (i)]
 the tangential coframe pulls back to a local coframe field on $\Sigma$:
 \begin{equation}
   \label{PullbackOfTangentialCoframeBeingCoframe}
   (e,\psi)
   \;:=\;
   \phi^\ast\big(
     PE,\, P\Psi
   \big)
   \mathrlap{
   \;\;\;
   \mbox{
     is a coframe field
   }
   }
 \end{equation}

\item[\bf (ii)] the transversal bosonic coframe field pulls back to zero 
\begin{equation}
  \label{PullbackOfTransversalBosonicCoframeVanishes}
  \phi^\ast \overline{P} E
  \;=\;
  0
\end{equation}

\item[\bf (iii)] the transversal fermionic coframe field pulls back to
\begin{equation}
  \label{ExistenceOfShearMap}
  \phi^\ast 
  \overline{P}\Psi
  \;=\;
  \mathrm{Sh}
  \cdot \psi
\end{equation}
for some {\it fermionic shear} field $\mathrm{Sh}$ on $\Sigma$, i.e. pointwise valued in $\mathrm{Spin}(d-p)$-equivariant linear maps
\vspace{1mm} 
\begin{equation}
  \label{EquivarianceOfShearMap}
  \underset{
    \scalebox{.7}{$
      \sigma
      \in \bosonic{\Sigma}
    $}
  }{\forall}\hspace{.4cm}
  \mathrm{Sh}_\sigma
  \;:\;
  \mathbf{n} \simeq P\mathbf{N}
  \xrightarrow{\quad}
  \overline{P}\mathbf{N}
  \,.
\end{equation}
\end{itemize}
\end{definition}
\begin{example}[\bf M5 super-immersions]
  If the projection operator \eqref{GeneralSuperProjectionOperator} is that from \eqref{SuperProjectionOperator},
  then we have the case of {\it M5-brane super-immersions} (\cite[\S 3]{GSS24-FluxOnM5}, going back to \cite{HoweSezgin97b}).
\end{example}

Remarkably (cf. Rem. \ref{AbsenceOfFluxedHolographicM5branes} below), the shear map \eqref{EquivarianceOfShearMap} turns out to encode the flux density of any higher gauge field on the worldvolume $\Sigma$. If this vanishes (as it does in the example holographic M5-branes presented in a moment) the definition simplifies to:

\begin{definition}[\bf Fluxless $\sfrac{1}{2}$BPS super-immersion]
\label{FluxlessHalfBPSImmersion}
A $\sfrac{1}{2}$BPS super-immersion (Def. \ref{HalfBPSSuperImmersion})
is {\it fluxless} if
its super-Darboux coframes $(E,\Psi)$ are characterized more simply by
\begin{equation}
  \label{ConditionOnFluxlessBPSImmersion}
  \def\arraystretch{1.4}
  \begin{array}{rccl}
  \scalebox{.8}{
    \color{darkblue}
    \bf
    Tangential condition:
  }
  &
  (e,\psi)
  &:=&
  \phi^\ast\big(
    P E
    ,\,
    P\Psi
  \big)
  \;\;
  \mbox{
    is a coframe field
  }
  \\
  \scalebox{.8}{  
    \color{darkblue}
    \bf
    Transversal condition:
  }
  &
  0
  &=&
  \phi^\ast\big(\hspace{.7pt}
    \overline{P}E
    ,\,
    \overline{P}\Psi
  \big)
  \,.
  \end{array}
\end{equation}
\end{definition}
This is manifestly super-analogous to classical Darboux coframe theory (recalled in \cite[\S 2.1]{GSS24-FluxOnM5}) and this is what we establish for holographic M5-branes in Thm. \ref{ExistenceOfHolographicM5BraneProbes} below.

\begin{remark}[\bf Relation to the literature]
\label{RelationOfSuperembeddingToTheLiterature}
$\,$

  \noindent {\bf (i)} The conditions \eqref{PullbackOfTangentialCoframeBeingCoframe} and \eqref{PullbackOfTransversalBosonicCoframeVanishes} on a $\sfrac{1}{2}$BPS super-immersion are (for more details see \cite[Rem. 2.23]{GSS24-FluxOnM5}) a slight strengthening of the ``super-embedding'' condition used by \cite{Sorokin00}, following \cite{BPSTV95}\cite{HoweSezgin97a}\cite{HoweSezgin97b}\cite{HRS98}. 

  \vspace{1mm} 
    \noindent {\bf (ii)} In particular, $(e,\psi)$ being a super-coframe field \eqref{PullbackOfTangentialCoframeBeingCoframe}
  entails that $\phi^\ast P E =: e$ has no component along $\psi$, which is the ``basic super-embedding condition'' of \cite[(6)]{HoweSezgin97a}\cite[(2)]{HRS98}, earlier known as the ``geometrodynamical condition'' \cite[(2.23)]{BPSTV95}.

  \vspace{1mm} 
  \noindent {\bf (iii)}  The difference is that more generally one may allow the pullback of the transversal gravitino to have also a bosonic component $\tau$ (cf. \cite[Rem. 3.13]{GSS24-FluxOnM5}), generalizing \eqref{ExistenceOfShearMap} to
  \vspace{1mm} 
  \begin{equation}
    \label{TheTauComponent}
    \phi^\ast
    \overline{P}\Psi
    \;=\;
    \mathrm{Sh} \cdot \psi
    \,+\,
    \tau_a e^a
    \,.
  \end{equation}
  However, only for $\tau = 0$:
  \begin{itemize}[
    leftmargin=.8cm,
    topsep=1pt,
    itemsep=2pt
  ]
  \item[\bf (a)] does the worldvolume Bianchi identity $\mathrm{d}\, H_3 \,=\, \phi^\ast G_4$
  actually follow without fermionic corrections \cite[Rem. 3.19]{GSS24-FluxOnM5},
  \item[\bf (b)] is the super-embedding really $\sfrac{1}{2}$BPS \cite[Rem. 3.13]{GSS24-FluxOnM5}.
  \end{itemize}
    In any case, for the example of holographic M5-branes obtained in Thm. \ref{ExistenceOfHolographicM5BraneProbes} below we do have $\tau = 0$ (and no other explicit examples seem to have been discussed in the literature before).
\end{remark}

\medskip

\noindent
{\bf The holographic M5 super-immersion.} We may now define and analyze the super-geometric enhancement of the immersion of M5-worldvolumes parallel and near to the horizon of their own black brane incarnation (cf. again \hyperlink{FigureBraneConfiguration}{Figure B}):

\begin{definition}[\bf Holographic super-immersion]
\label{HolograhicSuperImmersion}
We extend the holographic immersion \eqref{OrdinaryHolographicM5Embedding} to a super-immersion (Def. \ref{SuperImmersion}) 
in an evident way:
\begin{equation}
  \label{TheSuperImmersion}
  \begin{tikzcd}[
    column sep=-5pt,
    row sep=-6pt
  ]
    \mathbb{R}^{
      1,5 \,\vert\, 2 \cdot \mathbf{8}
    }
    \ar[
      rr,
      "{\phi}"
    ]
    &\hspace{30pt}&
    \mathbb{R}^{1,5}
    &\times& 
    \mathbb{R}_{> 0}
    &\times&
    \mathbb{D}^4
    &\times& 
    \mathbb{R}^{0 \vert 2 \cdot \mathbf{8}_+}
    &\times& 
    \mathbb{R}^{0 \vert 2 \cdot \mathbf{8}_-}
    &
    \;\;\;\;\;
    \subset
    \; \mathbb{R}^{1,10 \,\vert\, \mathbf{32}}
    \\
    x^a
    &\overset{\phi^\ast}{\longmapsfrom}&
    X^a
    \\
    r_{\mathrm{prb}}
    &\overset{\phi^\ast}{\longmapsfrom}&
    &&
    r
    \\
    s^i_{\mathrm{prb}}
    &\overset{\phi^\ast}{\longmapsfrom}&
    &&
    &&
    X^i
    \\
    \theta^\alpha
    &\overset{\phi^\ast}{\longmapsfrom}&
    &&
    &&
    &&
    \big(P\Theta\big)^\alpha
    \\
    0
    &\overset{\phi^\ast}{\longmapsfrom}&
    &&
    &&
    &&
    &&
    \big(
      \overline{P}\Theta
    \big)^\alpha
    \,,
  \end{tikzcd}
\end{equation}
where $P = \tfrac{1}{2}\big(1 + \Gamma_{5'6789}\big)$, cf. \eqref{ProjectorOn6dSpinRep},
which defines the super-coordinates on the worldvolume to be the projected pullbacks of those of target space,
and where 
\vspace{1mm} 
$$
  r_{\mathrm{prb}}
  ,\,
  s^i_{\mathrm{prb}}
  \;\in\;
  \mathbb{R}
  \xhookrightarrow{\quad}
  C^\infty(\mathbb{R}^{1,5})
  \xhookrightarrow{\quad}
  C^\infty(\mathbb{R}^{1,5 \,\vert\, 2 \cdot \mathbf{8}})
  \hspace{.6cm}
  \mbox{for}
  \hspace{.6cm}
  i \in \{6,7,8,9\}
$$
are the chosen constants parametrizing the transverse position of the immersion (cf. \hyperlink{FigureBraneConfiguration}{Figure B}).
\end{definition}

\newpage 
\begin{lemma}[\bf M5-Worldvolume super-fields to first $\theta$-order]
  \label{WorldvolumeSuperFieldsToFirstOrder}
 $\,$
  
\noindent {\bf (i)} Under the holographic super-immersion \eqref{TheSuperImmersion}, the first-order super-fields \eqref{SuperfieldsToFirstOrder} pull back to
  \begin{equation}
    \label{WorldvolumeFieldsToFirstThetaOrder}
    \hspace{-6mm}
    \def\arraystretch{1.6}
    \begin{array}{rcccrcrcl}
      e^a 
      &:=&
      \phi^\ast E^a
      &=&
      \tfrac
        { r_{\mathrm{prb}} }
        { N^{1/3} }
      \,
      \mathrm{d}x^a
      &+&
      \big(\hspace{.8pt}
        \overline{\theta}
        \,\gamma^a\,
        \mathrm{d}\theta
      \big)
      &+&
      \mathcal{O}\big(
        \theta^2
      \big)
      \\
      &&
      \phi^\ast
      E^{5'}
      &=&
      &&
      &&
      \mathcal{O}\big(\theta^2\big)
      \\
      \psi^\alpha
      &:=&
      \phi^\ast (P \Psi)^\alpha
      &=&
      &&      
      \mathrm{d}\theta^\alpha
      &+&
      \mathcal{O}\big(
        \theta^2
      \big)
      \\
      &&
      \phi^\ast\big(
        \hspace{.7pt}
        \overline{P}\Psi
      \big)^\alpha
      &=&
      \big(
        \tfrac{1}{2}
        \tfrac
          {r_{\mathrm{prb}}}
          { N^{2/3} }
        -
        \tfrac{c}{12}
        \fixed{
          \tfrac
            {r_{\mathrm{prb}}}
            {N^{1/3}}
        }
      \big)
      (
        \Gamma_{a5'}
        \theta
      )^\alpha
      \,
      \mathrm{d}x^a
      &&
      &+&
      \mathcal{O}\big(
        \theta^2
      \big)
      \\
      e^{b_1}
      \SecondFundamentalForm
        ^{ 5' }_{ b_1 b_2 }
      \,+\,
      \psi^\beta
      \SecondFundamentalForm
        ^{ 5' }_{ \beta\,  b_2 }
      &:=&
      \phi^\ast
      \Omega^{5'}{}_{b_2}
      &=&
      \tfrac
        { r_{\mathrm{prb}} }
        { N^{2/3} }
      \delta_{b_1 b_2}
      \,
      \mathrm{d}x^{b_1} 
      &+&
      \tfrac{1}{6}\big(\hspace{.8pt}
        \overline{\theta}
        \,\gamma^a\,
        \mathrm{d}\theta
      \big)
      &+&
      \mathcal{O}\big(\theta^2\big)
      \mathrlap{\,.}
    \end{array}
  \end{equation}
 \noindent {\bf (ii)} The 2nd fundamental super-form
  $\SecondFundamentalForm^{5'}$
  {\rm (cf. \cite[(67)]{GSS24-FluxOnM5})}
  has the following components:
  \begin{equation}
    \label{ScndFundamentalFormToFirstOrder}
    \def\arraystretch{1.5}
    \begin{array}{cclcl}
      \SecondFundamentalForm
        ^{5'}
        _{b_1 b_2}
      &=&
      \tfrac{1}{N^{1/3}}
      \delta_{b_1 b_2}
      &+&
      \mathcal{O}
      \big(
        \theta^2
      \big)
      \\
      \SecondFundamentalForm
        ^{5'}
        _{\beta \, b_2}
      &=&
      \big(
        \tfrac{1}{N^{1/3}}
        \,-\,
        \tfrac{c}{6}
      \big)
      \big(\hspace{.8pt}
        \overline{\theta}\,\gamma_{b_1}
      \big)_\beta
      &+&
      \mathcal{O}
      \big(
        \theta^2
      \big)
      \,.
    \end{array}
  \end{equation}
\end{lemma}
\begin{proof}
  The first line in \eqref{WorldvolumeFieldsToFirstThetaOrder} is evident. For the second line just note that $\big(\overline{\theta}\,\Gamma^{5'}\,\mathrm{d}\theta\big) = 0$ by \eqref{PropertiesOfProjectionOperator}.
  For the third line, notice similarly that
  $$
    \def\arraystretch{1.6}
    \begin{array}{l}
      \phi^\ast (P \Psi)^\alpha
      \;=\;
      \mathrm{d}\theta^\alpha
      \,+\,
      \Big(
        -\tfrac{1}{2}
        \tfrac
          { r_{\mathrm{prb}} }
          { N^{2/3} }
        \underbrace{
          (P\,\Gamma_{5' a} \theta)
            ^{\mathrlap{\alpha}}
        }_{\color{gray}  = 0}
        \;\;
        -
        \tfrac{c}{12}
        \fixed{
          \tfrac{
            r_{\mathrm{prb}}
          }{N^{1/3}}
        }
        \underbrace{
          (
            P\,
            \Gamma_a
            \Gamma_{6789}
            \theta
          )^{\mathrlap{\alpha}}
        }_{\color{gray} = 0 }
        \;\;
      \Big)
      \,
      \mathrm{d}x^a
      \;\;
      +
      \mathcal{O}(\theta^2)
      \mathrlap{\,,}
    \end{array}
  $$

  \vspace{1mm}
  \noindent
  where the terms over the braces vanish by \eqref{PropertiesOfProjectionOperator}: 
  $$
    P \Gamma_{5' a} \theta
    \;=\;
    \Gamma_{5' a}
    \overline{P} \theta
    \;=\;
    \;=\;
    \Gamma_{5' a}
    \overline{P} P \theta
    \;=\;
    0
    \hspace{.6cm}
    \mbox{and}
    \hspace{.6cm}
    P \Gamma_a \Gamma_{6789} \theta
    \;=\;
    \Gamma_a \Gamma_{6789} 
    \overline{P} \theta
    \;=\;
    \Gamma_a \Gamma_{6789} 
    \overline{P} P\theta
    \;=\;
    0
    \,.
  $$
  The fourth line works analogously but complementarily:
  \begin{equation}
    \label{VanishingTransveralPsiToFirstOrder}
    \def\arraystretch{1.5}
    \begin{array}{ccl}
      \phi^\ast 
      \big(\hspace{.7pt}
        \overline{P} \Psi
      \big)^\alpha
      &=&
      \Big(\!
        -\tfrac{1}{2}
        \tfrac
          { r_{\mathrm{prb}} }
          { N^{2/3} }
        \underbrace{
          \big(\hspace{.7pt}
            \overline{P} \, \Gamma_{5' a} \theta
          \big)
            ^{\mathrlap{\alpha}}
        }_{\color{gray} 
          - \Gamma_{a5'}\theta
        }
        \;\;
        -
        \tfrac{c}{12}
        \fixed{
          \tfrac{
            r_{\mathrm{prb}}
          }{N^{1/3}}
        }
        \underbrace{
          \big(\hspace{.7pt}
            \overline{P}\, 
            \Gamma_a
            \Gamma_{6789}
            \theta
          \big)^{\mathrlap{\alpha}}
        }_{\color{gray} 
          \Gamma_{a 5'}\theta
        }
        \;\;
      \Big)
      \,
      \mathrm{d}x^a
      \;\;
      +
      \mathcal{O}(\theta^2)
      \,,
    \end{array}
  \end{equation}

  \vspace{.1cm}
  \noindent
  where under the braces we again used \eqref{PropertiesOfProjectionOperator}:
  \vspace{1mm} 
  $$
    \overline{P}\,\Gamma_{5'a}\theta
    \;=\;
    \Gamma_{5'a}P\theta
    \;=\;
    \Gamma_{5' a}\theta
    \;=\;
    -\Gamma_{a 5'}\theta
    \hspace{.7cm}
    \mbox{and}
    \hspace{.7cm}
    \overline{P}\,\Gamma_a \Gamma_{6789}
    \theta
    \;=\;
    \Gamma_a \Gamma_{6789}
    P\theta
    \;=\;
    \Gamma_{a 5'} \theta
    \,.
  $$

  \vspace{1mm} 
\noindent  Finally, the fifth line follows again similarly, now using that $\Gamma_{5'6789} P = P$, again by \eqref{PropertiesOfProjectionOperator}.
  From this, 
  the last statement \eqref{ScndFundamentalFormToFirstOrder} is checked by expanding out:
  $$
  \hspace{-2mm} 
    \def\arraystretch{1.6}
    \begin{array}{llll}
      e^{b_1}
      \, 
      \SecondFundamentalForm
        ^{5'}
        _{ b_1 b_2 }
      \,+\,
      \psi^\beta
      \,
      \SecondFundamentalForm
        ^{5'}
        _{ \beta b_2 }
      &
      \;=\;
      \Big(
        \tfrac{r_{\mathrm{prb}}}{N^{1/3}}
        \mathrm{d}x^{b_1}
        \,+\,
        \big(\hspace{.8pt}
          \overline{\theta}
          \,\gamma^{b_1}\,
          \mathrm{d}\theta
        \big)
      \Big)
      \tfrac{1}{N^{1/3}}
      \delta_{b_1 b_2}
      \,+\,
      \mathrm{d}\theta^\beta
      \Big(
        \tfrac{1}{N^{1/3}}
        -
        \tfrac{c}{6}
      \Big)
      \big(\hspace{.8pt}
        \overline{\theta}\,\gamma_{b_2}
      \big)_\beta
      & +\;
      \mathcal{O}\big(
        \theta^2
      \big)
      &
      \proofstep{
        by
        \eqref{WorldvolumeFieldsToFirstThetaOrder}
        \&
        \eqref{ScndFundamentalFormToFirstOrder}
      }
      \\
      &\;=\;
      \tfrac{r_{\mathrm{prb}}}{N^{1/3}}
      \delta_{b_1 b_2}
      \mathrm{d}x^{b_2}
      \,+\,
      \tfrac{c}{6}
      \big(\hspace{.8pt}
        \overline{\theta}
        \,\gamma_{b_2}\,
        \mathrm{d}\theta
      \big)
      & +\;
      \mathcal{O}\big(
        \theta^2
      \big)
      \\
     & \;=\;
      \phi^\ast\, 
      \Omega^{5'}{}_{b_2}
      &
      +\, \mathcal{O}\big(\theta^2\big)
      &
      \proofstep{  
        by
        \eqref{WorldvolumeFieldsToFirstThetaOrder}
      }.
    \end{array}
  $$

\vspace{-.3cm}
\end{proof}

\begin{remark}[\bf \fixedtext{Positive charge}  of holographic M5-brane probe]
\label{CriticalDistanceForHolographicM5}
Since $c = \pm 6/N^{1/3}$
\eqref{ScaleFactorOfCFieldFlux},
Lem. \ref{WorldvolumeSuperFieldsToFirstOrder} implies that the holographic super-immersion \eqref{TheSuperImmersion} 
is (fluxless) $\sfrac{1}{2}$BPS (Def. \ref{FluxlessHalfBPSImmersion}) to first order in $\theta$ iff the background C-field flux density 
 \eqref{ScaleFactorOfCFieldFlux} is positive \footnote{
   Since the difference of signs in \eqref{ScaleFactorOfCFieldFlux} signifies the difference between black branes and black {\it anti-}branes that source the C-field flux, and if in the spirit of microscopic $p$-brane holography (p. \pageref{MicroscopicHolography}) we think of the black brane and its holographic probe as two aspects of the same physical system, then the requirement (Rem. \ref{CriticalDistanceForHolographicM5}) of the positive sign for the existence of the holographic probe characterizes this as an actual brane instead of an anti-brane.
 },
 in that, by \eqref{WorldvolumeFieldsToFirstThetaOrder} \eqref{VanishingTransveralPsiToFirstOrder}:
\begin{equation}
  \label{ConditionForPositiveCharge}
  \phi^\ast\big(\hspace{.8pt}
    \overline{P}\Psi
  \big)
  \;=\;
  \mathcal{O}\big(
    \theta^2
  \big)
  \hspace{.6cm}
  \Leftrightarrow
  \hspace{.6cm}
  \left\{\!\!
  \def\arraystretch{1.4}
  \begin{array}{l}
    c > 0,\;\; \mbox{\rm i.e.}\;\;
    G_4 \,=\, + \tfrac{6}{N^{1/3}}
    \mathrm{dvol}_{S^4_{N \mathrm{M5}}}
    \,.
  \end{array}
  \right.
\end{equation}
\end{remark}

\smallskip

Next, from the first-order formulas \eqref{WorldvolumeFieldsToFirstThetaOrder}, we now proceed by induction to the full computation of the worldvolume fields. For this, let now 
\vspace{1mm} 
\begin{equation}
  \label{SuperCoframeFieldOnPoincareSuperChart}
  (E,\Psi)
  \;\in\;
  \Omega^1_{\mathrm{dR}}
  \big(
    \mathbb{R}^{1,5}
    \!\times\!
    \mathbb{R}_{>0}
    \!\times\!
    \mathbb{D}^4
    \!\times\!
    \mathbb{R}^{0\vert\mathbf{32}}
    ;\,
    \mathbb{R}^{1,10\,\vert\, \mathbf{32}}
  \big)
\end{equation}

\vspace{1mm} 
\noindent denote the super coframe fields \eqref{CoFrameAroundHolographicM5}
on the Poincar{\'e}  neighborhood \eqref{NeighbourhoodChartOfWorldvolume}
of $\mathrm{AdS}_7 \!\times\! S^4$ 
uniquely extended to super-space via WZT gauge {\rm (Def. \ref{WZTGauge})}, to all orders in $\Theta$.

\medskip

Now we are ready for the main statement of this section:

\begin{theorem}[\bf Existence of  fluxless
$\sfrac{1}{2}$BPS holographic M5-brane probes]
  \label{ExistenceOfHolographicM5BraneProbes}
  The holographic super-immersion \eqref{TheSuperImmersion}
  of an M5-brane probe near the horizon of $N$ coincident black M5-branes 
  is (fluxless) $\sfrac{1}{2}$BPS
  {\rm (Def. \ref{FluxlessHalfBPSImmersion})}.
\end{theorem}

\begin{proof}
 By Lem. \ref{WorldvolumeSuperFieldsToFirstOrder} 
 with Rem. \ref{CriticalDistanceForHolographicM5},
 the statement holds to first order in the odd worldvolume coordinates. Hence it is sufficient to check that all higher contributions actually vanish.

First, the vanishing of the higher orders of the transversal gravitino,
\vspace{1mm} 
\begin{equation}
  \label{PullbackOfBarPPsiVanishesAtAnyThetaOrder}
  \phi^\ast \overline{P} \Psi
  \;=\;
  0
  \,,
  \;\;\;\;\;
  \mbox{equivalently}
  \;\;\;\;
  \phi^\ast \Psi
  \;=\;
  P \phi^\ast \Psi
\end{equation}
(using throughout that $\phi^\ast \circ P \,=\, P \circ \phi^\ast$  and similarly for $\overline{P}$)
follows via the decoupled recursion relations from Lem. \ref{DecoupledRheonomyRecursionRelations} 
by induction on the $\theta$-order:
\begin{itemize}
\item For the even component by
$$
\hspace{-3mm} 
  \def\arraystretch{1.7}
  \begin{array}{lll}
    \phi^\ast
    \big(\hspace{.8pt}
      \overline{P}
      \,
      \Psi^{(n+2)}_r
    \big)^\alpha
    \cdot
    (n+2)(n+1)
    \\
    \;=\;
    -\tfrac{1}{4}
    \big(\hspace{.8pt}
      \overline{P}
      \Gamma_{a_1 a_2}
      \theta
    \big)^\alpha
    \big(\hspace{.8pt}
      \overline{\theta}
      \,K^{a_1 a_2}\,
      {\color{darkorange}}
      \phi^\ast \Psi^{(n)}_r
    \big)
    \,+\,
    \big(\hspace{.8pt}
      \overline{P}
      H_a \theta
    \big)^\alpha
    \big(\hspace{.8pt}
      \overline{\theta}
      \,\Gamma^a\,
      {\color{darkorange}}
      \phi^\ast \Psi^{(n)}_r
    \big)
    &
    a_i \in \{0,\!\cdots\!,\!5,5',6,\!\cdots\!,\!9\}
    &
    \proofstep{
      by
      \eqref{DecoupledRecursionForGravitino}
      \&
      \eqref{TheSuperImmersion}
    }
    \\
    \;=\;
    -\tfrac{1}{4}
    \big(\hspace{.8pt}
      \overline{P}
      \Gamma_{a_1 a_2}
      \theta
    \big)^\alpha
    \big(\hspace{.8pt}
      \overline{\theta}
      \,K^{a_1 a_2}\,
      {\color{darkorange}P}
      \phi^\ast \Psi^{(n)}_r
    \big)
    \,+\,
    \big(\hspace{.8pt}
      \overline{P}
      H_a \theta
    \big)^\alpha
    \big(\hspace{.8pt}
      \overline{\theta}
      \,\Gamma^a\,
      {\color{darkorange}P}
      \phi^\ast \Psi^{(n)}_r
    \big)
    &
    a_i \in \{0,\!\cdots\!,\!5,5',6,\!\cdots\!,\!9\}
    &
    \proofstep{
      \hspace{-11pt}
      \def\arraystretch{.9}
      \begin{tabular}{l}
        by induction
        \\
        assumption
      \end{tabular}
    }
      \\
      \;=\;
      -
      \tfrac{1}{2}
      \big(\hspace{.8pt}
        \overline{P}
        \,
        \Gamma_{5' a}
        \theta
      \big)^\alpha
      \big(\hspace{.8pt}
        \overline{\theta}
        \,K^{5' a}\,
        P
        \phi^\ast
        \Psi^{(n)}_r
      \big)
      \,+\,
      \big(\hspace{.8pt}
        \overline{P}
        \,
        H_a \theta
      \big)^\alpha
      \big(\hspace{.8pt}
        \overline{\theta}
        \,\Gamma^a\,
        P
        \phi^\ast
        \Psi^{(n)}_r
      \big)
      &
      a \in \{0,\!\cdots\!,\!5\}
      &
      \proofstep{
        by
        \eqref{PropertiesOfProjectionOperator}
      }
      \\
      \;=\;
      -
      \tfrac{1}{2}
      \tfrac{c}{6}
      \big(\hspace{.8pt}
        \overline{P}
        \,
        \underbrace{
          \Gamma_{5' a}
        }_{\color{gray} 
          \mathclap{
            - \Gamma_{a 5'}
          }
        }
        \theta
      \big)^\alpha
      \big(\hspace{.8pt}
        \overline{\theta}
        \,\Gamma^a \,
        P
        \phi^\ast
        \Psi^{(n)}_r
      \big)
      \,-\,
      \tfrac{c}{12}
      \big(\hspace{.8pt}
        \overline{P}
        \,
        \Gamma_{a 5'}
        \theta
      \big)^\alpha
      \big(\hspace{.8pt}
        \overline{\theta}
        \,\Gamma^a\,
        P
        \phi^\ast
        \Psi^{(n)}_r
      \big)
      &
      a \in \{0,\!\cdots\!,\!5\}
      &
      \proofstep{
        by
        \eqref{OddComponentsOfFieldStrengthsOnAdS} \&
    \eqref{PropertiesOfProjectionOperator}
      }
    \\[-5pt]
    \;=\;
    0
    \,.
  \end{array}
$$
\item For the odd component 
by use of the Fierz identity 
from Lem. \ref{FierzIdentityIn6d}:
$$
\hspace{-3mm} 
   \def\arraystretch{1.6}
   \begin{array}{lll}
      \phi^\ast
      \big(\hspace{.8pt}
        \overline{P}
        \,
        \Psi^{(n+2)}_\rho
      \big)^\alpha
      \cdot (n+4)(n+3)\frac{1}{2} 
      \\
      \;=\;
      \tfrac{1}{4}
      \big(\hspace{.8pt}
        \overline{P}
        \,
        \Gamma_{a_1 a_2}
        \theta
      \big)^\alpha
      \big(\hspace{.8pt}
        \overline{\theta}
        \,K^{a_1 a_2}\,
        \phi^\ast
        \Psi^{(n)}_\rho
      \big)
      \,+\,
      \big(\hspace{.8pt}
        \overline{P}
        \,
        H_a \theta
      \big)^\alpha
      \big(\hspace{.8pt}
        \overline{\theta}
        \,\Gamma^a\,
        \phi^\ast
        \Psi^{(n)}_\rho
      \big)
      &
      a_i \in \{0,\!\cdots\!,\!5,5',6,\!\cdots\!,\!9\}
      &
      \proofstep{
        by \eqref{DecoupledRecursionForGravitino}
        \&
        \eqref{TheSuperImmersion}
      }
      \\
      \;=\;
      \tfrac{1}{4}
      \big(\hspace{.8pt}
        \overline{P}
        \,
        \Gamma_{a_1 a_2}
        \theta
      \big)^\alpha
      \big(\hspace{.8pt}
        \overline{\theta}
        \,K^{a_1 a_2}\,
        {\color{darkorange}P}
        \phi^\ast
        \Psi^{(n)}_\rho
      \big)
      \,+\,
      \big(\hspace{.8pt}
        \overline{P}
        \,
        H_a \theta
      \big)^\alpha
      \big(\hspace{.8pt}
        \overline{\theta}
        \,\Gamma^a\,
        {\color{darkorange}P}
        \phi^\ast
        \Psi^{(n)}_\rho
      \big)
      &
      a_i \in \{0,\!\cdots\!,\!5,5',6,\!\cdots\!,\!9\}
      &
      \proofstep{
      \hspace{-12pt}
      \def\arraystretch{.9}
      \begin{tabular}{l}
        by induction
        \\
        assumption
      \end{tabular}
      }
      \\
      \;=\;
      \tfrac{1}{2}
      \big(\hspace{.8pt}
        \overline{P}
        \,
        \Gamma_{5' a}
        \theta
      \big)^\alpha
      \big(\hspace{.8pt}
        \overline{\theta}
        \,K^{5' a}\,
        P
        \phi^\ast
        \Psi^{(n)}_\rho
      \big)
      \,+\,
      \big(\hspace{.8pt}
        \overline{P}
        \,
        H_a \theta
      \big)^\alpha
      \big(\hspace{.8pt}
        \overline{\theta}
        \,\Gamma^a\,
        P
        \phi^\ast
        \Psi^{(n)}_\rho
      \big)
      &
      a \in \{0,\!\cdots\!,\!5\}
      &
      \proofstep{
        by
        \eqref{PropertiesOfProjectionOperator}
      }
      \\
      \;=\;
      \tfrac{1}{2}
      \tfrac{c}{6}
      \big(\hspace{.8pt}
        \overline{P}
        \,
        \Gamma_{5' a}
        \theta
      \big)^\alpha
      \big(\hspace{.8pt}
        \overline{\theta}
        \,\Gamma^a \,
        P
        \phi^\ast
        \Psi^{(n)}_\rho
      \big)
      \,-\,
      \tfrac{c}{12}
      \big(\hspace{.8pt}
        \overline{P}
        \,
        \Gamma_{a 5'}
        \theta
      \big)^\alpha
      \big(\hspace{.8pt}
        \overline{\theta}
        \,\Gamma^a\,
        P
        \phi^\ast
        \Psi^{(n)}_\rho
      \big)
      &
      a \in \{0,\!\cdots\!,\!5\}
      &
      \proofstep{
        by
        \eqref{OddComponentsOfFieldStrengthsOnAdS} \&
    \eqref{PropertiesOfProjectionOperator}
      }
      \\
      \;=\;
      \tfrac{c}{6}
      \big(\hspace{.8pt}
        \overline{P}
        \,
        \Gamma_{5'}
        \underbrace{
          \gamma_{a}
          \theta
          \big)^\alpha
          \big(\hspace{.8pt}
          \overline{\theta}
          \,\gamma^a
        }_{\color{gray}
          = 0
        }
        \,
        P
        \phi^\ast
        \Psi^{(n)}_\rho
      \big)
      \\[-5pt]
      \;=\;
      0
      &&
      \proofstep{
        by
        \eqref{6dFierzIdentity}
      }.
      \end{array}
    $$     
\end{itemize}

\smallskip 
From this, it then follows that:
\begin{itemize}[
  topsep=1pt,
  itemsep=2pt
]
\item
The pullback of the radial \& transversal vielbein vanishes to all orders:
\begin{equation}
  \label{PullbackOfTransversalVielbeinVanishesToAllOrderes}
  \phi^\ast \overline{P} E
  \;=\;
  0
\end{equation}
because we now have for $E^{5'}$ that
$$
  \def\arraystretch{1.6}
  \begin{array}{ll}
    \phi^\ast
    \big(
      E^{(n+1)}
    \big)^{5'}_r
    \\
    \;=\;
    \tfrac{2}{n+1}
    \big(\hspace{.8pt}
      \overline{\theta}
      \,\Gamma^{5'}
      \phi^\ast
      \Psi^{(n)}_r
    \big)
    &
    \proofstep{
      by
      \eqref{RecursionRelationForGraviton}
      \&
      \eqref{TheSuperImmersion}
    }
    \\
    \;=\;
    \tfrac{2}{n+1}
    \big(\hspace{.8pt}
      \overline{\theta}
      \,\Gamma^{5'}
      {\color{darkorange}P}
      \phi^\ast
      \Psi^{(n)}_r
    \big)
    &
    \proofstep{
      by
      \eqref{PullbackOfBarPPsiVanishesAtAnyThetaOrder}
    }
    \\
    \;=\;
    0
    &
    \proofstep{
      by
      \eqref{PropertiesOfProjectionOperator},
    }
  \end{array}
  \hspace{1.7cm}
  \def\arraystretch{1.6}
  \begin{array}{ll}
    \phi^\ast
    \big(
      E^{(n+1)}
    \big)^{5'}_\rho
    \\
    \;=\;
    \tfrac{2}{n+2}
    \big(\hspace{.8pt}
      \overline{\theta}
      \,\Gamma^{5'}
      \phi^\ast
      \Psi^{(n)}_\rho
    \big)
    &
    \proofstep{
      by
      \eqref{RecursionRelationForGraviton}
      \&
      \eqref{TheSuperImmersion}
    }
    \\
    \;=\;
    \tfrac{2}{n+2}
    \big(\hspace{.8pt}
      \overline{\theta}
      \,\Gamma^{5'}
      {\color{darkorange}P}
      \phi^\ast
      \Psi^{(n)}_\rho
    \big)
    &
    \proofstep{
      by
      \eqref{PullbackOfBarPPsiVanishesAtAnyThetaOrder}
    }
    \\
    \;=\;
    0
    &
    \proofstep{
      by
      \eqref{PropertiesOfProjectionOperator},
    }
  \end{array}
$$
and verbatim so for $E^i$.

\item
The fermionic component of the tangential coframe field equals
\begin{equation}
  \label{FermionicTangentialCoframeFieldToAllOrders}
  \psi \,=\, \mathrm{d}\theta
\end{equation}
to all orders in $\theta$, because it does so to first order by \eqref{WorldvolumeFieldsToFirstThetaOrder} and all higher orders vanish (now $a_i \in 
\{0,\!\cdots\!,9\}$):
\begin{equation}
  \label{DeducingVanisingOfPsiR}
  \hspace{-.6cm}
  \def\arraystretch{1.8}
  \begin{array}{lll}
    &
    \big(
      \psi^{(n+2)}
    \big)^\alpha_r
    \;:=\;
    \phi^\ast
    \big(
      P
      \,
      \Psi^{(n+2)}
    \big)^\alpha_r
    \\
    &
    \;=\;
    -\tfrac{1}{n+2}
    \tfrac{1}{n+1}
    \tfrac{1}{4}
    \big(
      P
      \,
      \Gamma_{a_1 a_2}
      \theta
    \big)^\alpha
    \big(\hspace{.8pt}
      \overline{\theta}
      \,K^{a_1 a_2}\,
      \phi^\ast\Psi^{(n)}_r
    \big)
    \,+\,
    \tfrac{1}{n+2}
    \tfrac{1}{n+1}
    \big(
      P
      \,
      H_a \theta
    \big)^\alpha
    \big(\hspace{.8pt}
      \overline{\theta}
      \,\Gamma^a\,
      \phi^\ast \Psi^{(n)}_r
    \big)
    &
    \proofstep{
      by 
      \eqref{DecoupledRecursionForGravitino}
      \&
      \eqref{TheSuperImmersion}
    }
    \\
    &
    \;=\;
    -\tfrac{1}{n+2}
    \tfrac{1}{n+1}
    \tfrac{1}{4}
    \big(
      P
      \,
      \Gamma_{a_1 a_2}
      \theta
    \big)^\alpha
    \big(\hspace{.8pt}
      \overline{\theta}
      \,K^{a_1 a_2}
      {\color{darkorange}P}
      \phi^\ast\Psi^{(n)}_r
    \big)
    \,+\,
    \tfrac{1}{n+2}
    \tfrac{1}{n+1}
    \big(
      P
      \,
      H_a \theta
    \big)^\alpha
    \big(\hspace{.8pt}
      \overline{\theta}
      \,\Gamma^a
      {\color{darkorange}P}
      \phi^\ast \Psi^{(n)}_r
    \big)
    &
    \proofstep{
      by
      \eqref{PullbackOfBarPPsiVanishesAtAnyThetaOrder}
      \footnotemark
    }
  \\
  &
  \;=\;
  0
  &
  \proofstep{
    by
    \eqref{OddComponentsOfFieldStrengthsOnAdS}
    \&
    \eqref{PropertiesOfProjectionOperator}
  }
\end{array}
\end{equation}

\vspace{-2mm} 
and
\vspace{-2mm} 
$$
\hspace{-3mm}
  \def\arraystretch{1.8}
  \begin{array}{lll}
    &
    \big(
      \psi^{(n+2)}
    \big)^\alpha_\rho
    \;:=\;
    \phi^\ast
    \big(
      P
      \,
      \Psi^{(n+2)}
    \big)^\alpha_\rho
    \\
    & 
    \;=\;
    -\tfrac{1}{n+4}
    \tfrac{1}{n+3}
    \tfrac{1}{4}
    \big(
      P
      \,
      \Gamma_{a_1 a_2}
      \theta
    \big)^\alpha
    \big(\hspace{.8pt}
      \overline{\theta}
      \,K^{a_1 a_2}\,
      \phi^\ast\Psi^{(n)}_\rho
    \big)
    \,+\,
    \tfrac{1}{n+4}
    \tfrac{1}{n+3}
    \big(
      P
      \,
      H_a \theta
    \big)^\alpha
    \big(\hspace{.8pt}
      \overline{\theta}
      \,\Gamma^a\,
      \phi^\ast \Psi^{(n)}_\rho
    \big)
    &
    \proofstep{
      by 
      \eqref{DecoupledRecursionForGravitino}
      \&
      \eqref{TheSuperImmersion}
    }
    \\
    & 
    \;=\;
    -\tfrac{1}{n+4}
    \tfrac{1}{n+3}
    \tfrac{1}{4}
    \big(
      P
      \,
      \Gamma_{a_1 a_2}
      \theta
    \big)^\alpha
    \big(\hspace{.8pt}
      \overline{\theta}
      \,K^{a_1 a_2}
      {\color{darkorange}P}
      \phi^\ast\Psi^{(n)}_\rho
    \big)
    \,+\,
    \tfrac{1}{n+4}
    \tfrac{1}{n+3}
    \big(
      P
      \,
      H_a \theta
    \big)^\alpha
    \big(\hspace{.8pt}
      \overline{\theta}
      \,\Gamma^a
      {\color{darkorange}P}
      \phi^\ast \Psi^{(n)}_\rho
    \big)
    &
    \proofstep{
      by
      \eqref{PullbackOfBarPPsiVanishesAtAnyThetaOrder}
    }
  \\
    & 
    \;=\;
    0
  &
  \proofstep{
    by
    \eqref{OddComponentsOfFieldStrengthsOnAdS}
    \&
    \eqref{PropertiesOfProjectionOperator}\,.
  }
\end{array}
$$
Note that in the last step, in both cases, we observe
from \eqref{OddComponentsOfFieldStrengthsOnAdS}
that $K^{a_1 a_1}$ and $H_a$ have for all index values the same parity (with respect to the projectors $P$, $\overline{P}$\,) as $\Gamma_{a_1 a_2}$ and $\Gamma^a$, respectively so that the two terms $P\,  \Gamma_{a_1 a_2} P$ 
and $\overline{P} K^{a_1 a_2} P$ can never both be non-vanishing,
and similarly for $P H_a P$ and $\overline{P}\, \Gamma^a P$.

\smallskip 
\item
The bosonic component of the tangential coframe field equals
\vspace{1mm} 
\begin{equation}
  \label{ExactBosonicCoframeField}
  e^a
  \;=\;
  \mathrm{d}x^a
  \,+\,
  \big(\hspace{.8pt}
    \overline{\theta}
    \,\gamma^a\,
    \mathrm{d}\theta
  \big)
\end{equation}
to all orders in $\theta$,
because it does so to first order by \eqref{WorldvolumeFieldsToFirstThetaOrder}, 
and since all higher orders vanish, as follows:
$$
\hspace{-2mm}
  \def\arraystretch{1.6}
  \begin{array}{lll}
    \big(
      e^{(n+1)}
    \big)^{a}_r
   & :=\;
    \phi^\ast
    \big(
      E^{(n+1)}
    \big)^a_r
    \\
    &=\;
    \tfrac{2}{n+1}
    \big(\hspace{.8pt}
      \overline{\theta}
      \,\gamma_a\,
      \phi^\ast \Psi^{(n)}_r
    \big)
    &
    \proofstep{
      by
      \eqref{RecursionRelationForGraviton}
      \&
      \eqref{TheSuperImmersion}
    }
    \\
    &=\;0
    &
    \proofstep{
      by
      \eqref{PullbackOfBarPPsiVanishesAtAnyThetaOrder}
      \&
      \eqref{WorldvolumeFieldsToFirstThetaOrder}
    },
  \end{array}
  \hspace{1cm}
  \def\arraystretch{1.6}
  \begin{array}{lll}
    \big(
      e^{(n+2)}
    \big)^{a}_\rho
   & :=\;
    \phi^\ast
    \big(
      E^{(n+2)}
    \big)^a_\rho
    \\
    &=\;
    \tfrac{2}{n+3}
    \big(\hspace{.8pt}
      \overline{\theta}
      \,\gamma_a\,
      \phi^\ast \Psi^{(n+1)}_\rho
    \big)
    &
    \proofstep{
      by
      \eqref{RecursionRelationForGraviton}
      \&
      \eqref{TheSuperImmersion}
    }
    \\
    & =\;0
    &
    \proofstep{
      by
      \eqref{PullbackOfBarPPsiVanishesAtAnyThetaOrder}
      \&
      \eqref{WorldvolumeFieldsToFirstThetaOrder}
    }.
  \end{array}
$$
\end{itemize}

\smallskip

To conclude:
\begin{itemize}
\item
the statements \eqref{PullbackOfTransversalVielbeinVanishesToAllOrderes} and \eqref{PullbackOfBarPPsiVanishesAtAnyThetaOrder} establish the transversal condition in \eqref{ConditionOnFluxlessBPSImmersion} that was to be shown, namely that $\phi^\ast \big( \,\overline{P}E,\, \overline{P}\Psi \big) = 0$.

\vspace{1mm} 
\item
The statements \eqref{FermionicTangentialCoframeFieldToAllOrders} and  \eqref{ExactBosonicCoframeField} establish the tangential condition in \eqref{ConditionOnFluxlessBPSImmersion} that was to be shown, namely that $(e,\psi)$ is a coframe field, manifestly so by expanding the coordinate differentials in their $(e,\psi)$-components as
$$
  \def\arraystretch{1.6}
  \begin{array}{ll}
    \mathrm{d}x^a
    \;=\;
    \tfrac{N^{1/3}}{r_{\mathrm{prb}}}
    \, e^a
    \,-\,
    \big(\,
      \overline{\theta}
      \,\gamma^a\,
      \psi
    \big)
    \\
    \mathrm{d}\theta^\alpha
    \;=\;
    \psi^\alpha.
  \end{array}
$$
\end{itemize}
This completes the check that $\phi$ \eqref{TheSuperImmersion} is a (fluxless) $\sfrac{1}{2}$BPS super-immersion (Def. \ref{FluxlessHalfBPSImmersion}), hence that the holographic probe M5-brane really exists.
\end{proof}

\begin{remark}[\bf Bianchi identity and vanishing $H_3$-flux density]
$\,$

\noindent {\bf (i)}  For flux quantization on holographic M-branes in \cite{SS24-Companion}, the key point of establishing the $\sfrac{1}{2}$BPS property of the holographic M5-brane immersion, via Thm. \ref{ExistenceOfHolographicM5BraneProbes}, is that this establishes a solution to the equations of motion of the $H_3$-flux density on the worldvolume (\cite[Prop. 3.17]{GSS24-FluxOnM5}), namely the appropriate self-duality, the Bianchi identity, and rheonomy.

\vspace{1mm} 
\noindent {\bf (ii)}  In the present case of {\it vanishing} flux density, this may look fairly trivial, but it is still crucial to establish it unambiguously as a solution because (only) then is flux quantization guaranteed to produce the exact completed field content which may still be non-trivial (namely torsion-charged), as discussed in 
\cite{SS24-Companion}. 

\vspace{1mm} 
\noindent {\bf (iii)}  In any case, it is immediate to check the conclusions of \cite[Prop. 3.17]{GSS24-FluxOnM5} in the present case: In particular, with \eqref{FluxDensityOnChart} and \eqref{TheSuperImmersion} we have
 \vspace{1mm} 
  \begin{equation}
    \label{4FluxVanishesOnHolographicWorldvolume}
    \phi^\ast G_4 \;=\; 0
  \end{equation}

\vspace{1mm} 
\noindent  
so that the general worldvolume Bianchi identity $\mathrm{d}H_3 \,=\, \phi^\ast G_4$ (cf. \cite[(1)]{GSS24-FluxOnM5}) is un-twisted and becomes 
\vspace{1mm} 
  $$
    \mathrm{d}H_3
    \;=\;
    0
    \,,
  $$
  which is clearly satisfied by $H_3 = 0$.
\end{remark}

\begin{remark}[\bf Absence of fluxed $\sfrac{1}{2}$BPS holographic M5-branes]
\label{AbsenceOfFluxedHolographicM5branes}
The proof of Thm. \ref{ExistenceOfHolographicM5BraneProbes} also readily shows that it is impossible to
have non-vanishing worldvolume flux density $H_3 \neq 0$ on a holographic M5-brane \eqref{TheSuperImmersion}
while keeping its $\sfrac{1}{2}$BPS- (``super-embedding''-) property (at least with respect to the given coframe field \eqref{SuperCoframeFieldOnPoincareSuperChart}).
Namely, by \cite[(40)]{HoweSezgin97b}\cite[(7)]{HoweSezginWest97}\cite[p. 91]{Sorokin00} (re-derived in \cite[(126)]{GSS24-FluxOnM5}) such non-trivial flux corresponds to modifying the super-immersion \eqref{TheSuperImmersion} by a summand $\slashed{\tilde H}_3$
\vspace{1mm} 
$$
  \phi^\ast P \Theta
  \;=\;
  \theta 
    \,+\,
  \slashed{\tilde H}_3
  \theta
  \,,
  \hspace{.5cm}
  \mbox{for}
  \hspace{.4cm}
  \slashed{\tilde H}_3
  \;\defneq\;
  \tfrac{1}{3!}
  (\tilde H_3)_{a_1 a_2 a_3}
  \gamma^{a_1 a_2 a_3}
  \,,
$$

\vspace{1mm}
\noindent which vanishes iff the actual flux density $H_3$ vanishes (cf. \cite[Rem. 3.18]{GSS24-FluxOnM5}) --
but non-vanishing such $\tilde{H}_3$ immediately fails the Darboux condition \eqref{PullbackOfTransversalVielbeinVanishesToAllOrderes}, by the computation shown right below there.
(This is in contrast notably to the case of the rectilinear embedding of the M5-brane into flat Minkowski superspacetime, which allows any constant $H_3$-flux to be switched on, see \cite[Ex. 3.14]{GSS24-FluxOnM5}).
\end{remark}

This phenomenon naturally leads over to the discussion of flux-quantization on holographic M5-branes in \cite{SS24-Companion}.
Namely a constraint of vanishing flux density
$$
  H_3 \;=\; 0
$$
trivializes the higher gauge field on holographic M5-branes only locally, on any (contractible) coordinate chart, while the globally completed higher gauge field, controlled by a flux quantization law, may still attain non-trivial configurations carrying non-trivial {\it torsion charges}. 

\smallskip 
In other words, while flux quantization completes general gauge field configurations by torsion-charged sectors, this is particularly relevant for configurations with vanishing flux, as found here on holographic M5-branes, in which case the non-trivial higher gauge field content is invisible by traditional local field analysis and is {\it all} contained in the subtleties of the flux quantization law. This is what we discuss 
in \cite{SS24-Companion}.

\medskip

\section{Holographic M2-Branes}
\label{HolographicM2Branes}

For comparison, we also make an analogous analysis for the M2-brane case. Hence we determine the holographic M2-brane probes super-embedded into the near-horizon super-geometry of their own black brane incarnation.
This is the case of microscopic $p$-brane holography, which was originally considered in \cite{BDPS87}\cite{DFFFTT99}\cite{PST99}.

\subsection{Spinors on M2-Branes}
\label{SpinorsOnM2Branes}

\noindent
{\bf Spinors in 3d from spinors in 11d.}
We conveniently identify (cf. \cite[Lem. 4.11]{HSS19}) the $\mathrm{Spin}(1,2)$-representation $8 \cdot \mathbf{2} \,\in\, \mathrm{Rep}_{\mathbb{R}}\big(\mathrm{Spin}(1,2)\big)$ with the linear subspace of the $\mathrm{Spin}(1,10)$-representation $\mathbf{32}$ \eqref{The11dMajoranaRepresentation} that is the image of either of the projection operators
\begin{equation}
  \label{ProjectorForM2Brane}
  \def\arraystretch{1.4}
  \begin{array}{l}
    P
    \;:=\;
    \tfrac{1}{2}\big(
      1 + 
      \Gamma_{2'3456789}
    \big)
    \\
    \Porth
    \;:=\;
    \tfrac{1}{2}\big(
      1 
        -
      \Gamma_{2'3456789}
    \big)    
  \end{array}
  \;:\;
  \mathbf{32}
  \xrightarrow{\quad}
  \mathbf{32}
  \,,
\end{equation}
satisfying the following evident but consequential relations:
\begin{equation}
  \label{ProjectionOperatorRelationsOn2Brane}
  \def\arraycolsep{2pt}
  \def\arraystretch{1.1}
  \begin{array}{rcl}
    P P &=& P
    \\
    \Porth \Porth &=& \Porth
    \\
    \Porth P &=& 0
    \\
    P \Porth &=& 0
  \end{array}
  \hspace{1.2cm}
  \begin{array}{l}
    \def\arraystretch{1.1}
    \begin{array}{l}
      \Gamma^a P \;=\;
        P \, \Gamma^a
      \\
      \Gamma^a \Porth \;=\;
        \Porth \, \Gamma^a
    \end{array}
    \;\;\;
    a \in \{0,1,2\}
    \\
    \def\arraystretch{1.1}
    \begin{array}{l}
      \Gamma^{2'} \! P \;=\;
        \Porth \, \Gamma^{2'}
      \\
      \Gamma^{2'} \! \Porth \;=\;
        P \, \Gamma^{2'}
    \end{array}    
    \\
    \def\arraystretch{1.1}
    \begin{array}{l}
      \Gamma^{i}\; P \;=\;
        \Porth \, \Gamma^{i}
      \\
      \Gamma^{i} \; \Porth \;=\;
        P \, \Gamma^{i}
    \end{array}    
    \;\;\;
    i \in \{3,4,5,6,7,8,9\}
  \end{array}
  \hspace{.9cm}
  \def\arraystretch{1.3}
  \begin{array}{l}
    \overline{P} \;=\; P
    \\
    \overline{\Porth} \;=\; \Porth
    \,,
  \end{array}
\end{equation}
where we suggestively denote the 11d Clifford generators as follows:
\begin{equation}
  \hspace{-6cm}
  \begin{tikzcd}[
    column sep=3pt,
    row sep=0pt
  ]
    &
    \mathclap{
      \overbrace{
      \hspace{-3pt}
      \phantom{------\,\,}
      }^{
        \scalebox{.7}{
          \color{darkblue}
          \bf tangential
        }
      }
    }
    &
    &
    \overbrace{
      \phantom{--}
    }^{
        \scalebox{.7}{
          \color{darkblue}
          \bf radial
        }
    }
    &&&
    &
    \mathclap{
      \hspace{0pt}
      \overbrace{
      \hspace{-2pt}
     \phantom{\phantom{---------------}}
      }^{
        \hspace{0pt}
        \scalebox{.7}{
          \color{darkblue}
          \bf transversal
        }
      }
    }
    &
    \\[-14pt]
    \Gamma_0
    \ar[
      d,
      |->,
      start anchor={[yshift=+1pt]},
      end anchor={[yshift=-1pt]},
      bend right=90,
      shift right=0pt,
      "{\color{darkgreen}
        \def\arraystretch{.9}
        \begin{array}{l}
          P(-)P
        \end{array}
      }"{swap, pos=.5, xshift=4pt}
    ]
    &
    \Gamma_1
    &
    \Gamma_2
    &
    \Gamma_{2'}
    &
    \Gamma_3
    &
    \Gamma_4
    &
    \Gamma_{5}
    &
    \Gamma_{6}
    &
    \Gamma_{7}
    &
    \Gamma_{8}
    &
    \Gamma_{9}
    &
    \mathrlap{
       \in 
      \mathrm{Pin}^+(1,10)
      \,\subset\,
      \mathrm{End}_{\mathbb{R}}(\mathbf{32})
    }
    \\
    \gamma_0
    &
    \gamma_1
    &
    \gamma_2
    &
    &
    &
    &
    &&&&&
    \mathrlap{
    \in 
    \mathrm{Pin}^+(1,2)
    \,\subset\,
    \mathrm{End}_{\mathbb{R}}
    (
      8 \cdot \mathbf{2}
    )
    \,,
    }
  \end{tikzcd}
\end{equation}
in that under the corresponding inclusion
$$
  \mathrm{Spin}(1,2)
  \xhookrightarrow{\quad}
  \mathrm{Spin}(1,10)
$$
there are linear isomorphisms 
\vspace{1mm}
$$
  8 \cdot \mathbf{2}
  \;\;
  \underset{
    \mathclap{
      \raisebox{-2pt}{
        \scalebox{.7}{$
          \mathrm{Spin}(1,2)
        $}
      }
    }
  }{
   \;\;\simeq\;\;
  }
  \;
  P(\mathbf{32})
  \;
  \underset{
    \mathclap{
      \raisebox{-2pt}{
        \scalebox{.7}{$
          \mathrm{Spin}(1,2)
        $}
      }
    }
  }{
   \;\;\simeq\;\;
  }
 \; \Porth(\mathbf{32})
  \,.
$$

Combined with the vector representation of $\mathrm{Spin}(1,10)$ and $\mathrm{Spin}(1,2)$ on $\mathbb{R}^{1,10}$ and $\mathbb{R}^{1,2}$, respectively, we may regard $P$
\eqref{ProjectorForM2Brane}
as a projector of super-vector spaces
\begin{equation}
  \label{SuperProjectionOperatorForM2}
  \adjustbox{
    raise=-6pt
  }{
   \begin{tikzcd}
    \mathbb{R}
      ^{
      1,10 \,\vert\, \mathbf{32}
      }
    \ar[
      r,
      ->>
    ]
    \ar[
      rr,
      rounded corners,
      to path={
           ([yshift=+00pt]\tikztostart.north)  
        -- ([yshift=+08pt]\tikztostart.north)
        -- node[]{
          \scalebox{.8}{
            \colorbox{white}{
              $P$
            }
          }
        }
           ([yshift=+08pt]\tikztotarget.north)
        -- ([yshift=-00pt]\tikztotarget.north)
      }
    ]
    &
    \mathbb{R}
      ^{1,2 \,\vert\, 8 \cdot \mathbf{2}}
    \ar[
      r,
      hook
    ]
    &
    \mathbb{R}^{
      1,10 \,\vert\, \mathbf{32}
    }
  \end{tikzcd}
  }
  \hspace{.4cm}
  \def\arraystretch{1.3}
  \begin{array}{l}
     P 
       \,:=\, 
    \tfrac{1}{2}
    \big(
      1 + \Gamma_{2'3456789}
    \big) 
    \,,
  \end{array}
\end{equation}
which is convenient for unifying the conditions on tangential and transversal super-coframe components in a $\sfrac{1}{2}$BPS super-immersion (Def. \ref{HalfBPSSuperImmersion}).

\medskip

\noindent
{\bf Hodge duality on the M2.}
By 11d Hodge duality, we have 
\begin{equation}
  \def\arraystretch{1.4}
  \begin{array}{l}
    \Gamma_{012}
    \underset{
      \scalebox{.7}{
        \eqref{CliffordVolumeFormIn11d}
      }
    }{\;=\;}
    \Gamma_{012}
    \Gamma_{0122'3456789}
    \;=\;
    +
    \Gamma_{2'3456789}
  \end{array}
\end{equation}
so that the projector $P$ 
\eqref{ProjectorForM2Brane} is alternatively expressed as
$$
  P 
  \;=\;
  \tfrac{1}{2}
  \big(
    1 + \Gamma_{012}
  \big)
$$
which makes it manifest that
\begin{equation}
  \label{3HodgeDuality}
  \def\arraystretch{1.4}
  \begin{array}{l}
    \tfrac{1}{2}
    \epsilon_{a \, b_1 b_2}
    \Gamma^{b_1 b_2}
    P
    \;=\;
    \tfrac{1}{2}
    \epsilon_{a \, b_1 b_2}
    \Gamma^{b_1 b_2}
    \,
    \Gamma_{012}
    \,
    P
    \;=\;
    -
    \Gamma_a P\,.
  \end{array}
\end{equation}

\medskip

\subsection{Super $\mathrm{AdS}_4$-spacetime}
\label{SuperAdS4}

With the result of \S\ref{ExplicitRheonomy} in hand, we may give explicit formulas for super $\mathrm{AdS}_4 \!\times\! S^7$-spacetime by first recalling the ordinary bosonic AdS-geometry and then rheonomically extending to super-spacetime.

\medskip

\noindent
{\bf Near-horizon geometry of black M2-branes.}
The bosonic near-horizon geometry of $N$ black M2-brane is (cf. \cite[(5)]{AFHS00}, following \cite{GibbonsTownsend93}\cite{DuffGibbonsTownsend94}\cite[(2.3)]{DFFFTT99}) represented on a chart of the form 
\vspace{.1cm} 
\begin{equation}
  \label{ChartNearM2Singularity}
  \mathbb{R}^{1,10} 
    \setminus 
  \mathbb{R}^{1,2}
  \;\;
    \underset{\color{orangeii} \mathrm{diff}}{\simeq}
  \;\;
  \mathbb{R}^{1,2}
  \times
  \big(
    \mathbb{R}^8 \setminus \{0\}
  \big)
  \;\;
    \underset{\color{orangeii} \mathrm{diff}}{\simeq}
  \;\;
  \mathbb{R}^{1,2}
  \times
  \mathbb{R}_{> 0}
  \times 
  S^7
\end{equation}
with its canonical coordinate functions
\begin{equation}
  \label{CanonicalCoordinatesOnAdS4PoincareChart}
  \begin{tikzcd}[
    sep=0pt
  ]
    X^a
    &:&
    \mathbb{R}^{1,2}
    \ar[r]
    &[20pt]
    \mathbb{R}
    &[10pt]
    \mbox{for $a \in \{0,1,\cdots, 2\}$}
    \\[-2pt]
    r
    &:&
    \mathbb{R}_{> 0}
    \ar[
      r,
      hook
    ]
    &
    \mathbb{R}
  \end{tikzcd}
\end{equation}
by the $\mathrm{AdS}_4$-metric (cf. \cite[\S 39.3.7]{Blau22})
plus the metric on the round $S^7$:
\vspace{1mm} 
\begin{equation}
  \label{MetricTensorForBlackM2}
  \begin{array}{ccl}
  \mathrm{d}s^2_{N \mathrm{M2}}
  &=&
  \tfrac{r^2}{N^{2/6}}
  \mathrm{d}s^2_{\mathbb{R}^{1,2}}
  \,+\,
  \tfrac{N^{2/6}}{r^2}
  \mathrm{d}r^2
  \;\;+\;\;
  4\, N^{2/6}
  \,
  \mathrm{d}s^2_{S^7}
  \end{array}
\end{equation}
(where $2 \, R_{N \mathrm{M2}}\,:=\, 2\, N^{1/6}$ is the radius of the 7-sphere in Planck units $ \tfrac{\pi^{1/3}}{2^{1/6}}\, \ell_P$). The C-field flux density $G_4$ supporting this is a multiple of the volume form on the $\mathrm{AdS}_4$-factor pulled back to the chart along the canonical map
\begin{equation}
  \label{CFieldFluxDensityNearM2Singularity}
  G_4
  \;=\;
  c 
  \,
  \mathrm{dvol}_{\mathrm{AdS}^{N \mathrm{M2}}_4}
  \,\in\,
  \Omega^4_{\mathrm{dR}}
  \big(
    \mathrm{AdS}_4
  \big)
  \xhookrightarrow{\phantom{--}}
  \Omega^4_{\mathrm{dR}}
  \big(
    \mathbb{R}^{1,2}
    \times 
    \mathbb{R}_{> 0}
    \times
    S^4
  \big)  
  \,,
\end{equation}
for some prefactor $c$ which is determined, up to its sign, by the Einstein equations, see \eqref{OnAdS4ScaleFactorOfCFieldFlux} below, and whose sign is determined by the condition that a holographic M2-embedding exists (Rem. \ref{CriticalDistanceForHolographicM2} below).

\medskip

\noindent
{\bf Chart around a holographic M2-brane embedding.} 
We pick a point
$
  s_{\mathrm{prb}} \,\in\, S^7
  \subset \mathbb{R}^8 \setminus \{0\}
$
to designate the direction in which we wish to consider a probe M2-brane worldvolume immersed into this background, at some coordinate distance $r_{\mathrm{prb}}$ from the M2 singularity (cf. \cite[(14)]{BDPS87}\cite[(12)]{PST99} and 
\hyperlink{FigureBraneConfiguration}{Figure B}):
\begin{equation}
  \label{OrdinaryHolographicM2Embedding}
  \begin{tikzcd}[
    column sep=25pt,
    row sep = -2pt
  ]
    \mathllap{
    \scalebox{.7}{
      \color{darkblue}
      \bf
      \def\arraystretch{.9}
      \begin{tabular}{c}
        probe M2
        \\
        worldvolume
      \end{tabular}
    }
    }
    \mathbb{R}^{1,2}
    \ar[
      rr,
      "{ \phi }",
      "{
        \scalebox{.7}{
          \color{darkgreen}
          \bf
          embedding
        }
      }"{swap, yshift=-2pt}
    ]
    &&
    \mathbb{R}^{1,2}
    \times 
    \mathbb{R}_{> 0}
    \times
    S^7
    \\
    x
    &\longmapsto&
    \big(
      x,
      r_{\mathrm{prb}}
      ,
      s_{\mathrm{prb}}
    \big).
  \end{tikzcd}
\end{equation}
Restricting to an open ball around $s_{\mathrm{prb}} \in S^7$, we have a contractible chart around this immersed worldvolume, of the form 
\begin{equation}
  \label{NeighbourhoodChartOfM2Worldvolume}
  \begin{tikzcd}[column sep=large]
    \mathbb{R}^{1,2}
    \times
    \mathbb{R}_{> 0 }
    \times
    \mathbb{D}^7
    \ar[
      r,
      hook,
      "{
        \mathrm{id}
        \times
        \iota
      }"
    ]
    &
    \mathbb{R}^{1,2}
    \times
    \mathbb{R}_{> 0 }
    \times
    S^7
    \,.
  \end{tikzcd}
\end{equation}

\newpage
\noindent
{\bf Cartan geometry around the holographic M2.}
On the chart \eqref{NeighbourhoodChartOfM2Worldvolume},
we evidently have the following coframe forms 
\begin{equation}
  \label{CoFrameAroundHolographicM2}
  \def\arraystretch{1.7}
  \begin{array}{clclcl}
    &
    E^a 
      &:=&
    \tfrac{r}{N^{1/6}}
    \;
    \mathrm{d}X^a
    &
    \scalebox{.8}{tangential}
    &
    a \in \{0,1,2\}
    \\[-11pt]
    \scalebox{.7}{
      \color{darkblue}
      \bf
      AdS
    }
    \\[-11pt]
    &
    E^{2'} 
    &:=&
    \tfrac{N^{1/6}}{r}
    \;
    \mathrm{d} r
    &
    \scalebox{.8}{radial}
    &
    a \in \{2'\}
    \\
    \scalebox{.7}{
      \color{darkblue}
      \bf
      S
    }
    &
    E^a 
    &=&
    \tfrac{N^{1/6}}{1/2}
    \delta^a_i E^i
    &
    \scalebox{.8}{transversal}
    &
    a \in \{3,4,5,6,7,8,9\}
    \,,
  \end{array}
\end{equation}
which are orthonormal for the metric \eqref{MetricTensorForBlackM2}, in that $\mathrm{d}s^2_{N \mathrm{M2}} \,=\, \eta_{a b} \, E^a \otimes E^b$ on this chart, and make the C-field flux density \eqref{CFieldFluxDensityNearM2Singularity} appear as 
\begin{equation}
  \label{CFieldFluxOnAdS4InFrame}
  G_4 \;=\; 
  \pm c \, 
  E^0 E^1 E^2 E^{2'}
  \,.
\end{equation}

For the following formulas, we may focus on the AdS-factor in \eqref{CoFrameAroundHolographicM2}. Hence we let the indices $a_i, b_i$ run only through $\{0,1,2\}$, to be called the {\it tangential} index values -- namely tangential to the worldvolume \eqref{OrdinaryHolographicM2Embedding} -- with the further \textit{radial} index $2'$ carried along separately.

\smallskip

The torsion-free {\bf spin connection} on the AdS-factor of \eqref{CoFrameAroundHolographicM2},
characterized by
$$
      \mathrm{d}E^a 
    \,=\,
    \Omega^a{}_b\, E^b
    + 
    \Omega^a{}_{2'}\, E^{2'},
    \qquad
    \mathrm{d}E^{2'}
    \,=\,
    \Omega^{2'}{}_a \, E^a
    \,,
$$
is readily seen to have as only non-vanishing components:
\begin{equation}
  \label{AdS4SpinConnection}
  \Omega^{a 2'}
  \,=\,
  -\Omega^{2' a}
  \;=\;
    -
    \tfrac{r}{N^{2/6}}
    \,
    \mathrm{d}X^a
    \;\;
    \scalebox{.8}{
      tangential $a$.
    }
\end{equation}

Therefore its {\bf curvature} 2-form has non-vanishing components
\begin{equation}
  \label{AdS4CurvatureTwoForm}
  \def\arraystretch{1.5}
  \begin{array}{ccl}
  R^{a 2'}
  &=&
  \mathrm{d}
  \Omega^{a 2'}
  \;=\;
  -
  \tfrac{1}{N^{2/6}}
  \,
  \mathrm{d}r
  \,
  \mathrm{d}X^a
  \;=\;
  -
  \tfrac{1}{N^{2/6}}
  \,
  E^{2'}\, E^a
  \\
  &=&
  +
  \tfrac{1}{N^{2/6}}
  E^{a}\, E^{2'}
  \\
  R^{a_1 a_2}
  &=&
  -
  \Omega^{a_1}{}_{2'}
  \,
  \Omega^{2' a_2}
  \;=\;
  +
  \tfrac{r^2}{N^{4/6}}
  \mathrm{d}X^{a_1}
  \,
  \mathrm{d}X^{a_2}
  \\
  &=&
  +
  \tfrac{1}{N^{2/6}}
  \,
  E^{a_1} \, E^{a_1}
  \,.
  \end{array}
\end{equation}
Hence the {\bf Riemann tensor}
has non-vanishing components (cf. our normalization of $\delta$ in \eqref{KroneckerSymbol})
\begin{equation}
  \label{AdS4RiemannTensor}
  \def\arraystretch{1.6}
  \begin{array}{ccl}
  R^{a 2'}{}_{b 2'}
  &=&
  +\tfrac{1}{N^{2/6}}
  \,
  \delta^{a}_b
  \\
  R^{a_1 a_2}{}_{b_1 b_2}
  &=&
  +
  \tfrac{2}{N^{2/6}}
  \,
  \delta^{a_1 a_2}{}_{b_1 b_2}
  \,,
  \end{array}
\end{equation}
and the {\bf Ricci tensor} is proportional to the metric tensor, as befits an Einstein manifold:
\begin{equation}
  \label{AdS4RicciTensor}
   \def\arraystretch{1.4}
  \begin{array}{ccl}
  \mathrm{Ric}_{a_1 a_2}
  &=&
  R_{a_1}{}^b{}_{b a_2}
  +
  R_{a_1}{}^{2'}{}_{2' a_2}
  \\
  &=&
  -
  \tfrac{(3-1)}{N^{2/6}}
  \,
  \eta_{a_1 a_2}
  -
  \tfrac{1}{N^{2/6}} 
  \,
  \eta_{a_1 a_2}
  \\
  &=&
  -
  \tfrac{3}{N^{2/6}}
  \,
  \eta_{a_1 a_2}
\\
  \mathrm{Ric}_{2' 2'}
  &=&
  R^{2'}{}^a{}_{a 2'}
  \\
  &=&
  -
  \tfrac{3}{N^{2/6}}
  \,,
  \end{array}
\end{equation}
similar to the Ricci tensor of the 7-sphere factor (e.g. \cite[Cor. 11.20]{Lee18}):
 
\begin{equation}
  \label{RicciTensorOfSevenSphere}
  \mathrm{Ric}_{i_1 i_2}
  \;=\;
  +
  \tfrac{6}{
    4 
    \,
    N^{2/6}
  }
  \, 
  \delta_{i_1 i_2}
  \,.
\end{equation}

\smallskip
Therefore the  {\bf Einstein equation} with source the C-field flux density \eqref{CFieldFluxOnAdS4InFrame} 
$$
  (G_4)_{a_1 \cdots a_4}
  \;=\;
  \left\{\!\!
  \def\arraystretch{1.2}
  \begin{array}{ll}
    c \, \epsilon_{a_1 a_2 a_3 a_4}
    & 
    a_i \in \{0,1,2,2'\}
    \\
    0 & \scalebox{.9}{otherwise}
  \end{array}
  \right.
  \quad 
  \underset{
    \scalebox{.7}{\eqref{ContractingKroneckerWithSkewSymmetricTensor}}
  }{
    \Rightarrow
  }
  \quad 
  \left\{\!
  \def\arraystretch{1.3}
  \def\arraycolsep{2pt}
  \begin{array}{rcl}
    (G_4)_{b_1 b_2 b_3 b_4}
    (G_4)^{b_1 b_2 b_3 b_4}
    &=&
    - 24 \, c^2
    \\
    (G_4)_{a_1 b_1 b_2 b_3}
    (G_4)_{a_2}{}^{b_1 b_2 b_3}
    &=&
    -6 \, c^2 \,  \eta_{a_1 a_2}
  \end{array}
  \right.
$$
has non-vanishing components (cf. \cite[(174-5)]{GSS24-SuGra})
\begin{equation}
  \label{AdS4EinsteinEquation}
  \def\arraystretch{1.4}
  \begin{array}{cccl}
    &
    \mathrm{Ric}_{a_1 a_2}
    &=&
    \tfrac{1}{12}
    (G_4)_{a_1 \, b_1 b_2 b_3}
    (G_4)_{a_2}{}^{b_1 b_2 b_3}
    \,-\,
    \tfrac{1}{12}
    \tfrac{1}{12}
    (G_4)_{b_1 \cdots b_4}
    (G_4)^{b_1 \cdots b_4}
    \,
    \eta_{a_1 a_2}    
    \\
    \Leftrightarrow
    &
    -
    \tfrac{3}{
      N^{2/6}
    } \, \eta_{a_1 a_2} 
    &=&
    -\tfrac{1}{2}
    c^2\, \eta_{a_1 a_2}
    \,+\,
    \tfrac{1}{6}
    c^2 \, \eta_{a_1 a_2}
    \\
    &
    &=&
    -
    \tfrac{1}{3}
    c^2
    \, 
    \eta_{a_1 a_2}
    \\[+9pt]
    &
    \mathrm{Ric}_{i_1 i_2}
    &=&
    -
    \tfrac{1}{12}
    \tfrac{1}{12}
    (G_4)_{a_1 \cdots a_4}
    (G_4)^{a_1 \cdots a_4}
    \,
    \delta_{i_1 i_2}
    \\
    \Leftrightarrow
    &
    +
    \tfrac{6}{
      {\color{purple}4}
      \,
      N^{2/6} 
    }
    \,
    \delta_{i_1 i_2}
    &=&
    +
    \tfrac{1}{6}
    c^2
    \,
    \delta_{i_1 i_2}
  \end{array}
\end{equation}

\vspace{2mm} 
\noindent thus is solved 
\footnote{Note that the last line in \eqref{AdS4EinsteinEquation} is the reason that the radius of $S^7$ has to be {\it twice} that of $\mathrm{AdS}_4$ in \eqref{MetricTensorForBlackM2}. } 
by
\vspace{1mm} 
\begin{equation}
  \label{OnAdS4ScaleFactorOfCFieldFlux}
  c 
  \;=\; 
  \pm
  \frac{3}
  {N^{1/6}
  }
  \,,
  \hspace{.6cm}
  \underset{
    \scalebox{.7}{
      \eqref{CFieldFluxOnAdS4InFrame}
    }
  }{\mbox{hence}} 
  \hspace{.5cm}
  G_4
  \;=\;
  \pm 
  \tfrac{3}{N^{1/6}}
  \, 
  E^0 E^1 E^2 E^{2'}
  .
\end{equation}
At this point, both of the signs in \eqref{OnAdS4ScaleFactorOfCFieldFlux} are equally admissible, but we see below in Rem. \ref{CriticalDistanceForHolographicM5} that the + sign is singled out by the existence of a holographic M2-brane embedding.

 \medskip 

\noindent
{\bf Super-Cartan geometry around holographic M2s.}
We now obtain the super-extension of the above Cartan geometry \eqref{CoFrameAroundHolographicM2}.
Inserting the bosonic AdS Cartan geometry \eqref{CoFrameAroundHolographicM2} \eqref{AdS4SpinConnection} into the initial conditions for WZT gauge \eqref{GaugeConditions}
means that
\begin{equation}
  \label{InitialValuesInAdS4}
  \def\arraystretch{1.9}
  \begin{array}{lcrcll}
    \big(
      E^{(0)}
    \big)^a
    &=&
    \tfrac{r}{N^{1/6}}
    \,
    \mathrm{d}X^a
    &
      \;\;\;\;\;
      \Leftrightarrow
      \;\;\;\;\;
    &
    \Big(
      \big(
        E^{(0)}
      \big)^a_r \,=\,
      \tfrac{r}{N^{1/6}}
      ,
      &
      \big(
        E^{(0)}
      \big)^a_\rho \,=\,
      0
      \;    
    \Big)
    \\
    \big(
      E^{(0)}
    \big)^{2'}
    &=&
    \tfrac{N^{2/6}}{r}
    \,
    \mathrm{d}X^{2'}
    &
      \;\;\;\;\;
      \Leftrightarrow
      \;\;\;\;\;
    &
    \Big(
      \big(
        E^{(0)}
      \big)^{2'}_r \,=\,
      \tfrac{N^{2/6}}{r}
      ,
      &
      \big(
        E^{(0)}
      \big)^{2'}_\rho \,=\,
      0
      \;
    \Big)
    \\
    \big(
      \Psi^{(0)}
    \big)^\alpha
    &=&
    \mathrm{d}\Theta^\alpha
    &
      \;\;\;\;\;
      \Leftrightarrow
      \;\;\;\;\;
    &
    \Big(
      \big(
        \Psi^{(0)}
      \big)^\alpha_r
      \,=\,
      0
      ,
      &
      \big(
        \Psi^{(0)}
      \big)^\alpha_\rho
      \,=\,
      \delta^\alpha_\rho
    \Big)
    \\
    \big(
      \Omega^{(0)}
    \big)^{2' a}
    &=&
    \tfrac{r}{N^{2/6}} 
    \, 
    \mathrm{d}X^a
    &
      \;\;\;\;\;
      \Leftrightarrow
      \;\;\;\;\;
    &
    \Big(
      \big(
        \Omega^{(0)}
      \big)^{2' a}_r
      \,=\,
      \tfrac{r}{N^{2/3}}
      ,
      &
      \big(
        \Omega^{(0)}
      \big)^{2' a}_\rho
      \,=\,
      0\;
      \Big)
      \,.
  \end{array}
\end{equation}

\vspace{2mm} 
\noindent Moreover, inserting the flux density \eqref{OnAdS4ScaleFactorOfCFieldFlux} into the super-field strength components \eqref{HK} yields
\begin{equation}
  \label{OddComponentsOfFieldStrengthsOnAdS4}
  \def\arraystretch{1.4}
  \begin{array}{ccl}
    H_a
    &=&
    +
    \tfrac{c}{6}
    \,
    \tfrac{1}{2}
    \epsilon_{a \, b_1 b_2}
    \,
    \Gamma^{b_1 b_2}
    \Gamma^{2'}
    \\
    H_{2'}
    &=&
    -
    \tfrac{c}{6}
    \,
    \Gamma_{012}
    \\
    H_i 
    &=&
    -\tfrac{c}{12}
    \Gamma_{i} 
    \Gamma_{0 1 2} 
    \Gamma_{2'}
    \\[6pt]
    K^{a_1 a_2}
    &=&
    -\tfrac{c}{3}
    \epsilon^{a_1 a_2 a_3} 
    \Gamma_{a_3} 
    \,
    \Gamma_{2'}
    \\
    K^{2' a}
    &=&
    +\tfrac{c}{6}
    \epsilon^{a \, b_1 b_2}
    \Gamma_{b_1 b_2}
    \\
    K^{i_1 i_2}
    &=&
    +
    \,
    c
    \,
    \Gamma^{i_1 i_2}
    \Gamma^{0 1 2}
    \Gamma^{2'}
    \\
    K^{i a}
    &=&
    0
    \\
    K^{2' i}
    &=&
    0\,.
  \end{array}
  \hspace{1.5cm}
  \mbox{for}
  \hspace{1cm}
  \def\arraystretch{1.2}
  \begin{array}{l}
    a_i \in \{0,1,2\}
    \\
    \,i_i \in \{3,4,5,6,7,8,9\}
  \end{array}
\end{equation}

From this, we now obtain the super-field extension of the supergravity fields on $\mathrm{AdS}_4 \times S^7$:

\begin{example}[\bf $\mathrm{AdS}_4 \times S^7$-super-fields to first $\Theta$-order]
\label{AdS4SuperFieldsToFirstThetaOrder}
Based on the 0th-order expressions \eqref{InitialValuesInAdS4},
we obtain to first order in $\Theta$ (similar to \cite[(3.11)]{dWPPS98}):
\begin{equation}
  \label{AdS4SuperfieldsToFirstOrder}
  \hspace{-4.3mm} 
  \def\arraystretch{1.8}
  \begin{array}{lrcrcrlll}
  E^a
    \, \;=
  &
  \tfrac{r}{N^{1/6}}
  \,
  \mathrm{d}X^a
  &+&
  \big(\,
    \overline{\Theta}
    \,\Gamma^a\,
    \mathrm{d}\Theta
  \big)
  &+&
  \mathcal{O}\big(\Theta^2\big)
  &
  \proofstep{
    by \eqref{RecursionRelationForGraviton}
  }
  \\
  E^{2'}\;=
  &
  \tfrac{N^{\fixed{2}/6}}{r}
  \,
  \mathrm{d}X^{2'}
  &+&
  \big(\,
    \overline{\Theta}
    \,\Gamma^{2'}\,
    \mathrm{d}\Theta
  \big)
  &+&
  \mathcal{O}\big(\Theta^2\big)
  &
  \proofstep{
    by \eqref{RecursionRelationForGraviton}
  }
  \\
  \Omega^{2' a}=
  &
  \tfrac{r}{N^{2/6}}
  \,
  \mathrm{d}X^a
  \,
  &+&
  \tfrac{c}{6}
  \,
  \tfrac{1}{2}
  \epsilon^{a\, b_1 b_2}
  \big(\,
    \overline{\Theta}
    \,
    \Gamma_{b_1 b_2}
    \,
    \mathrm{d}\Theta
  \big)
  &+&
  \mathcal{O}\big(\Theta^2\big)
  &
  \proofstep{
    by
    \eqref{RecursionForSpinConnection},
    \eqref{OddComponentsOfFieldStrengthsOnAdS4}
    \&
    \eqref{GaugeConditions}
  }
  \\
  \Psi^\alpha
  \;\;= 
  &
  \mathrm{d}\Theta^\alpha
  &+&
  \big(
    -
    \tfrac{1}{2}
    \tfrac{r}{N^{2/6}}
    (\Gamma_{2' a}\Theta)^\alpha
    \,+\,
    \tfrac{c}{6}
    \fixed{
      \tfrac
        {r}
        {N^{1/6}}
    }
    \tfrac{1}{2}
    \epsilon_{a\, b_1 b_2}
    (
      \Gamma^{b_1 b_2} \Gamma^{2'}
      \Theta
    )^\alpha
  \big)
  \,
  \mathrm{d}X^a
  \\
  &
  &+&
  \tfrac{c}{6}
  \fixed{
    \tfrac
      {N^{2/6}}
      {r}
  }
  \tfrac{1}{2}
  \epsilon_{a\, b_1 b_2}
  (
    \Gamma^{b_1 b_2}
    \Gamma^{2'}
    \Theta
  )^\alpha
  \;
  \mathrm{d}X^{2'}
  \\
  &
  &+&
  -\tfrac{c}{12}
  (
    \Gamma_{i}
    \Gamma_{012}
    \Gamma_{2'}
    \Theta
  )^\alpha
  \,
  \mathrm{d}X^i
  &+&
  \mathcal{O}\big(
    \Theta^2
  \big)
  &
  \proofstep{
    by
    \eqref{RecursionRelationForGravitino}
    \&
    \eqref{OddComponentsOfFieldStrengthsOnAdS4},
  }
  \end{array}
\end{equation}
where $a \in \{0,1,2\}$.
\end{example}

\medskip

\subsection{Holographic M2 immersion}
\label{HolographicM2Immersions}

With the background super-spacetime in hand (\S\ref{SuperAdS4}), we are ready to inspect the $\sfrac{1}{2}$BPS super-immersions of holographic M2-branes. We shall be content here with showing the analysis to first order in the odd coordinates (cf. the analogous discussion for M5-branes to all order in \S\ref{TheHolographicM5Immersion}).

\begin{definition}[\bf Holographic M2 super-immersion]
\label{HolograhicM2SuperImmersion}
We extend the holographic immersion \eqref{OrdinaryHolographicM2Embedding} to a super-immersion (Def. \ref{SuperImmersion}) 
in an evident way:
\begin{equation}
  \label{TheM2SuperImmersion}
  \begin{tikzcd}[
    column sep=-5pt,
    row sep=-6pt
  ]
    \mathbb{R}^{
      1,2 \,\vert\, 8 \cdot \mathbf{2}
    }
    \ar[
      rr,
      "{\phi}"
    ]
    &\hspace{30pt}&
    \mathbb{R}^{1,2}
    &\times& 
    \mathbb{R}_{> 0}
    &\times&
    \mathbb{D}^7
    &\times& 
    \mathbb{R}^{0 \vert 8 \cdot \mathbf{2}}
    &\times& 
    \mathbb{R}^{0 \vert 8 \cdot \mathbf{2}}
    &
    \;\;\;\;\;
    \subset
    \; \mathbb{R}^{1,10 \,\vert\, \mathbf{32}}
    \\
    x^a
    &\overset{\phi^\ast}{\longmapsfrom}&
    X^a
    \\
    r_{\mathrm{prb}}
    &\overset{\phi^\ast}{\longmapsfrom}&
    &&
    r
    \\
    s^i_{\mathrm{prb}}
    &\overset{\phi^\ast}{\longmapsfrom}&
    &&
    &&
    X^i
    \\
    \theta^\alpha
    &\overset{\phi^\ast}{\longmapsfrom}&
    &&
    &&
    &&
    \big(P\Theta\big)^\alpha
    \\
    0
    &\overset{\phi^\ast}{\longmapsfrom}&
    &&
    &&
    &&
    &&
    \big(
      \Porth\Theta
    \big)^\alpha
    .
  \end{tikzcd}
\end{equation}
\end{definition}

\begin{lemma}[\bf M2-Worldvolume super-fields to first $\theta$-order]
  \label{M2WorldvolumeSuperFieldsToFirstOrder}
 $\,$
  
\noindent {\bf (i)} Under the holographic super-immersion \eqref{HolograhicM2SuperImmersion}, the first-order super-fields \eqref{InitialValuesInAdS4} pull back to

  \begin{equation}
    \label{M2WorldvolumeFieldsToFirstThetaOrder}
    \hspace{-8mm}
    \def\arraystretch{1.6}
    \begin{array}{rcccrcrcl}
      e^a 
      &:=&
      \phi^\ast E^a
      &=&
      \tfrac
        { r_{\mathrm{prb}} }
        { N^{1/6} }
      \,
      \mathrm{d}x^a
      &+&
      \big(\hspace{.8pt}
        \overline{\theta}
        \,\gamma^a\,
        \mathrm{d}\theta
      \big)
      &+&
      \mathcal{O}\big(
        \theta^2
      \big)
      \\
      &&
      \phi^\ast
      E^{2'}
      &=&
      &&
      &&
      \mathcal{O}\big(\theta^2\big)
      \\
      \psi^\alpha
      &:=&
      \phi^\ast (P \Psi)^\alpha
      &=&
      &&      
      \mathrm{d}\theta^\alpha
      &+&
      \mathcal{O}\big(
        \theta^2
      \big)
      \\
      &&
      \phi^\ast\big(
        \hspace{.7pt}
        \Porth\Psi
      \big)^\alpha
      &=&
      \big(
        \tfrac{1}{2}
        \tfrac
          {r_{\mathrm{prb}}}
          { N^{2/6} }
        -
        \tfrac{c}{6}
        \fixed{
          \tfrac
            {r_{\mathrm{prb}}}
            {N^{1/6}}
        }
      \big)
      (
        \Gamma_{a2'}
        \theta
      )^\alpha
      \,
      \mathrm{d}x^a
      &&
      &+&
      \mathcal{O}\big(
        \theta^2
      \big)
      \\
      e^{b_1}
      \SecondFundamentalForm
        ^{ 2' }_{ b_1 b_2 }
      \,+\,
      \psi^\beta
      \SecondFundamentalForm
        ^{ 2' }_{ \beta\,  b_2 }
      &:=&
      \phi^\ast
      \Omega^{2'}{}_{b_2}
      &=&
      \tfrac
        { r_{\mathrm{prb}} }
        { N^{2/6} }
      \delta_{b_1 b_2}
      \,
      \mathrm{d}x^{b_1} 
      &+&
      \tfrac{1}{6}\big(\hspace{.8pt}
        \overline{\theta}
        \,\gamma^a\,
        \mathrm{d}\theta
      \big)
      &+&
      \mathcal{O}\big(\theta^2\big)
      \mathrlap{\,.}
    \end{array}
  \end{equation}
 \noindent {\bf (ii)} The 2nd fundamental super-form
  $\SecondFundamentalForm^{2'}$
  {\rm (cf. \cite[(67)]{GSS24-FluxOnM5})}
  has the following components:
  \begin{equation}
    \label{M2ScndFundamentalFormToFirstOrder}
    \def\arraystretch{1.5}
    \begin{array}{cclcl}
      \SecondFundamentalForm
        ^{2'}
        _{b_1 b_2}
      &=&
      \tfrac{1}{N^{1/6}}
      \delta_{b_1 b_2}
      &+&
      \mathcal{O}
      \big(
        \theta^2
      \big)
      \\
      \SecondFundamentalForm
        ^{2'}
        _{\beta \, b_2}
      &=&
      \big(
        \tfrac{1}{N^{1/6}}
        \,-\,
        \tfrac{c}{3}
      \big)
      \big(\hspace{.8pt}
        \overline{\theta}\,\gamma_{b_1}
      \big)_\beta
      &+&
      \mathcal{O}
      \big(
        \theta^2
      \big)
      \,.
    \end{array}
  \end{equation}
\end{lemma}

\begin{proof}
  The first line in \eqref{M2WorldvolumeFieldsToFirstThetaOrder} is evident. For the second line just note that $\big(\,\overline{\theta}\,\Gamma^{2'}\mathrm{d}\theta\big) \,=\, \big(\overline{\theta}\, P \Gamma^{2'}\! P\,  \mathrm{d}\theta \big) \,=\, 0$ by \eqref{ProjectionOperatorRelationsOn2Brane}.
  For the third line notice similarly that
  $$
    \def\arraystretch{1.6}
    \begin{array}{l}
      \phi^\ast (P \Psi)^\alpha
      \;=\;
      \mathrm{d}\theta^\alpha
      \,+\,
      \Big(
        -\tfrac{1}{2}
        \tfrac
          { r_{\mathrm{prb}} }
          { N^{2/6} }
        \underbrace{
          (P\,\Gamma_{2' a} \theta)
            ^{\mathrlap{\alpha}}
        }_{\color{gray}  = 0}
        \;\;
        +
        \tfrac{c}{6}
        \,
        \tfrac{1}{2}
        \epsilon_{a\, b_1 b_2}
        \underbrace{
          (
            P \Gamma^{b_1 b_2}
            \Gamma^{2'}\!
            \theta
          )^{\mathrlap{\alpha}}
        }_{\color{gray} = 0 }
        \;\;
      \Big)
      \,
      \mathrm{d}x^a
      \mathrlap{\,,}
    \end{array}
  $$

  \vspace{1mm}
  \noindent
  where the terms over the braces vanish by \eqref{ProjectionOperatorRelationsOn2Brane}: 
  $$
    P \Gamma_{2' a} \theta
    \,=\,
    P \Gamma_{2' a} P \theta
    \,=\,
    \Gamma_{2' a} \Porth P \theta
    \,=\,
    0\,,
    \;\;\;\;
    \mbox{and}
    \;\;\;\;
    P\Gamma^{b_1 b_2}\Gamma^{2'}\! \theta
    \,=\,
    P\Gamma^{b_1 b_2}\Gamma^{2'}\! P\theta
    \,=\,
    \Gamma^{b_1 b_2}\Gamma^{2'}\! 
    \Porth P\theta
    \,.
  $$
The fourth line works analogously but complementarily:
\begin{equation}
  \label{PullbackOfTransvereOddCoframeToHolographicM2}
  \hspace{-5mm}
  \def\arraystretch{1.6}
  \begin{array}{llll}
    \phi^\ast\big(
      \Porth
      \Psi
    \big)^\alpha
    &
    \;=\;
    \big(
      -\tfrac{1}{2}
      \tfrac
        { r_{\mathrm{prb}} }
        {N^{2/6}}
      (
        \Porth 
        \,
        \Gamma_{2' a}\theta
      )^\alpha
      \,+\,
      \tfrac{c}{6}
        \fixed{
          \tfrac
            {r_{\mathrm{prb}}}
            {N^{1/6}}
        }
      \tfrac{1}{2}
      \epsilon_{a \, b_1 b_2}
      (
        \Porth
        \,
        \Gamma^{b_1 b_2}
        \Gamma^{2'}\!\theta
      )^\alpha
    \big)
    \mathrm{d}x^a
    &
    +
    \mathcal{O}(\theta^2)
    &
    \proofstep{
      by 
      \eqref{TheM2SuperImmersion}
      \& \eqref{AdS4SuperfieldsToFirstOrder}
    }
    \\
    &\;=\;
    \big(
      +\tfrac{1}{2}
      \tfrac
        { r_{\mathrm{prb}} }
        {N^{2/6}}
      (
        \Gamma_{a}
        \Gamma_{2'}\theta
      )^\alpha
      \,-\,
      \tfrac{c}{6}
        \fixed{
          \tfrac
            {r_{\mathrm{prb}}}
            {N^{1/6}}
        }
      (
        \Gamma_a
        \Gamma_{2'}\theta
      )^\alpha
    \big)
    \mathrm{d}x^a
    &
    + \mathcal{O}(\theta^2)
    &
    \proofstep{
      by
      \eqref{ProjectionOperatorRelationsOn2Brane}
      \&
      \eqref{3HodgeDuality}
    }.
  \end{array}
\end{equation}
\end{proof}

\vspace{-.2cm}
Hence in analogy with Rem. \ref{CriticalDistanceForHolographicM5} we find:

\begin{remark}[\bf \fixedtext{Posititve charge of} holographic M2-brane probe]
\label{CriticalDistanceForHolographicM2}
The expression \eqref{PullbackOfTransvereOddCoframeToHolographicM2} must vanish for the immersion to be $\sfrac{1}{2}$BPS.
With $c = \pm 3/N^{1/6}$ \eqref{OnAdS4ScaleFactorOfCFieldFlux}
this requires again that we have the positive sign in 
\eqref{OnAdS4ScaleFactorOfCFieldFlux}
\begin{equation}
  \label{PositiveChargeForM2Brane}
  G_4
  \;=\;
  +
  \tfrac{3}{N^{1/6}}
  E^0 E^1 E^2 E^3
\end{equation}
analogous to the case of holographic M5-branes (Rem. \ref{CriticalDistanceForHolographicM5}).
\end{remark}

\section{Conclusion}

Motivated by an old proposal, due to Duff et al., that holography is microscopically realized by fluctuations of branes probing their own near-horizon geometry, we have constructed the $\sfrac{1}{2}$BPS super-embedding of an M5-brane super-worldvolume parallel to the horizon of super-$\mathrm{AdS}_7 \times S^4$-spacetime (which in fact seems to be the first explicit example of any non-trivial M5 super-embedding to be recorded). 

\smallskip 
In doing so, we have set to zero the transverse worldvolume gravitino field, $\tau_a = 0$, highlighting (Rem. \ref{RelationOfSuperembeddingToTheLiterature}) that only then is the Bianchi identity of the self-dual worldvolume flux density as expected (namely without fermionic correction) and is the super-embedding really $\sfrac{1}{2}$BPS in a sense which we made precise in Def. \ref{HalfBPSSuperImmersion}, both following \cite{GSS24-FluxOnM5}.

\smallskip
Under this condition, we have shown in Thm. \ref{ExistenceOfHolographicM5BraneProbes} --- by inductively solving the super-embedding conditions in super-normal coordinates ---
that there is a unique such holographic super-embedding (up to the irrelevant choice of a direction on the surrounding 4-sphere) for any radial distance from the asymptotic boundary. 

\smallskip 
This super-embedding should, therefore, serve as the {\it super-space} background configuration around which to compute, after super-diffeomorphism gauge fixing, the desired super-conformal fluctuations. We hope to discuss this elsewhere.

\medskip

\appendix

\section{Appendix}

\subsection{Tensor conventions and 11d spinors}
\label{TensorConventionsAnd11dSpinors}

\noindent
{\bf Tensor conventions.}
Our tensor conventions are standard, but since the computations below crucially depend on the corresponding prefactors, here to briefly make them explicit:\begin{itemize}[leftmargin=.4cm]
\item
  The Einstein summation convention applies throughout: Given a product of terms indexed by some $i \in I$, with the index of one factor in superscript and the other in subscript, then a sum over $I$ is implied:
  $
    x_i \, y^i
    :=
    \sum_{i \in I} 
    x_i \, y^i
  $.

\item
Our Minkowski metric is the matrix
\begin{equation}
  \label{MinkowskiMetric}
  \big(\eta_{ab}\big)
    _{a,b = 0}
    ^{ d }
  \;\;
    =
  \;\;
  \big(\eta^{ab}\big)
    _{a,b = 0}
    ^{ d }
  \;\;
    :=
  \;\;
  \big(
    \mathrm{diag}
      (-1, +1, +1, \cdots, +1)
  \big)_{a,b = 0}^{d} \;.
\end{equation}
\item
  Shifting position of frame indices always refers to contraction with the  Minkowski metric \eqref{MinkowskiMetric}:
  $$
    V^a 
      \;:=\;
    V_b \, \eta^{a b}
    \,,
    \;\;\;\;
    V_a \;=\; V^b \eta_{a b}
    \,.
  $$
\item Skew-symmetrization of indices is denoted by square brackets ($(-1)^{\vert\sigma\vert}$ is sign of the permutation $\sigma$):
$$
  V_{[a_1 \cdots a_p]}
  \;:=\;
  \tfrac{1}{p!}
  \sum_{
    \sigma \in \mathrm{Sym}(n)
  }
  (-1)^{\vert \sigma \vert}
  V_{ a_{\sigma(1)} \cdots a_{\sigma(p)} }\,.
$$
\item
We normalize the Levi-Civita symbol to \begin{equation}
  \label{transversalizationOfLeviCivitaSymbol}
  \epsilon_{0 1 2 \cdots} 
    \,:=\, 
  +1
  \;\;\;\;\mbox{hence}\;\;\;\;
  \epsilon^{0 1 2 \cdots} 
    \,:=\, 
  -1
  \,.
\end{equation}
\item
We normalize the Kronecker symbol to
\begin{equation}
  \label{KroneckerSymbol}
  \delta
    ^{a_1 \cdots a_p}
    _{b_1 \cdots b_p}
  \;:=\;
  \delta^{[a_1}_{[b_1}
  \cdots
  \delta^{a_p]}_{b_p]}
  \;=\;
  \delta^{a_1}_{[b_1}
  \cdots
  \delta^{a_p}_{b_p]}
  \;=\;
  \delta^{[a_1}_{b_1}
  \cdots
  \delta^{a_p]}_{b_p}
\end{equation}
so that
\begin{equation}
  \label{ContractingKroneckerWithSkewSymmetricTensor}
  V_{
    \color{darkblue}
    a_1 \cdots a_p
  }
  \tensor*
    {\delta}
    {
    ^{ 
       \color{darkblue}
       a_1 \cdots a_p 
    }
    _{b_1 \cdots b_p}
    }
  \;\;
  =
  \;\;
  V_{[b_1 \cdots b_p]}  
  \;\;\;\;
  \mbox{and}
  \;\;\;\;
  \epsilon^{
    {\color{darkblue}  
      c_1 \cdots c_p
    }
    a_1 \cdots a_q
  }
  \,
  \epsilon_{
    {\color{darkblue}
    c_1 \cdots c_p 
    }
    b_1 \cdots b_q
  }
  \;\;
  =
  \;\;
  -
  \,
  p! \cdot q!
  \,
  \delta
    ^{a_1 \cdots a_q}
    _{b_1 \cdots b_q}
  \,.
\end{equation}
\end{itemize}

\medskip

\noindent
{\bf Spinors in 11d.}
We briefly recall the following standard facts (proofs and references are given in \cite[\S 2.2.1]{GSS24-SuGra}):
There exists an $\mathbb{R}$-linear representation $\mathbf{32}$ of $\mathrm{Pin}^+(1,10)$ with generators
\begin{equation}
  \label{The11dMajoranaRepresentation}
  \Gamma_a 
  \;:\;
  \mathbf{32}
  \xrightarrow{\;\;}
  \mathbf{32}
\end{equation}
and equipped with a skew-symmetric bilinear form
\begin{equation}
  \label{TheSpinorPairing}
  \big(\hspace{.8pt}
    \overline{(-)}
    (-)\,
  \big)
  \;:\;
  \mathbf{32}
  \otimes
  \mathbf{32}
  \xrightarrow{\quad}
  \mathbb{R}
\end{equation}
with the following properties, where as 
usual we denote skew-symmetrized product of $k$ Clifford generators by
  \begin{equation}
    \label{CliffordBasisElements}
    \Gamma_{a_1 \cdots a_k}
    \;:=\;
    \tfrac{1}{k!}
    \underset{
      \sigma \in
      \mathrm{Sym}(k)
    }{\sum}
    \mathrm{sgn}(\sigma)
    \,
    \Gamma_{a_{\sigma(1)}}
    \cdot
    \Gamma_{a_{\sigma(2)}}
    \cdots
    \Gamma_{a_{\sigma(n)}}
    :
  \end{equation}

\begin{itemize}[leftmargin=.4cm]
  \item
  The Clifford generators square to plus the Minkowski metric \eqref{MinkowskiMetric}
  \begin{equation}\label{CliffordDefiningRelation}
    \Gamma_a
    \Gamma_b
    +
    \Gamma_b
    \Gamma_a
    \;\;=\;\;
    +2 \, \eta_{a b}
    \,
    \mathrm{id}_{\mathbf{32}}
    \,.
  \end{equation}

  \item The Clifford product is given on the basis elements \eqref{CliffordBasisElements}
  as
\begin{equation}
  \label{GeneralCliffordProduct}
  \Gamma^{a_j \cdots a_1}
  \,
  \Gamma_{b_1 \cdots b_k}
  \;=\;
  \sum_{l = 0}^{
    \mathrm{min}(j,k)
  }
  \pm
  l!
\binom{j}{l}
 \binom{k}{l}
  \,
  \delta
   ^{[a_1 \cdots a_l}
   _{[b_1 \cdots b_l}
  \Gamma^{a_j \cdots a_{l+1}]}
  {}_{b_{l+1} \cdots b_k]}
  \,.
\end{equation}
  
  \item 
  The Clifford volume form equals the Levi-Civita symbol 
  \eqref{transversalizationOfLeviCivitaSymbol}:
  \begin{equation}
    \label{CliffordVolumeFormIn11d}
    \Gamma_{a_1 \cdots a_{11}}
    \;=\;
    \epsilon_{a_1 \cdots a_{11}}
    \mathrm{id}_{\mathbf{32}}
    \,.
  \end{equation}
  \item The Clifford generators are skew self-adjoint with respect to the pairing \eqref{TheSpinorPairing}
  \vspace{1mm} 
  \begin{equation}
    \label{SkewSelfAdjointnessOfCliffordGenerators}
    \overline{\Gamma_a}
    \;=\;
    - \Gamma_a
    \;\;\;\;\;\;
    \mbox{in that}
    \;\;\;\;\;\;
    \underset{
      \phi,\psi \in \mathbf{32}
    }{\forall}
    \;\;
    \big(\,
      \overline{(\Gamma_a \phi)}
      \,
      \psi
    \big)
    \;=\;
    -
    \big(\,
      \overline{\phi}
      \,
      (\Gamma_a \psi)
    \big)
    \,,
  \end{equation}
  so that generally
  \vspace{1mm} 
  \begin{equation}
    \label{AdjointnessOfCliffordBasisElements}
    \overline{\Gamma_{a_1 \cdots a_p}}
    \;=\;
    (-1)^{
      p + p(p-1)/2
    }
    \,
    \Gamma_{a_1 \cdots a_p}
    \,.
  \end{equation}

  \item
  The $\mathbb{R}$-vector space of $\mathbb{R}$-linear  endomorphisms of $\mathbf{32}$ has a linear basis given by the $\leq 5$-index Clifford elements 
  \vspace{1mm} 
  \begin{equation}
    \label{CliffordElementsSpanningLinearMaps}
    \mathrm{End}_{\mathbb{R}}\big(
      \mathbf{32}
    \big)
    \;\;
    =
    \;\;
    \big\langle
      1
      ,\,
      \Gamma_{a_1}
      ,\,
      \Gamma_{a_1 a_2}
      ,\,
      \Gamma_{a_1 a_2 a_3}
      ,\,
      \Gamma_{a_1 \cdots a_4}
      ,\,
      \Gamma_{a_1 \cdots a_5}
    \big\rangle_{
      a_i = 0, 1, \cdots
    }
    \,.
  \end{equation}

  \item
  The $\mathbb{R}$-vector space space of {\it symmetric} bilinear forms on $\mathbf{32}$
  has a linear basis given by the expectation values with respect to \eqref{TheSpinorPairing} of the 1-, 2-, and 5-index Clifford basis elements:
  \begin{equation}
    \label{SymmetricSpinorPairings}
    \mathrm{Hom}_{\mathbb{R}}
    \Big(
    (\mathbf{32}\otimes \mathbf{32})_{\mathrm{sym}}
    ,\,
    \mathbb{R}
    \Big)
    \;\;
    \simeq
    \;\;
    \Big\langle
    \big(
      (\overline{-})
      \Gamma_a
      (-)
    \big)
    \,,\;\;
    \big(
      (\overline{-})
      \Gamma_{a_1 a_2}
      (-)
    \big)
    \,,\;\;
    \big(
      (\overline{-})
      \Gamma_{a_1 \cdots a_5}
      (-)
    \big)
    \Big\rangle_{
      a_i = 0, 1, \cdots
      \,,
    }
  \end{equation}
  while a basis for the skew-symmetric bilinear forms is given by
  \begin{equation}
    \label{SkewSpinorPairings}
    \mathrm{Hom}_{\mathbb{R}}
    \Big(
    (\mathbf{32}\otimes \mathbf{32})_{\mathrm{skew}}
    ,\,
    \mathbb{R}
    \Big)
    \;\;
    \simeq
    \;\;
    \Big\langle
    \big(
      (\overline{-})
      (-)
    \big)
    \,,\;\;
    \big(
      (\overline{-})
      \Gamma_{a_1 a_2 a_3}
      (-)
    \big)
    \,,\;\;
    \big(
      (\overline{-})
      \Gamma_{a_1 \cdots a_4}
      (-)
    \big)
    \Big\rangle_{
      a_i = 0, 1, \cdots
      \,.
    }
  \end{equation}

\item
  Any linear endomorphism $\phi \in \mathrm{End}_{\mathbb{R}}(\mathbf{32})$ is uniquely a linear combination of Clifford elements as:
  \begin{equation}
    \label{CliffordExpansionOfEndomorphismOf32}
    \phi
      \;=\;
    \tfrac{1}{32}
    \sum_{p = 0}^5
    \;
    \tfrac{
      (-1)^{p(p-1)/2}
    }{ p! }
    \mathrm{Tr}\big(
      \phi \circ 
      \Gamma_{a_1 \cdots a_p}
    \big)
    \Gamma^{a_1 \cdots a_p}
    \,,
    \hspace{1cm}
    \scalebox{0.9}{$ a _i \in \{0,\!\cdots\!,5',6,7,8,9\} $}
    \,.
  \end{equation}

\end{itemize}

\medskip

\noindent
{\bf Background formulas for 11d Supergravity.}
Our notation and conventions for super-geometry and for on-shell 11d supergravity on super-space follow \cite[\S 2.2 \& \S 3]{GSS24-SuGra}, to which we refer for further details and exhaustive referencing.

We denote the local data of a super-Cartan connection on (a surjective submersion $\CoverOf{X}$ of) (super-)spacetime $X$, representing a super-gravitational field configuration, as\footnote{
  Our use of different letters for the even and odd components of a super co-frame follows e.g. \cite{CDF91}. Other authors write ``$E^\alpha$'' for what we denote ``$\Psi^\alpha$'', e.g. \cite{BandosSorokin23}. While it is of course part of the magic of supergravity that $E^a$ and $E^\alpha$/$\Psi^\alpha$ are unified into a single super-coframe field $E$, we find that for reading and interpreting formulas it is helpful to use different symbols for its even and odd components.
}
\begin{equation}
  \label{LocalCartanConnection}
  \def\arraystretch{1.7}
  \begin{array}{clcl}
    \scalebox{.7}{
      \color{darkblue}
      \bf
      Graviton
    }
    &
    \big(
      E^a
    \big)_{a=0}^{D-1}
    &\in&
    \Omega^1_{\mathrm{dR}}\big(
      \CoverOf{X}
      ;\,
      \mathbb{R}^{1,D-1}
    \big)
    \\
    \scalebox{.7}{
      \color{darkblue}
      \bf
      Gravitino
    }
    &
    \big(
      \Psi^\alpha
    \big)_{\alpha=1}^{N}
    &\in&
    \Omega^1_{\mathrm{dR}}\big(
      \CoverOf{X}
      ;\,
      \mathbf{N}_{\mathrm{odd}}
    \big)
    \\
    \scalebox{.7}{
      \color{darkblue}
      \bf
      \def\arraystretch{.9}
      \begin{tabular}{c}
        Spin-
        \\
        connection
      \end{tabular}
    }
    &
    \big(
      \Omega^{ab}
      =
      -\Omega^{b a}
    \big)_{a,b = 0}^{D-1}
    &\in&
    \Omega^1_{\mathrm{dR}}\big(
      \CoverOf{X}
      ;\,
      \mathfrak{so}(1,D-1)
    \big)
  \end{array}
\end{equation}

\vspace{1mm} 
\noindent and the corresponding Cartan structural equations (cf. \cite[Def. 2.78]{GSS24-SuGra}) for the supergravity field strengths as
\begin{equation}
  \label{GravitationalFieldStrengths}
  \def\arraystretch{1.7}
  \begin{array}{ccccll}
    \scalebox{.7}{
      \color{darkblue} 
      \bf
      \def\arraystretch{.9}
      \begin{tabular}{c}
        Super-
        \\
        Torsion
      \end{tabular}
    }
    &
    \big(
    \,
    T^a
    &:=&
    \differential
    \,
        E^a
    &
    - \;
    \Omega^{a}{}_b \, E^b
    -
    (\,
      \overline{\Psi}
      \,\Gamma^a\,
      \Psi
    )
    \,
    \big)_{a=0}^{D-1}
    \\
    \scalebox{.7}{ \color{darkblue} \bf 
      \def\arraystretch{.9}
      \def\tabcolsep{-.2cm}
      \begin{tabular}{c}
        Gravitino
        \\
        field strength
      \end{tabular}
    }
    &
    \big(
    \,
    \rho
    &:=&
    \differential 
    \,
        \Psi
    &
    -\;
    \tfrac{1}{4}
    \Omega^{a b} 
     \,
    \Gamma_{a b}
    \psi
    \,
    \big)_{\alpha=1}^{N}
    \\
    \scalebox{.7}{
      \color{darkblue} 
      \bf 
      Curvature
    }
    &
    \big(
    \,
    R^{ab}
    &:=&
    \differential
    \,
      \Omega^{a b}
    &
    -\;
    \Omega^a{}_c 
    \,
    \Omega^{c b}
    \,
    \big)_{a,b = 0}^{D-1}\,.
  \end{array}
\end{equation}
Finally, we denote the corresponding components in the given 
local super-coframe $(E,\Psi)$ by \cite[(127-8)]{GSS24-SuGra}:
\begin{equation}
  \label{CoFrameComponentsOfSuperFieldStrengths}
  \def\arraystretch{1.3}
  \begin{array}{lcl}
    T^a &\defneq& 0
    \\
    \rho &=:& 
    \tfrac{1}{2}
    \rho_{a b}
    \,
    E^a \, E^b
    \,+\,
    H_a \Psi\, E^a
    \\
    R^{a_1 a_2}
    &=:&
    \tfrac{1}{2}
    R^{a_1 a_2}{}_{b_1 b_2}
    \,
    E^{a_1}\, E^{a_2}
    \,+\,
    \big(\hspace{.8pt}
      \overline{J}^{a_1 a_2}{}_b
      \Psi
    \big) E^b
    \,+\,
    \big(\hspace{.8pt}
      \overline{\Psi}
      \,K^{a_1 a_2}\,
      \Psi
    \big)
    \,,
  \end{array}
\end{equation}
where all components not explicitly appearing vanish identically by the superspace torsion constraints \cite[(121), (137)]{GSS24-SuGra}. In addition, in the main text we consider the situation that also $\rho_{a b} = 0$  \eqref{VanishingGravitinoFieldStrength} whence also $J^{a_1 a_2}{}_b = 0$ \eqref{CoFrameComponentsOfSuperFieldStrengthsOnGravitinoFlat}.

\subsection{Holographic M2 Equation of Motion}
\label{M2EquationsOfMotion}

We compare here to existing computations via (just) the equations of motion.

\smallskip 
Holographic embeddings of the M2-brane have originally been discussed in \cite{BDPS87}, where a critical radius $r_{\mathrm{prn}} = \infty$ was found for static embeddings with respect to the ``static coordinate'' chart (reproduced as Ex. \ref{HolographicM2ProbeEOMInStaticChart}).
Then \cite[(2.23) \& \S A]{CKKTvP98} observed that there are different static holographic embeddings compatible with different AdS-coordinate charts, and found no critical radius for static embeddings with respect to the Poincar{\'e} chart (reproduced as Ex. \ref{CriticalM2ProbeDistanceFromEOMInPoincareChart} below).
Note that these existing computation use (only) the equations of motion, not a super-embedding (cf. {\it The need for probe brane super-embeddings} on p. \pageref{NeedForSuperEmbeddings}).

\medskip

\noindent
{\bf The bosonic equation of motion for the M2-brane} (as given e.g. in \cite[(4)]{BDPS87}) is the following (in this subsection we stick to these authors' notation, for ease of comparison \footnote{Except that we change the sign of the third summand in \eqref{BosonicM2EOMByDuffEtAl} compared to \cite[(4)]{BDPS87}, hence equivalently of the normalization of the sign of the flux density, as was also done in \cite[(A.5)]{CKKTvP98}.}):
\vspace{0mm}
\begin{align}
  \label{BosonicM2EOMByDuffEtAl}
  \hspace{-3mm} 
  \partial_i
  \Big(
    \sqrt{-h}
    \,
    h^{i j}
    (\partial_j X^N)
    g_{M N}
  \Big)
  &\,+\,
  \tfrac{1}{2}
  \sqrt{-h}
  \,
  h^{i j}
  (\partial_i X^N)
  (\partial_j X^P)
  (\partial_M g_{N P})
  \\ \nonumber 
  &\, -
  \tfrac{1}{6}
  \epsilon^{i j k}
  (\partial_i X^N)
  (\partial_j X^P)
  (\partial_k X^Q)
  \,
  F_{M N P Q}
  \;=\;
  0\;.
\end{align}

\begin{example}[\bf M2 probe EoM in Poincar{\'e} chart]
  \label{CriticalM2ProbeDistanceFromEOMInPoincareChart}
  Consider the case where $g$ is the $\mathrm{AdS}_4$-metric in Poincar{\'e} coordinates as in \eqref{MetricTensorForBlackM2}
  \vspace{1mm}
  $$
    \mathrm{d}s^2_g
    \;\defneq\;
    \tfrac{r^2}{N^{2/6}}
    \big(
      -
      \mathrm{d}X^0 \otimes \mathrm{d}X^0
      +
      \mathrm{d}X^1 \otimes \mathrm{d}X^1
      +
      \mathrm{d}X^2 \otimes \mathrm{d}X^2
    \big)
    \,+\,
    \tfrac{N^{2/6}}{r^2}
    \,
    \mathrm{d}r \otimes \mathrm{d}r
  $$
  and \eqref{OnAdS4ScaleFactorOfCFieldFlux}
  \begin{equation}
    \label{FluxDensityForM2InStaticCoordinates}
    F_{MNPQ} 
    \;\defneq\;
    \pm
    \tfrac{ 3 }{ N^{1/6} }
    \, 
    \sqrt{-g}
    \;
   \scalebox{1.3}{$\epsilon$}_{M N P Q}
    \,,
    \;\;\;\;\;\;
    \sqrt{-g}
    \;=\;
    \tfrac{r^2}{N^{2/6}}
    \,.
  \end{equation}
 
  \vspace{2mm} 
  \noindent With the holographic embedding  
  \eqref{OrdinaryHolographicM2Embedding}
  \vspace{0mm} 
  \begin{equation}
    \label{OrdinaryHolographicM2EmbeddingRecalled}
    \def\arraystretch{1.2}
    \begin{array}{rcc}
      X^0(x^0, x^1, x^2) 
        &\defneq& 
      x^0
      \\
      X^1(x^0, x^1, x^2) 
        &\defneq& 
      x^1
      \\
      X^2(x^0, x^1, x^2) 
        &\defneq& 
      x^2
      \\
      r\,(x^0, x^1, x^2) 
        &\defneq& 
      r_{\mathrm{prb}}\;,
    \end{array}
  \end{equation}
  the induced metric is
  \begin{equation}
    \mathrm{d}s^2_h
    \;=\;
    \tfrac{r_{\mathrm{prb}}^2}{N^{2/6}}
    \big(
      -
      \mathrm{d}x^0 \otimes \mathrm{d}x^0
      +
      \mathrm{d}x^1 \otimes \mathrm{d}x^1
      +
      \mathrm{d}x^2 \otimes \mathrm{d}x^2
    \big)    
    \,,
    \;\;\;\;\;\;\;
    \sqrt{-h}
    \;=\;
    \tfrac
      { r^3_{\mathrm{prb}} }
      { N^{3/6} }\;.
  \end{equation}
  and the corresponding radial component of the 
  equation of motion
  \eqref{BosonicM2EOMByDuffEtAl}
  becomes
  \begin{equation}
    \def\arraystretch{1.5}
    \begin{array}{lcl}
    0 &=&
  \partial_i
  \Big(
    \tfrac
      {r^3_{\mathrm{prb}}}
      {N^{3/6}}
    \,
    h^{i j}
    \underbrace{
      (\partial_j r)
    }_{ \color{gray} = 0 }
    \tfrac
      { N^{2/6} }
      { r^2_{\mathrm{prb}} }
  \Big)
  \,+\,
  \tfrac{1}{2}
    \tfrac
      {r^3_{\mathrm{prb}}}
      {N^{3/6}}
  \,
  \underbrace{
  \eta^{i j} \eta_{i j}
  }_{\color{gray} = 3 }
  \tfrac
    { N^{2/6} }
    { r^2_{\mathrm{prb}} }
  \tfrac
    { 2 r_{\mathrm{prb}} }
    { N^{2/6} }
  \,\pm\,
    \tfrac{1}{6}
    \tfrac{3}{N^{1/6}}
    \tfrac
      {r^2_{\mathrm{prb}}}
      {N^{2/6}}
    \,
    \underset{
      \color{gray}
      =\, -6
    }{
      \underbrace{
        \epsilon^{i j k}
        \epsilon_{i j k}
      }
    }
    \\[15pt]
    &=&
    3
      \Big(
        \tfrac
          { r_{\mathrm{prb}}^2 }
          { N^{3/6} }
        \,\mp\,
        \tfrac
          { r_{\mathrm{prb}}^2 }
          { N^{3/6} }
      \Big)
      \,,
    \end{array}
  \end{equation}
  which is equivalent to the flux density carrying the positive sign
  \vspace{1mm} 
  \begin{equation}
    F_{M N P Q}
    \;=\;
    +
    \tfrac{3}{N^{1/6}}
    \sqrt{-g}
    \, 
    \scalebox{1.3}{$\epsilon$}_{M N P Q}\;.
  \end{equation}
  and no further condition on $r_{\mathrm{prb}}$ (in accord with \cite[(2.33)]{CKKTvP98}).
\end{example}

\smallskip 
\begin{example}[\bf Holographic M2 probe EOM in static chart]
  \label{HolographicM2ProbeEOMInStaticChart}
  Consider the case where $g$ is the $\mathrm{AdS}_4$-metric in ``static coordinates'' \cite[(39.76)]{Blau22}\cite[(12)]{BDPS87}\cite[(A.4)]{CKKTvP98}
  \vspace{1mm} 
  $$
    \mathrm{d}s^2_g
    \;=\;
    -
    \big(
      1 + a^2 r^2
    \big)
    \,
    \mathrm{d}T \otimes \mathrm{d}T
    \,+\,
    r^2
    \,
      \mathrm{d}\Theta
      \otimes
      \mathrm{d}\Theta
    +
    r^2
      \mathrm{sin}^2(\Theta)
    \,
      \mathrm{d}\Phi
      \otimes
      \mathrm{d}\Phi
    \,+\,
    \big(
      1 + a^2 r^2
    \big)^{-1}
    \,
    \mathrm{d}r \otimes \mathrm{d}r
  $$
  and
  \vspace{1mm} 
  $$
    F_{M N P Q}
    \;\defneq\;
    \pm
    3 a \, \sqrt{-g}
    \;
    \scalebox{1.3}{$\epsilon$}_{M N P Q}
    \,,
    \;\;\;\;\;\;\;\;
    \mbox{
      for $a,r \in \mathbb{R}_{\geq 0}$
      \,,
    }
    \;\;\;\;\;\;\;
    \sqrt{-g} 
    \,=\,
    r^2 \mathrm{sin}(\Theta)
    \,.
  $$
  \vspace{1mm} 
  
  \noindent   
  With the corresponding static embedding now being \cite[(14)]{BDPS87}
  \begin{equation}
    \label{StaticM2EmbeddingInStaticChart}
    \def\arraystretch{1.3}
    \begin{array}{rcc}
      T(t, \theta, \phi) 
        &\defneq& 
      t
      \\
      \Theta(t, \theta, \phi) 
        &\defneq& 
      \theta
      \\
      \Phi(t, \theta, \phi) 
        &\defneq& 
      \phi
      \\
      r\,(t, \theta, \phi) 
        &\defneq& 
      r_{\mathrm{prb}}\;,
    \end{array}
  \end{equation}
  and hence intrisnically different from \eqref{OrdinaryHolographicM2EmbeddingRecalled},
  the induced metric is
  \vspace{1mm} 
  \begin{equation}
    \mathrm{d}s^2_h
    \;=\;
    -\big(1 + a^2 r^2_{\mathrm{prb}}\big)
    \, 
    \mathrm{d}t \otimes \mathrm{d}t
    \,+\,
    r^2_{\mathrm{prb}}
    \,
    \mathrm{d}\theta \otimes \mathrm{d}\theta
    \,+\,
    r^2_{\mathrm{prb}} 
    \, \mathrm{sin}^2(\theta)
    \, 
    \mathrm{d}\phi \otimes \mathrm{d}\phi
    \,,
    \;\;\;\;\;\;
    \sqrt{-h}
    \;=\;
    (1 + a^2 r^2_{\mathrm{prb}})^{1/2}
    \,
    r^2_{\mathrm{prb}}
    \,
    \mathrm{sin}(\Theta)
    \,.
  \end{equation}
  \vspace{2mm} 

  \noindent  
  Now the radial component of the 
  equation of motion
  \eqref{BosonicM2EOMByDuffEtAl} becomes
  \vspace{2mm}
  \begin{equation}
    \label{AnalogEquationOfBDPS87-15}
    \begin{array}{rcl}
      0 &=&
      r_{\mathrm{prb}}^2
      \,
      \mathrm{sin}(\Theta)
      \Big(
      \tfrac{1}{2}
      \,
      (
        1+ a^2 r_{\mathrm{prb}}^2
      )^{1/2}
      \Big(
        \tfrac{
          2 a^2 r_{\mathrm{prb}}
        }{
          1 + a^2 r^2_{\mathrm{prb}}
        }
        \,+\,
        \tfrac{2}{
          r_{\mathrm{prb}}
        }        
        \,+\,
        \tfrac{2}{
          r_{\mathrm{prb}}
        }        
      \Big)
      \,\pm\,
      \tfrac{1}{6}
      \,
      3a
      \underbrace{
        \epsilon^{i j k}
        \epsilon_{i j k}
      }_{
        \color{gray}
        -6
      }
      \Big).
    \end{array}
  \end{equation}
  This is equivalent to
  \vspace{1mm}
  $$
    F_{M N P Q}
    \;=\;
    + 3a \, \sqrt{-g} 
    \,
     \scalebox{1.3}{$\epsilon$}_{M N P Q}
    \,,
    \;\;\;\;\;\;\;\;\;\;
    a^2 r_{\mathrm{prb}}^3
    (1 + a^2 r_{\mathrm{prb}}^2)^{-1}
    +
    2 r_{\mathrm{prb}}
    -
    3 a r_{\mathrm{prb}}^2
    (1 + a^2 r_{\mathrm{prb}}^2)^{-1/2}
    \,=\,
    0
  $$  
  
  \vspace{1mm} 
  \noindent   
  as in \cite[(15)]{BDPS87} (cf. also \cite[(A.9)]{CKKTvP98}), nominally solved by 
  \begin{equation}
    \label{NominaSolutionToStaticChartEmbedding}
    r_{\mathrm{prb}} \in \{0, \infty\}\,.
  \end{equation}
\end{example}

\smallskip 
\begin{remark}[\bf Non-existence of static M2-embeddings in the static chart]
Beware that neither of the values \eqref{NominaSolutionToStaticChartEmbedding} can be used to construct an actual M2 super-embedding: 
\begin{itemize}[
  leftmargin=.65cm,
  topsep=2pt,
  itemsep=2pt
]
\item[\bf (i)] At the value $r_{\mathrm{prb}} = 0$ the map \eqref{StaticM2EmbeddingInStaticChart} is constant and hence is not an immersion (much less an embedding).
\item[\bf (ii)] The would-be value $r_{\mathrm{prb}} = \infty$ is outside the actual range of this variable 

(one might interpret it as only a mnemonic for taking the limit of observables as $r_{\mathrm{prb}} \to \infty$, but since the EOM is violated at all the finite $r_{\mathrm{prb}}$ whose limiting case would be computed thereby, it may be hard to interpret the result).
\end{itemize}
This problem suggests that for the purpose of microscopic $p$-brane holography, the appropriate holographic embeddings are those static with respect to the Poincar{\'e} chart, as used in \cite{CKKTvP98}\cite{PST99} and here in the main text. 
\end{remark}

\end{document}